%% file: bare_conf_NDSS2024.tex
\documentclass[conference]{IEEEtran}
\ifCLASSINFOpdf
\else
\fi

\usepackage{tikz}
\usepackage{amsmath}
\usepackage{filecontents}
\usepackage{framed}
\usepackage{color}
\usepackage{xcolor}
\definecolor{LimeGreen}{rgb}{0.2, 0.8, 0.2}
\usepackage[linesnumbered,ruled,vlined]{algorithm2e}
\usepackage{algorithmic}
\usepackage{array}
\usepackage{multirow}
\usepackage{booktabs,subcaption,amsfonts,dcolumn}
\usepackage{paralist}
\usepackage{balance}
\usepackage{todonotes}
\usepackage{hyperref}
\usepackage [ n,
advantage , operators , sets , adversary , landau , probability , notions , logic ,
ff, mm, primitives , events , complexity , asymptotics , keys]{ cryptocode}
\usepackage{bbding}
\usepackage{dashrule}
\usepackage{tikz}
\usepackage{amsmath}
\usepackage{pifont}
\usepackage{multicol}
\usepackage{amsthm}
\newtheorem{theorem}{Theorem}
\newtheorem{definition}{Definition}
\usepackage{filecontents}
\usepackage{tablefootnote}
\newtheorem{lemma}{Lemma}

\usepackage{bm}
\input{macros}

\usepackage{amsmath}
\usepackage{graphicx}
\newtheorem*{remark}{Remark}
\usepackage{fancyhdr}
\usepackage{threeparttable}
\usepackage{lipsum}
\usepackage{pdfpages}
\usepackage[T1]{fontenc}

\usepackage{authblk}

\begin{document}
%
\title{MUDGUARD: Taming Malicious Majorities in Federated Learning using Privacy-Preserving Byzantine-Robust Clustering}

\author[1]{Rui Wang}
\author[2]{Xingkai Wang}
\author[1]{Huanhuan Chen}
\author[1]{J\'{e}r\'{e}mie Decouchant}
\author[1,3]{Stjepan Picek}
\author[4]{Nikolaos Laoutaris}
\author[1]{\\Kaitai Liang}
\affil[1]{Delft University of Technology}
\affil[2]{Shanghai Jiao Tong University}
\affil[3]{Radboud University}
\affil[4]{IMDEA Networks Institute}
\maketitle

\begin{abstract}
Byzantine-robust Federated Learning (FL) aims to counter malicious clients and train an accurate global model while maintaining an extremely low attack success rate. 
Most existing systems, however, are only robust when most of the clients are honest.
\texttt{FLTrust} (NDSS '21) and \texttt{Zeno++} (ICML '20) do not make such an honest majority assumption but can only be applied to scenarios where the server is provided with an auxiliary dataset used to filter malicious updates. 
\texttt{FLAME} (USENIX '22) and \texttt{EIFFeL} (CCS '22) maintain the semi-honest majority assumption to guarantee robustness and the confidentiality of updates. 
It is therefore currently impossible to ensure Byzantine robustness and confidentiality of updates without assuming a semi-honest majority. To tackle this problem, we propose a novel Byzantine-robust and privacy-preserving FL system, called \texttt{MUDGUARD}, that can operate under malicious minority \emph{or majority} in both the server and client sides. 
Based on DBSCAN, we design a new method for extracting features from model updates via pairwise adjusted cosine similarity to boost the accuracy of the resulting clustering.
To thwart attacks from a malicious majority, we develop a method called \textit{Model Segmentation}, that aggregates together only the updates from within a cluster, sending the corresponding model only to the clients of the corresponding cluster. The fundamental idea is that even if malicious clients are in their majority, their poisoned updates cannot harm benign clients if they are confined only within the malicious cluster. We also leverage multiple cryptographic tools to conduct clustering without sacrificing training correctness and updates confidentiality. 
We present a detailed security proof and empirical evaluation along with a convergence analysis for \texttt{MUDGUARD}.
Our experimental results demonstrate that the accuracy of \texttt{MUDGUARD} is practically close to the FL baseline using FedAvg without attacks ($\approx$0.8\% gap on average). Meanwhile, the attack success rate is around 0\%-5\% even under an adaptive attack tailored to \texttt{MUDGUARD}.
We further optimize our design by using binary secret sharing and polynomial transformation leading to communication overhead and runtime decreases of 67\%-89.17\% and  66.05\%-68.75\%, respectively.

\end{abstract}
\input{sections/introduction}
\input{sections/relatedworks}
\input{sections/problem_formulation}
\input{sections/protocol_description}
\input{sections/experiment}
\input{sections/conclusion} 

\bibliographystyle{IEEEtranS}
\bibliography{bare_conf_NDSS2024}
\appendix
\input{sections/notation}
\input{sections/algorithm}
\input{sections/preliminaries}
\input{sections/complexity_analysis.tex}
\input{sections/security_analysis}
\input{sections/converge_analysis}
\input{sections/other_experiment_results}
\input{sections/other_discussions}

\end{document}

%% file: macros.tex
\newcommand{\weights}{\textit{\textbf{w}}}
\newcommand{\gradients}{\textit{\textbf{g}}}
\newcommand{\aggradient}{\textit{\textbf{G}}}
\newcommand{\signn}{\mathsf{sign}}
\newcommand{\ks}{\mathsf{KS}}
\newcommand{\cosm}{\mathsf{CosM}}
\newcommand{\eucm}{\mathsf{EucM}}
\newcommand{\indm}{\mathsf{IndM}}
\newcommand{\encoding}{\mathsf{ECD}}
\newcommand{\decoding}{\mathsf{DCD}}

%% file: sections/introduction.tex
\section{Introduction}
\label{sec:intro}
%
%
Thanks to its privacy properties, 
Federated Learning (FL)~\cite{mcmahan2017communication} has been widely applied in real-world applications, e.g., 
prediction of the future oxygen requirements of symptomatic patients with COVID-19~\cite{dayan2021federated}.  
Despite its attractive benefits, FL is vulnerable to Byzantine attacks. 
For example, attackers may choose to deteriorate the testing accuracy of models in an untargeted attack. Alternatively, they might fool models to predict an attack-chosen label without downgrading the testing accuracy in a targeted attack.
Many research works~\cite{fang2020local,Xie2020DBA,DBLP:conf/nips/WangSRVASLP20} have proved the vulnerability of FL via well-designed attack methods, e.g., poisoning training data or manipulating updates. Other studies~\cite{blanchard2017machine,pmlr-v80-yin18a, guerraoui2018hidden,cao2020fltrust,nguyen2022flame,xie2020zeno++,roy2022EIFFeL,baybfed} have been dedicated to strengthening FL assuming that a minority of the clients can be malicious and that the server is honest.
\begin{table*}[t]
\centering
\captionsetup[table*]{labelsep=newline,singlelinecheck=false}
\begin{threeparttable}
\scalebox{0.9}{
\begin{tabular}{cccccccc}
\toprule
\multirow{3}{*}{\begin{tabular}[c]{@{}c@{}}Aggregation\\ strategy\end{tabular}} & \multicolumn{2}{c}{Threat model}                                                                                                       & \multirow{3}{*}{\begin{tabular}[c]{@{}c@{}}Byzantine\\ robustness\end{tabular}} & \multirow{3}{*}{\begin{tabular}[c]{@{}c@{}}Updates\\ confidentiality\end{tabular}} & \multirow{3}{*}{\begin{tabular}[c]{@{}c@{}}No requirement for\\ an auxiliary dataset\end{tabular}} & \multirow{3}{*}{\begin{tabular}[c]{@{}c@{}}Computation\\ complexity\end{tabular}} & \multirow{3}{*}{\begin{tabular}[c]{@{}c@{}}Communication\\ complexity\end{tabular}} \\ \cmidrule(lr){2-3}

                                      & \begin{tabular}[c]{@{}c@{}}Malicious \\ server(s)\end{tabular} & \begin{tabular}[c]{@{}c@{}}Malicious \\ majority clients\end{tabular}   &                                        &                                                                                                     \\ \toprule  
                                      
\texttt{Zeno++}~\cite{xie2020zeno++}& \XSolidBrush             & \CheckmarkBold& \CheckmarkBold & \XSolidBrush                                      & \XSolidBrush   &$O(d)$ & $O(nd)$  \\ \hline
\texttt{FLTrust}~\cite{cao2020fltrust}                               & \XSolidBrush             & \CheckmarkBold& \CheckmarkBold & \XSolidBrush                                      & \XSolidBrush  & $O(n)$           & $O(nd)$                                                                                           \\ \hline
\texttt{FLAME}~\cite{nguyen2022flame}                               & \XSolidBrush        & \XSolidBrush  & \CheckmarkBold  & \CheckmarkBold                                      & \CheckmarkBold & $O(d(S^{2}+n^{2}))$ & $O(d^{2}+Sn^{2})$                                                                                                       \\ \hline
\texttt{EIFFeL}~\cite{roy2022EIFFeL} & \XSolidBrush \tnote{1}       & \XSolidBrush  & \CheckmarkBold  & \CheckmarkBold                                      & \CheckmarkBold   &  $O((n+d)n\log^{2}n\log\log n+md \min(n,m^{2}))$     & $O(n^{2}+md\min (n,m^{2}))$                                                                                               \\ \hline
\texttt{MUDGUARD} (Ours)                     & \CheckmarkBold & \CheckmarkBold& \CheckmarkBold & \CheckmarkBold                                      & \CheckmarkBold   & $O(d+n^{3})$ & $O(S(d+n^{2}))$                                                                                                    \\ \bottomrule
\end{tabular}}
\begin{tablenotes}
\footnotesize
    \item $d$ stands for the dimension of a model. $n, m$, and $S$ represent the number of clients, malicious clients, and servers, respectively. 
    \item[1] EIFFeL considers a malicious server to be one that infers privacy information from other parties, which is equivalent to a semi-honest server in our context.
\end{tablenotes}
\end{threeparttable}
\caption{Comparison of FL systems}
\label{tab:comparison}
\end{table*}
Beyond Byzantine attacks, FL could put clients at high risk of privacy breach~\cite{aono2017privacy,NEURIPS2020_c4ede56b} even if clients' datasets are maintained locally. 
Several studies~\cite{nguyen2022flame,roy2022EIFFeL,truex2019hybrid} have applied secure tools, e.g., Additive Homomorphic Encryption (AHE)~\cite{paillier1999public}, Differential Privacy (DP)~\cite{dwork2008differential,wu2020value}, and Secure Multiparty Computation (MPC), to protect clients' updates\footnote{Note AHE and MPC require onerous computation over ciphertexts so that the computational complexity could naturally increase.}. However, these works only guarantee security when all servers are (semi-)honest and when a minority of the clients are malicious.

{To the best of our knowledge, there does not exist any FL system that is capable of withstanding the presence of a majority of Byzantine clients, as well as malicious servers, while also guaranteeing the confidentiality of updates.}
One may think that existing Byzantine-robust solutions could be trivially extended to address the above challenge. 
However, that is not the case because they either violate privacy preservation requirements or are only effective in the honest majority scenario.
For example, \texttt{FLTrust}~\cite{cao2020fltrust} and \texttt{Zeno++}~\cite{xie2020zeno++} require an auxiliary dataset that is independently and identically distributed (iid) with the clients' training datasets to rectify malicious updates, which evidently violates the clients' privacy.
As for \texttt{FLAME}~\cite{nguyen2022flame}, it clusters updates and considers the smallest cluster as a malicious group, which makes sense in the malicious minority context. 
However, in the case of a malicious majority, it is difficult to assert if a given large/small-size cluster is malicious.
\texttt{EIFFeL}~\cite{roy2022EIFFeL} shows similar infeasibility,
since it combines existing Byzantine-robust methods (e.g., 
\texttt{FLTrust}~\cite{cao2020fltrust}) with secure aggregation~\cite{bonawitz2017practical}.

\noindent{\bf Contributions.} 
We propose a practical and secure Byzantine-robust FL system, \texttt{MUDGUARD}, that defends against malicious entities (i.e., malicious minority for servers and malicious majority for clients) with privacy preservation. 
Specifically, we perform feature extraction on the client updates by calculating the pairwise adjusted cosine similarity. 
These extracted features are taken to the DBSCAN clustering, which calculates the pairwise $L_2$ distance between the inputs and determines clusters based on density. 
Subsequently, a new method, \textit{Model Segmentation}, aggregates the updates based on their assigned cluster labels. The aggregated results are returned to clients in the respective clusters. 
Moreover, the integration of cryptographic tools with the aforementioned calculations allows for the establishment of both Byzantine robustness and privacy preservation within the \texttt{MUDGUARD} system.
In spite of the utilization of clustering for Byzantine robustness in~\cite{nguyen2022flame}, its approach relies on the assumption that only less than half of clients are malicious, since updates are directly incorporated into HDBSCAN clustering. In contrast to~\cite{nguyen2022flame}, our proposed methodology involves preliminary feature extraction of updates before sending them to the clustering algorithm and separate aggregation, accommodating non-iid and malicious majority clients. We stress that~\cite{nguyen2022flame} is vulnerable to the non-iid and majority of malicious clients problems due to the absence of pre-processing of updates and cluster detection. To the best of our knowledge, \texttt{MUDGUARD} is the first Byzantine-robust FL system that is able to defend against malicious majority clients without sacrificing clients' privacy.
We also are the first to provide detailed security and privacy analysis under Universal Composability (UC) in FL. In the literature on centralized learning, a few MPC-based solutions~\cite{lehmkuhl2021muse,koti2021swift,damgaard2019new} were proposed for model training and malicious servers under UC. Note that the key differences between these works and ours are the required computing functionalities for clustering and model training, as well as the required security properties of the solution. The impact of these works on performance and security is unknown when applied to FL.

We summarize the advantages of \texttt{MUDGUARD} on the SOTA FL systems in Table~\ref{tab:comparison}. For a
theoretical and empirical analysis of complexity, please refer to Appendix~\ref{sec:complexityana} and Section~\ref{sec:overheads}.
Our main contributions can be described as follows.\\
$\bullet$ We formulate a new aggregation strategy, \emph{Model Segmentation}, for Byzantine-robust FL to effectively avoid poisoning attacks from a majority of malicious clients without requiring the servers to own an auxiliary dataset. It posits that the utilization of complex algorithms for the detection of malicious updates is not necessary. Instead, it only suggests implementing measures to prevent the co-existence of malicious and semi-honest clients in one aggregation.
\\
$\bullet$ We propose a new method to improve the accuracy of updates clustering under non-iid scenarios.
Instead of using the updates directly for clustering, we first compute the pairwise adjusted cosine similarity of updates (featured by different directions and magnitudes of updates between every two clients). Then we input the results to DBSCAN. 
\\
$\bullet$ We design a secure FL system to be compatible with the cryptographic tools under the malicious context. 
To protect the updates on the server side and guarantee all clients receive correct aggregations, we construct a secure DBSCAN clustering that leverages cryptographic tools and secure aggregation with Homomorphic Hash Function (HHF)~\cite{kim2012device}.
We further optimize the secure {computations on the server side} based on binary secret sharing and polynomial transformation. 
\\
{$\bullet$ We provide a formal security proof for \texttt{MUDGUARD} in the UC framework. This proof captures dynamic security requirements, making \texttt{MUDGUARD} more practical than theoretical in security. \texttt{MUDGUARD} is the first UC-secure type in the research line of privacy-preserving FL.} 
\\
$\bullet$ We implement \texttt{MUDGUARD} and perform evaluations on (F)MNIST and CIFAR-10 to quantify its accuracy under untargeted attacks, the Attack Success Rate (ASR) under targeted attacks or under an adaptive attack tailored to \texttt{MUDGUARD}, as well as its runtime and communication costs. Our experimental results show that the model trained by \texttt{MUDGUARD} maintains comparable testing accuracy with the FL baseline - a ``no-attack-and-protection" FL with only honest parties ($\approx$0.8\% gap on average under untargeted attacks). The ASR under the targeted attacks is as low as 0\%-5\%.
After optimizing the cryptographic computations, the runtime and communication costs are reduced by about 66.05\%-68.75\% and 67\%-89.17\%, respectively. For example, in the training of ResNet-18 using CIFAR-10, our optimization strategy can reduce training time from 95 seconds to 48 seconds and communication costs from 16331\,MB to 5909\,MB, whereas a vanilla FL takes nearly 24 seconds and 758\,MB per round.

%% file: sections/relatedworks.tex
\section{Background and Related Work}
\label{relatedwork}

\subsection{Attacks against Federated Learning}

\noindent\textbf{Byzantine Attacks.}
\label{byzantineattack}
Malicious clients may attempt to deteriorate the testing accuracy of the global model by intentionally uploading poisoned updates (i.e., untargeted attacks). Instead of harming the accuracy, the attackers may also intentionally use samples with triggers to launch attacks that make the model misclassify (i.e., targeted attacks). 
In the following, we  review some classical and SOTA untargeted attacks (Gaussian Attack~\cite{fang2020local}, Label Flipping Attack~\cite{biggio2012poisoning}, Krum Attack and Trim Attack~\cite{fang2020local}) and targeted attacks (Backdoor Attack~\cite{bagdasaryan2020backdoor} and Edge-case Attack~\cite{10.5555/3495724.3497072}).\\
$\bullet$ \textbf{Gaussian Attack (GA).} Malicious clients degrade the model accuracy by uploading local updates randomly sampled from a Gaussian distribution. \\
$\bullet$ \textbf{Label Flipping Attack (LFA).} Malicious clients flip the local data labels to generate faulty gradients. In particular, the label of each sample is flipped from $y$ to $L-1-y, y\in [L]$, where $L$ is the total number of classes.\\
$\bullet$ \textbf{Krum and Trim Attacks.} These two untargeted local model poisoning attacks are optimized for \texttt{Krum}~\cite{blanchard2017machine} and the \texttt{Trim-mean/Median}~\cite{pmlr-v80-yin18a} aggregation strategies, respectively. They aim to pull the global model towards the opposite direction of the honest gradient when it is updated. Besides, they also have attack efficacy on \texttt{FedAvg}.\\
$\bullet$ \textbf{Backdoor Attack (BA).} Byzantine clients embed triggers to training samples and change their labels to targeted labels. Their goal is to make the global model misclassify the correct labels to the targeted ones when testing samples with triggers. \\
$\bullet$ \textbf{Edge-case Attack (EA).} 
The attack aims to misclassify seemingly similar inputs that are unlikely to be part of the training or testing data. 
For example, by labeling Ardis\footnote{A dataset extracted from 15,000 Swedish church records written by different priests with various handwriting styles in the nineteenth and twentieth centuries.}
“7” images as “1” and adding them to training data,  EA can easily backdoor an MNIST classifier. 
Similarly, the attack can use a Southwest airplanes dataset labeled “truck” to inject a backdoor into a CIFAR-10 classifier. 
Note that the attack relies on a restricted assumption that an extra dataset resemblance to the training dataset should be given. 

To defend against these attacks, we propose a new approach called \emph{Model Segmentation} in conjunction with feature-extracted DBSCAN. Different from other Byzantine-robust FL, our proposed method generates multiple global models and does not require servers to detect whether a particular group is benign or not. It aggregates only updates with the same cluster labels and returns the aggregations to the corresponding clients. 
This ensures that updates with similar directions and magnitudes are aggregated together (i.e., benign updates are not aggregated with malicious updates), providing a guarantee of Byzantine robustness in the case of malicious majority clients. Different from \texttt{FLAME}, \texttt{MUDGUARD} first uses pairwise adjusted cosine similarity to perform feature extraction on updates, then clusters through DBSCAN. The advantage of this is that it can reduce the false positive rate and be effective in non-iid situations. For detailed explanations, please refer to Section~\ref{sec:byagg} and Appendix~\ref{sec:advofcosm}. The experimental results show that the Byzantine robustness and clustering accuracy of \texttt{MUDGUARD} is better than that of \texttt{FLAME} as will be shown later in Section~\ref{sec:expacc}.

\noindent\textbf{Inference Attacks.}
Although local datasets are not directly revealed during the FL training process, the updates are still subject to privacy leakage if the server is semi-honest or even malicious~\cite{nasr2019comprehensive, luo2021feature, aono2017privacy}. 
For instance, Zhu \emph{et al.}~\cite{aono2017privacy} investigated a method of training data reconstruction via optimizing the $L_2$ distance between uploaded gradients and gradients trained from dummy samples using an L-BFGS solver.
This approach allows servers to easily reconstruct the local datasets and achieves even pixel-wise accuracy for images and token-wise matching accuracy for texts.
To defend against such attacks, we use Secret Sharing (SS)~\cite{shamir1979share} to split client updates into $s$ shares before sending them to servers. 
This way, updates are safeguarded from malicious servers since they do not get sufficient shares to perform update reconstruction.
Even if malicious servers collude with malicious clients, no extra benefit can be achieved towards compromising the update shares belonging to semi-honest clients.

\noindent\textbf{Differential Attack.}
We use DBSCAN in conjunction with \textit{Model Segmentation} to separate benign and malicious updates.
Features extraction with adjusted cosine similarity greatly descends the likelihood of false positives. 
However, we cannot guarantee that the clustering results are 100\% correct. 
A semi-honest client could be clustered together with other malicious clients by a small probability as it could have a smaller similarity from semi-honest clients than that from malicious clients. 
In particular, this case happens more frequently in SignSGD, where only taking signs of gradients to update the model because the algorithm computing adjusted cosine similarity disregards the magnitude of the gradients, resulting in the same effect as calculating cosine similarity.
The above phenomenon triggers the differential attack in the following cases.
Assuming that, at $t$-th round, a semi-honest client is misclustered to a malicious group.
After returning aggregation to the group, the benign updates can be easily revealed by subtracting those malicious updates, and the inference attack can be launched further. 
Another more common case is that a malicious adversary $\adv$ compromises $m$ clients and then makes one of them perform correct operations, i.e., acting as a semi-honest client. 
This malicious-but-act-semi-honest client, being assigned to a semi-honest group, can get benign aggregation from each round and then conduct an inference attack.
In this work, we apply DP to make aggregations obtained by malicious clients statistically indistinguishable from those containing benign updates, guaranteeing that benign updates cannot be easily identified from aggregations.

\subsection{Defenses}
\noindent\textbf{Byzantine-robust Federated Learning.}
Blanchard \emph{et al.}~\cite{blanchard2017machine} proposed \texttt{Krum} to select 1 out of $n$ (local updates) as a global update for each round, where the selected updates should have the smallest $L_2$ distance from others.
Yin \emph{et al.}~\cite{pmlr-v80-yin18a} introduced \texttt{Trim-mean} and \texttt{Median} to resist Byzantine attacks. 
The former uses a coordinate-wise aggregation strategy. 
The server calculates $n-2z$ values for each model parameter as the global model update, wherein the largest and smallest $z$ values are filtered.
Unlike \texttt{FedAvg}~\cite{mcmahan2017communication} computing the weighted average of all parameters, the latter calculates the median of parameters. 
This median serves as an update to the global model.
A major drawback of the aforementioned mechanism is that it is effective only under a majority of honest clients working with an honest server. 
In \texttt{Median}~\cite{pmlr-v80-yin18a}, the median calculated by the server can easily be malicious if malicious clients control a large/overwhelming proportion of updates.
This similarly applies to \texttt{Trim-mean} and \texttt{Krum}.
Cao \emph{et al.}~\cite{cao2020fltrust} proposed \texttt{FLTrust} to protect against a malicious majority at the client side, assuming an honest server holds a small auxiliary dataset. 
The server treats the gradients trained from this small dataset as the root of trust.
By comparing these trusted results with the updates sent by clients, the server can easily rule out malicious updates. 
Under the same assumption, \texttt{Zeno++}~\cite{xie2020zeno++} uses an auxiliary dataset to calculate the loss value of each local model. A client is determined to be honest if the loss value is beyond the preset threshold.
While using an auxiliary dataset could be intriguing, such approaches are not feasible in the context of FL as they violate the fundamental premise of FL in which local datasets are not to be shared with any parties.

\noindent\textbf{Privacy-preserving Federated Learning.}
Truex \emph{et al.}~\cite{truex2019hybrid} proposed a solution enabling clients to use AHE and DP to secure gradients in the semi-honest context (for both clients and the server). 
Since DP noise is applied on gradients, the accuracy of the global model is deteriorated.
In the scenario of honest majority clients with two semi-honest servers, Thien \emph{et al.}~\cite{nguyen2022flame} proposed \texttt{FLAME} using an MPC protocol to protect gradients from the servers and enabling the servers to perform clustering for Byzantine robustness.
Specifically, the clients can securely share their updates to the servers cryptographically, e.g., via secret sharing, and the servers can filter out malicious updates without knowing their concrete values. 
By expressing existing Byzantine-robust solutions (e.g., \texttt{FLTrust}) as arithmetic circuits, \texttt{EIFFeL}~\cite{roy2022EIFFeL} enables secure aggregation of verified updates.
Although \texttt{FLAME} and \texttt{EIFFeL} capture both Byzantine robustness and privacy preservation (i.e., update confidentiality), the accuracy of the global model could become equivalent to a random guess if the proportion of malicious clients is $\ge$50\%.

%% file: sections/problem_formulation.tex
\section{Problem Formulation}
\label{sec:problem}
\subsection{System Model}
Before proceeding, we provide some assumptions about \texttt{MUDGUARD}.
We assume training is conducted on a dataset $\mathcal{D}$ with $K$ data samples composed with feature space $\mathcal{X}$ (each sample containing all features) and a label set $\mathcal{Y}$. 
Additionally, $\mathcal{D}$ is horizontally partitioned among $n$ clients, indicated as
$
    \mathcal{X}_i=\mathcal{X}_j, \mathcal{Y}_i = \mathcal{Y}_j, \mathcal{I}_i\cap \mathcal{I}_j=\emptyset, \forall \mathcal{D}_i,\mathcal{D}_j, i\neq j,
    \notag
$
where all clients share the same feature space and labels but differ in sample index space $\mathcal{I}$.
FL aims to optimize a loss function:
\begin{math}
    \mathop {\arg\min }\limits_\weights\sum\limits_{i = 1}^n {\frac{{{k_i}}}{K}{\mathcal{L}_i}(} \weights,\mathcal{D}_i ),
    \notag
\end{math}
where ${\mathcal{L}_i}(\cdot)$ and $k_i$ are the loss function and local data size of $i$-th client. 

For reasons that relate to the versatility of the FL system, we also consider $S\ (>2)$ servers to carry out clustering and aggregation (e.g., \texttt{FedAvg}). This allows us to protect from malicious servers  who cannot reconstruct the secrets so long as their number is less than $S/2$ by using cryptographic tools, in which clients send updates in secret-shared format. 
We state that our Byzantine solution can also be executed by only one server. 
In this case, considering privacy, we have to assume that the server must be fully trusted or semi-honest.
Note that our focus here is on the existence of malicious servers. 
In this research line~\cite{7958569,lehmkuhl2021muse}, secure computation is considered among multiple servers.
Due to page limit, we summarize frequently used notations in Table~\ref{tab:notations} (see Appendix~\ref{sec:notation}).

\subsection{Threat Model}
\label{threatmodel}
We mainly consider potential threats incurred by participating clients, servers, and outside adversaries. \\
$\bullet$ \textbf{Attackers' goal.} We assume that two different entities are involved in the training: semi-honest and dynamic malicious parties (including servers and clients), in which both try to infer the privacy (updates) information of others from the received messages. 
Unlike the former, strictly following the designed algorithms, the malicious clients additionally
aim to deteriorate the performance or boost the ASR of the global model through untargeted or targeted poisoning attacks, respectively.\\
$\bullet$ \textbf{Attackers' capabilities.} 
The malicious servers (in a minority proportion) and clients (in a majority proportion) can deviate from the designed protocols. For example, the malicious servers can perform an incorrect aggregation and send it back to the semi-honest group. Moreover, malicious parties (servers and clients) can collude with each other to infer benign aggregations and maximize the efficacy of poisoning attacks (e.g., the Krum attack).
To resist outside adversaries, secret-shared messages are transmitted by private communication channels. 
Other messages are transmitted through public communication channels, where outsiders are allowed to eavesdrop on these channels and try to infer clients' (updates) privacy during the whole training phase.\\
$\bullet$ \textbf{Attackers' knowledge.} We assume that the loss function, data distributions, Byzantine-robust aggregation strategy, and public parameters (including training and security parameters) are revealed to all parties. The malicious clients can exploit this information to design and cast adaptive attacks tailored to \texttt{MUDGUARD}. For privacy reasons, the local updates and datasets of semi-honest clients are not revealed to malicious parties.

%% file: sections/protocol_description.tex
\section{MUDGUARD Overview and Design}
\label{sec:MUDGUARD}

\subsection{Overview}
In a traditional FL system, clients send updates to the servers for global model aggregation. 
Considering there exist malicious clients, we should maintain the Byzantine robustness such that malicious updates should be excluded properly.  
To do so, the servers must separate the malicious clients from the semi-honest clients. 
DBSCAN helps the servers to perform clustering. 
Since the main difference between the malicious and the benign is in the direction and magnitude of the updates, we use the adjusted cosine similarity of updates as feature extraction to obtain better clustering accuracy.
Under the (semi-)honest majority, the clustering result directly links to the group size. 
However, for a dynamic malicious majority, we cannot judge if a cluster is malicious only based on its size. 
To address this issue, we propose {\textit{Model Segmentation}}. 
Unlike traditional FL generating ``a unique" global model, our proposed algorithm can yield multiple aggregation results. 
It does not require the servers to know whether a given group is  malicious or not.
Moreover, it only aggregates the updates within the same cluster and then returns the results to the corresponding clients. 
We thus guarantee that the semi-honest will not be aggregated with the malicious. 

As far as fighting against inference attacks is concerned, we should protect the confidentiality of the updates.
For this, we use SS to wrap the updates into a secret shared format in the sense that individual secret shares cannot reveal the underlying information of the updates. 
By doing so, we guarantee that the updates are secured from eavesdroppers, semi-honest, or even malicious servers and further can be used on secure multiplication, comparison, and aggregation via cryptographic tools. 
{However, using SS alone is not sufficient to defend against differential attacks.} 
To thwart the attack, we apply DP to prevent the attackers from extracting benign updates from the semi-honest group.
Since injecting noise brings a negative influence on the accuracy of the training model, we enable clients to perform denoising before wrapping the results into shares. 
Note that this does not invalidate DP due to the post-processing nature~\cite{dwork2008differential}.
%
%
We also consider the malicious minority servers and thus leverage HHF to prevent malicious servers from performing incorrect aggregation, e.g., merging the gradients from two different groups. Due to the page limit, we review machine learning and security tools in Appendix~\ref{sec:background}.

\subsection{Byzantine-robust Aggregation Strategy with Cryptographic Computations}
\label{sec:byagg}
Our workflow of the Byzantine-robust aggregation strategy is as follows. 
Firstly, the clients upload the gradients of the local models to the server side. 
Secondly, servers extract features of gradients and split gradients into multiple clusters via DBSCAN. 
Finally, servers aggregate the gradients in the clusters separately and send aggregations to the corresponding clients.
In the following, we complete the strategy over secure cryptographic computations.

\noindent\textbf{Gradients Upload.} 
The use of the pairwise adjusted cosine similarity matrix ($\cosm$) as a method for extracting features is motivated by the fact that it measures both the difference in directions and magnitudes of updates. This is particularly useful when dealing with clients exhibiting various behaviors and non-iid cases. In this context, $\cosm$ and $L_2$ distance are used as input and the metric of DBSCAN, respectively.
The most direct method of computing $\cosm$ is as follows.
We first subtract updates with their mean values.
For the updates of each client, we compute the pairwise dot product and $L_2$ norm to derive the numerator and denominator, respectively, and
then we can calculate $\cosm$ from the division (of numerator and denominator). 
The above operations become inefficient if the processing is carried out using cryptographic tools. 
The servers are required to perform the computations of shared mean, numerator, denomination, and then division to finally get the shared adjusted cosine similarity matrix $[\![\cosm]\!]$.

\begin{figure}[t]
    \centering    \includegraphics[width=0.4\textwidth]{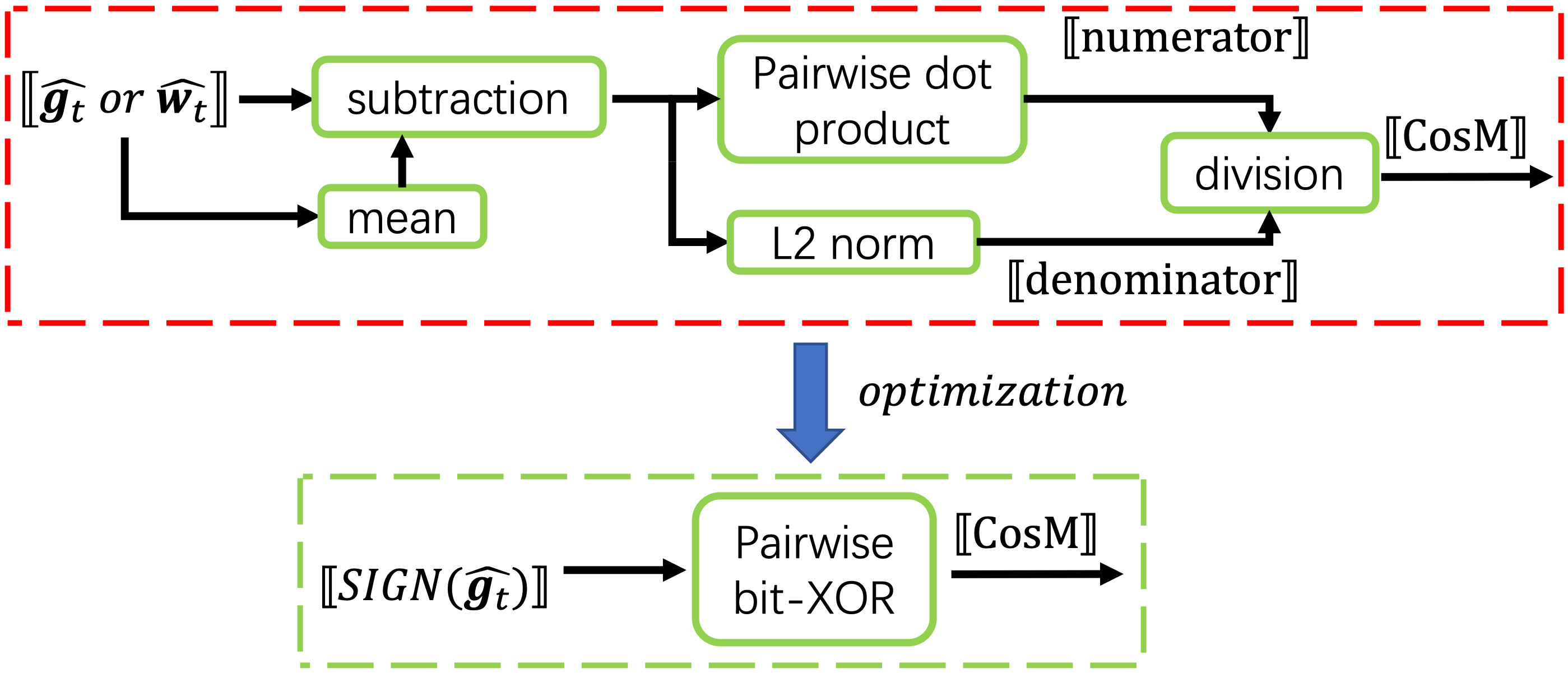}
    \caption{Optimization on the calculation of adjusted cosine similarity. \textcolor{LimeGreen}{\hdashrule[0.5ex]{1cm}{1.3pt}{1mm}}: optimized process: by binary SS, $[\![\cosm]\!]$ is computed via bit-XOR. \textcolor{red}{\hdashrule[0.5ex]{1cm}{1.1pt}{1mm}}: unoptimized process: taking gradients or weights as updates to compute $[\![\cosm]\!]$.}
    \label{fig:optimization}
\end{figure}

To improve efficiency and optimize the above method, we consider the denoised gradients of client $i$ at $t$-th round $\hat{\gradients}_t^i$ as updates and perform binary secret sharing via SignSGD. 
Note that SignSGD only takes the signs of gradients to the update model, resulting in benign and malicious having the same magnitudes. 
Thus, in this case, 
we can easily compute the adjusted cosine similarity via simple bit-wise XORing. 
%
Figure \ref{fig:optimization} depicts this optimization procedure.

Therefore, in each training round (of the optimization), client $i$ derives the gradients using SGD~\cite{bottou2010large}.
Considering the upcoming cryptographic clustering, one needs to compute the signs of gradients  $\signn(\hat{\gradients}_t^i) \in\{-1,+1\}$ as SignSGD~\cite{bernstein2018signsgd} and then encodes to Boolean representation, which is compatible with binary SS and XOR operations. 
Without loss of generality, we implement a widely used encoding/decoding method as
\begin{equation}
    \encoding(\signn(\hat{\gradients}_t^i))=
    \begin{cases}
    1,& \signn(\hat{\gradients}_t^i)=+1\\
    0,& otherwise
    \end{cases}\notag,
\end{equation}
\begin{equation}
  \decoding(\encoding(\signn(\hat{\gradients}_t^i)))=2\encoding(\signn(\hat{\gradients}_t^i))-1.
  \notag
\end{equation}

This method guarantees $\decoding(\encoding(\signn(\hat{\gradients}_t^i)))$ $=\signn(\hat{\gradients}_t^i)$. 
Each client sends the encoded updates to the servers via binary SS and broadcasts the hash results of unencoded updates 
for future verification.
Although SignSGD is lightweight, it brings a negative impact on the accuracy of clustering. 
Section~\ref{sec:expacc} provides a detailed analysis of this impact.
Note that SignSGD has the natural capability of defending against scaling attacks~\cite{bagdasaryan2020backdoor} since it only takes signs as updates and clips the magnitude of gradients. 
An attacker still can easily deteriorate the global model by constructing updates in the opposite direction of benign updates.

\noindent\textbf{Clustering.} As a crucial variable in FL, updates determine the directions and magnitudes of updating in the model, while Byzantine attackers introduce abnormal updates. 
{Traditional clustering approaches directly use updates as inputs, and cosine similarity as metric~\cite{nguyen2022flame}, causing informative redundancy and blurring obvious features, especially in deep models (e.g., ResNet~\cite{He_2016_CVPR}), thereby producing frequent false positives and negatives. 
Since the adjusted cosine similarity measures the difference in directions and magnitudes of updates at the same time, 
we use the pairwise adjusted cosine similarity $\cosm$ as a method for extracting features, i.e., $\cosm$ and $L_2$ distance, used as the input and the main metric of DBSCAN, respectively.
We find that this method is effective in distinguishing the updates because 1) calculating $\cosm$ (feature extraction) is equivalent to reducing the informative redundancy of updates to improve the clustering accuracy. 2) Using it to calculate the pairwise $L_2$ distance can expand the pairwise differences which are not clearly computed by $\cosm$. Thus there is a more clear density difference between honest and malicious updates.
3) 
By subtracting the mean updates, $\cosm$ helps to account for reducing the influence of non-iid, allowing for a more accurate comparison of clients' updates.
4) When the model converges, using cosine similarity is inappropriate because even semi-honest clients have updates in different directions. 
For example, under a Gaussian Attack (GA), the malicious and semi-honest clients become indistinguishable. The accuracy of the global model drops sharply to the level of the initial training. 
We provide concrete examples in Appendix~\ref{sec:advofcosm} to demonstrate the advantages of using $\cosm$.
Note that this advantage is more notable when the cryptographic tools are not optimized.}
Since the proposed optimization uses SignSGD to align the magnitudes of updates, computing cosine similarity on it naturally provides the same effect on clustering as adjusted cosine similarity.

Next, we describe the process of our clustering.
We first extract features (different updates directions with magnitudes) - calculating the $\cosm\leftarrow \overline{\gradients}_t^i\xor \overline{\gradients}_t^j, i,j\in[n]$, 
and then use it as the input for clustering, thereby reducing the rate of false positives.
Commonly, the adjusted cosine similarity of two vectors is obtained by first calculating the dot product of the vectors and then dividing them by the product of their respective $L_2$ norm.
Since encoding updates $\hat{\gradients}_t^i\in\{-1, +1\}^{d}$ to $\overline{\gradients}_t^i\in\{0, +1\}^{d}$ is inspired by~\cite{riazi2019xonn}, we compute XOR of $\overline{\gradients}_t^i$ with $d-2p$ bits, which is equivalent to the result of the dot product of $\hat{\gradients}_t^i$, where $p$ is the counted number of set bits.

After that, the servers collaboratively compute pairwise $L_2$ distance matrix $\eucm$ using secure multiplications ($vector_{ij}\leftarrow\cosm_i-\cosm_j, i,j\in [n]$, $x_{ij}\leftarrow vector_{ij}\cdot vector_{ij}$, and $\eucm_{ij}\leftarrow1+\frac{x_{ij}-1}{2}-\frac{(x_{ij}-1)^2}{8}+\frac{(x_{ij}-1)^3}{16}$) and compare with density $\alpha$ to derive an indication matrix $\indm$, where $\indm = 1$ if $\eucm \leq \alpha$, otherwise $\indm = 0$.
Then by applying the DBSCAN, one can derive cluster labels.
Note that the main focus of this paper is not on optimizing DBSCAN, {we thus do not describe how to retrieve cluster labels from $\indm$}.
We refer interested readers to~\cite{ester1996density}. 

We see that $\alpha$ has a crucial influence on clustering accuracy. According to the conclusion of Bhagoji \emph{et al.}~\cite{bhagoji2019analyzing}, benign and malicious updates follow normal distributions. We formally derive its upper bound of selection (see Theorem~\ref{the:1} and its proof in Appendix~\ref{sec:pot1}).

\begin{theorem}[Density Selection]
\label{the:1}
Suppose the distribution of benign and malicious updates obeys the normal distribution, setting $\alpha < \sqrt{2}$ guarantees that malicious clients conducting a poisoning attack will not be grouped together with benign clients.
\end{theorem}

Note that taking $\alpha<\sqrt{2}$ only allows malicious clients to be identified as noise points, which is not a 100\% guarantee that all semi-honest clients are clustered together. 
Due to the difference in training data, it could happen that the distances between a semi-honest client and other semi-honest clients are greater than $\sqrt{2}$ by chance, resulting in the semi-honest client being identified as a noise point.

\noindent\textbf{Model Segmentation.} To deal with  Byzantine-majority attacks, after obtaining the cluster labels, the servers aggregate the updates within the same cluster and return the results (and their hash values) to the corresponding clients.
In our design, unless a malicious client acts honestly, then it will not be grouped into a cluster with the semi-honest clients, with a relatively large probability. This protects benign clients by keeping poisonous updates from global model updates computed for benign clusters.  
Note here we do not further explore the case in which malicious clients choose to act honestly during training. In fact, if  malicious clients behave semi-honestly, we will obtain a more accurate global model. 
In a sense, this is a bonus for semi-honest clients.  
After all, in the context of \textit{Model Segmentation}, it is not required to identify malicious groups via any verification algorithms, which is a positive thing since it removes the processing burden from the servers as the latter do not need to run verification over ``encrypted-and-noised" updates. 
Compared with the method of \texttt{FLAME}, \emph{Model Segmentation} does not need to assume that most clients in the FL system are (semi-)honest. When integrated with an optimized clustering method, \emph{Model Segmentation} can also enhance the Byzantine robustness of the \texttt{MUDGUARD}.

\noindent\textbf{Resistance against Malicious Servers.} To further prevent malicious servers from casting and sending incorrect aggregation, we use HHF so that every client can verify if the received aggregation is correct. 
Specifically, the proposed method involves a pre-upload step in which client $i$ broadcasts hash values of signs of gradients $\hash_{\delta, \phi}(\signn(\hat{\gradients}t^i))$ to the remaining parties before uploading secret-shared updates to the server side. The servers use additive homomorphism of HHF to calculate hash values of aggregations $\prod_{i\in c_j}\hash_{\delta, \phi}(\signn(\hat{\gradients}t^i))=\hash_{\delta, \phi}(\aggradient_{t}^j)$ based on the clustering results, where $c_j$ refers to a cluster $j$ containing client indexes. After receiving the aggregations $\aggradient_{t}^j$, the clients can calculate $\hash_{\delta, \phi}(\aggradient_{t}^j)$ and $\prod_{i\in c_j}\hash_{\delta, \phi}(\signn(\hat{\gradients}_t^i))$ based on the cluster labels and the received hash values, and subsequently verify whether these two values are equal or not. 
Note that this method considers the possibility of malicious servers that may send incorrect aggregations and $\indm$ to the semi-honest clients. However, since only a minority of servers are assumed to be malicious, the semi-honest clients take the most consistent results as the real results.
\subsection{System Design}
Assume client $i \in [n]$ holds a horizontally partitioned dataset $\mathcal{D}_i$ satisfying $\mathcal{D}=\bigcup\limits_{i=1}^n \mathcal{D}_{i}$, at $t$-th round, \texttt{MUDGUARD} works as follows.  
\begin{framed}
\label{protocolMUDGUARD}
\begin{center}
	\textbf{Protocol MUDGUARD}
\end{center}
\ding{202} \emph{Local Training.} For each local minibatch, each client conducts SGD and takes gradients $\gradients_t^i$ as updates.\\
\ding{203} \emph{Noise Injection.} Each client adds noise into $\gradients_t^i$ to satisfy DP: $\widetilde{\gradients}_t^i\leftarrow \gradients_t^i/\max(1, |\!|\gradients_t^i|\!|_2/\Delta) + \mathcal{N}(0, \Delta^2\sigma^2)$. \\
\ding{204} \emph{Denoising.} To improve accuracy, each client denoises  $\widetilde{\gradients}_t^i$ by $\hat{\gradients}_t^i\leftarrow \ks(\widetilde{\gradients}_t^i, \mathcal{N})\cdot \widetilde{\gradients}_t^i$, where $\ks(\cdot)$ is the KS distance. \\
\ding{205} \emph{SS.} Each client splits $\overline{\gradients}_t^i\leftarrow \encoding(\signn(\hat{\gradients}_t^i))$ into $S$ shares by binary SS with Tiny Oblivious Transfer (OT) and sends the shares to $S$ servers: $[\![\overline{\gradients}_t^i]\!]\stackrel{SS}{\longleftarrow} \overline{\gradients}_t^i$. Besides, by running HHF, all clients broadcast $\hash_{\delta, \phi}(\signn(\hat{\gradients}_t^i))$. \\
\ding{206} \emph{Feature Extraction.} After receiving $n$ shares, each server locally computes a pairwise adjusted cosine similarity matrix by bit-XOR: $[\![\cosm_{ij}]\!]\leftarrow [\![\overline{\gradients}_t^i]\!]\xor [\![\overline{\gradients}_t^j]\!]$, $i, j\in [n]$. To further compute $L_2$ distance, all servers convert Boolean shares to arithmetic shares by correlated randomness. \\
\ding{207} \emph{$L_2$ Distance Computation.} After conversion, deriving multiplicative SS, each server  uses HE or OT to produce a triple, satisfying further multiplications. 
Therefore, each server takes $[\![\cosm]\!]$ as the inputs of DBSCAN and then computes $[\![\eucm]\!]$ by (a) pairwise subtraction: $[\![vector_{ij}]\!]\leftarrow[\![\cosm_i]\!]-[\![\cosm_j]\!], i,j\in [n]$, (b) dot product: $[\![x_{ij}]\!]\leftarrow[\![vector_{ij}]\!]\cdot [\![vector_{ij}]\!]$, and (c) approximated square root: $[\![\eucm_{ij}]\!]\leftarrow1+\frac{[\![x_{ij}]\!]-1}{2}-\frac{([\![x_{ij}]\!]-1)^2}{8}+\frac{([\![x_{ij}]\!]-1)^3}{16}.$ \\
\ding{208} \emph{Element-wise Comparison.} By comparing each element of $\eucm$ with density parameter $\alpha$, each server can derive shares of indicator matrix $[\![\indm]\!]$, $\{\indm_{ij}=1\mid \eucm_{ij} \leq \alpha  \}.$ \\
\ding{209} \emph{Reconstruction.} All servers run a reconstruction algorithm to reveal $\indm$: $\indm\stackrel{\mathsf{recon}}{\longleftarrow}[\![\indm]\!]$ and broadcast it to the client side. 
By DBSCAN, one can derive cluster labels. Based on these labels, the clients learn about clustering information to perform aggregation verification in step \ding{211}.\\
\ding{210} \emph{Model Segmentation.} The servers aggregate shares (based on the number of labels $c$) with the same labels after decoding: $\{[\![\aggradient_{t}^j]\!]\leftarrow\sum_{i\in c_j}\decoding([\![\overline{\gradients}_{t}^i]\!])\mid c_j=\{i\mid i\in [n]\}, j\in [c],\}$ and send to the corresponding clients. \\
\ding{211} \emph{Aggregation Verification.} After reconstructing aggregation, according to cluster labels, each client verifies aggregation by $\prod_{i\in c_j}\hash_{\delta, \phi}(\signn(\hat{\gradients}_t^i))\stackrel{?}{=}\hash_{\delta, \phi}(\aggradient_{t}^j).$ {If the equation holds, clients accept the aggregation results; otherwise, reject and abort.}
\end{framed}

We note that the corresponding implementation-level algorithms of \texttt{MUDGUARD} are given in Appendix~\ref{sec:IA} and will be used in the experiments.
\subsection{Privacy Preservation Guarantee}

\noindent\textbf{Differential attack resistance.} As shown in step \ding{203} and \ding{204} of Protocol~\ref{protocolMUDGUARD},
each client $i$ can add differentially private noise into gradients and perform denoising later. 
Like~\cite{nasr2020improving}, we use KS distance (of noised gradients and noise distribution) as a metric to denoise by multiplying noised gradients.
Differentially private updates are first denoised, taken signs, and encoded before being secretly shared. 

\noindent\textbf{Binary SS.} Unlike arithmetic SS in domain $\mathbb{Z}_{2^{b}}$, binary SS works with $b=1$, where $b$ is the bit length. 
To resist malicious clients deviating from SS specifications, we apply OT in our design (step \ding{205} of Protocol~\ref{protocolMUDGUARD}). 
However, this brings a considerable increase in communication costs. Furukawa \emph{et al.}~\cite{furukawa2017high} used TinyOT to generalize multi-party shares with communication complexity linear in the security parameter.
We follow this method 
so that each client $i$ binary shares its updates to $S$ servers. 
The SS scheme guarantees that a malicious server cannot reconstruct the secret even if colluding with the rest of the servers under a malicious minority setting. 

\noindent\textbf{XOR.} In step \ding{206} of Protocol~\ref{protocolMUDGUARD}, after receiving shares, each server can compute the pairwise dot product independently. 
Assume a server $s$ has $[\![\overline{\gradients}_t^i]\!]_s$, where $s\in [S]$. 
Since $\overline{\gradients}_t^i=[\![\overline{\gradients}_t^i]\!]_1\xor\cdots\xor [\![\overline{\gradients}_t^i]\!]_S$, we have $\overline{\gradients}_t^i\xor\overline{\gradients}_t^j=[\![\overline{\gradients}_t^i]\!]_1\xor[\![\overline{\gradients}_t^j]\!]_1\cdots\xor [\![\overline{\gradients}_t^i]\!]_S\xor[\![\overline{\gradients}_t^j]\!]_S, \forall i,j \in n$. 
Therefore, in this case, each server $s$ can compute $[\![dot\_product]\!]$ by $\{[\![dot\_product_{ij}]\!]_s=[\![\overline{\gradients}_t^i]\!]_s\xor[\![\overline{\gradients}_t^j]\!]_s\mid \forall i,j \in [n]\}$ locally and without interactions with other servers.
By multiplying a constant, one can derive shares of adjusted cosine similarity. Using binary SS can help us to save element multiplication and division operations.

\noindent\textbf{Bit to Arithmetic Conversion.} The servers also need to convert the shares in $\mathbb{Z}_{2}$ to arithmetic shares ($\mathbb{Z}_{2^{b}}$) to support the subsequent linear operations and multiplications. 
{We implement the conversion by following~\cite{cryptoeprint:2019:207}}. 
A common method is to use correlated randomness in these two domains (doubly-authenticated bits) and extend them. 
After this, the servers can derive arithmetic shares of the dot product.
Note some works~\cite{10.1145/3243734.3243854, mohassel2018aby3} leverage straightforward transformation under the cases with only semi-honest parties. %

\noindent\textbf{Multiplication.} As shown in step \ding{207} of Protocol~\ref{protocolMUDGUARD}, multiplications are necessary in DBSCAN. 
If we consider the semi-honest majority setting on the server side, the replicated SS and SSS can be applied here since both satisfy the multiplicative property, in which two shares multiplications can be computed locally without any interaction. 
For the existence of malicious servers, we consider the protocol proposed by Lindell \emph{et al.}~\cite{lindell2017framework}, modifying SPDZ~\cite{10.1007/978-3-642-32009-5_38} to the setting of multiplicative secret sharing modulo a prime (including replicated SS and SSS). 
Furukawa \emph{et al.}~\cite{furukawa2017high} also proposed a similar variant for TinyOT. 
Both are based on the observation that the optimistic triple production using HE or OT can be replaced by producing a triple using multiplicative secret sharing instead.

\noindent\textbf{Secure Comparison with Density $\alpha$.} With arithmetic shares, the comparison (step \ding{208} of Protocol~\ref{protocolMUDGUARD}) requires extra correlated randomness, especially secret random bits in the larger domains.
For the semi-honest majority servers, we follow the protocol~\cite{cryptoeprint:2020:1330} with $\mathbb{Z}_{2}$ to implement comparison efficiently.
Under the malicious minority, we should check if the output is actually a bit.
We follow~\cite{damgaard2013practical} to multiply a secret random bit with comparison output and then reconstruct it. If the reconstructed value is a bit, it proves that the malicious servers do not deviate from the comparison protocol. 
\subsection{Security Analysis}
\texttt{MUDGUARD} achieves security properties under \textit{malicious majority} clients and \textit{malicious minority} servers. 
Malicious parties may arbitrarily deviate from the protocol, while the rest of the parties are semi-honest, trying to infer information as much as possible (but following the protocol). 
We assume malicious clients and servers may collude with each other.  

A secure FL system satisfies correctness, privacy, and soundness. The latter two are security requirements. 
Informally, the requirements are: 
(1) the adversary learns nothing but the differentially private output; 
(2) the adversary cannot provide an invalid result accepted by a benign client.  
We first define the security in the UC framework~\cite{canetti2001universally}.
This allows the system to remain secure and capable to be arbitrarily combined with other UC secure instances. 
In such a framework, security is defined by a well-designed ideal functionality that captures several properties simultaneously, including correctness, privacy, and soundness. 
Specifically, Figure~\ref{fMUDGUARD} (Appendix~\ref{sec:analysis}) shows our ideal functionality $ \mathcal{F}_\textsf{MUDGUARD}$. 
The definition captures all required security properties except DP and soundness against malicious clients. Appendix~\ref{sec:analysis} will discuss the remaining. 

We prove the security in a $ \mathcal{F} $-hybrid model. 
Our proof adopts three existing ideal functionalities: $ \mathcal{F}_\textsf{RO} $, $ \mathcal{F}_\textsf{SS} $ and $ \mathcal{F}_\textsf{B2A} $. 
The first is for the random oracle model, and the latter two are the ideal functionalities of secret sharing \cite{furukawa2017high} and bit-to-arithmetic conversion \cite{cryptoeprint:2019:207} respectively. 
We have the following theorem: 

\begin{theorem}
	\texttt{MUDGUARD} securely realizes $ \mathcal{F}_\textsf{\texttt{MUDGUARD}} $ in the ($ \mathcal{F}_\textsf{RO} $, $ \mathcal{F}_\textsf{SS} $, $ \mathcal{F}_\textsf{B2A} $)-hybrid model, against malicious-majority clients and malicious-minority servers, considering arbitrary collusions between malicious parties. 
\end{theorem}

The remaining two properties are related to data output, which is not concerned with the cryptographic view. 
Specifically, DP is provided by adding noise (Appendix~\ref{sec:analysis}), and soundness against malicious clients is provided by \textit{Model Segmentation}. 

Note our well-designed functionality captures as many attacks as possible. 
In other words, soundness against malicious clients and DP cannot be achieved under the UC model. 
On the one hand, recognizing malicious clients is quite a \textit{subjective} task since they do not deviate from the protocol in cryptographic ways. 
There might be a benign client providing similar inputs that seem to be malicious, with a non-negligible possibility. 
On the other hand, the output with DP can be obtained by the adversary in our definition. 
Hence differential attacks should not be captured in the functionality. 
\subsection{Adaptive attack}
Recall that in Section~\ref{threatmodel}, a Byzantine-robust aggregation strategy is available to attackers. Malicious clients can adapt their attacks to nullify the robustness of the system. Note that untargeted attacks (e.g., Krum and Trim attacks) solve an optimization problem to maximize the efficacy of attacks, meaning the strategies of untargeted attacks are already optimal. Therefore, we design and evaluate an adaptive backdoor attack for \texttt{MUDGUARD}. Specifically, the attack is formulated by adding a sub-task to the attack optimization problem. Given the fact that \texttt{MUDGUARD} achieves Byzantine-robustness by aggregating only benign updates as much as possible based on adjusted cosine distance, the sub-task of this attack is to try to minimize the adjusted cosine distance of malicious updates from that of benign updates. Formally, a malicious client $i$ first derives benign and malicious updates ($\weights_t^{i}$ and $\weights_t^{i'}$) with owned unpoisoned and poisoned data ($\mathcal{D}_{i}$ and $\mathcal{D}_{i}^{'}$ ), respectively, at the $t$-th round:
\begin{equation}
    \weights_t^{i} \leftarrow \weights_{t-1}- \eta\nabla\mathcal{L}(\weights_{t-1}, \mathcal{D}_i),     \weights_t^{i'} \leftarrow \weights_{t-1}- \eta\nabla\mathcal{L}(\weights_{t-1}, \mathcal{D}_{i}^{'}).
\notag
\end{equation}
Then, client $i$ solves the optimization problem:
\begin{equation}
    \mathop {\arg\min }\limits_{\weights_t^{i'}} \lambda\mathcal{L}_i(\weights_{t-1},\mathcal{D}_{i}^{'})+(1-\lambda)\lVert \weights_t^{i} - \weights_t^{i'} \rVert_{COS},
    \notag
\end{equation}
where $\lVert \cdot \rVert_{COS}$ refers to adjusted cosine distance. $\lambda\in(0,1]$ is a hyperparameter to balance the efficacy and stealthiness of an attack. A smaller $\lambda$ makes the attack harder to be filtered, but its efficacy is less  to be upheld. Section~\ref{sec:expacc} gives a detailed analysis.

%% file: sections/experiment.tex
\section{Evaluation}
\label{sec:experiments}



We use MNIST~\cite{lecun-mnisthandwrittendigit-2010} and FMNIST~\cite{xiao2017fashion} to train CNN same with~\cite{cao2020fltrust} and CIFAR-10~\cite{krizhevsky2009learning} to train ResNet-18~\cite{He_2016_CVPR}. Please refer to Appendix~\ref{sec:otherexpsetup} for a detailed description.
To conduct a fair comparison against existing Byzantine-robust methods, we follow the training settings of~\cite{cao2020fltrust,nguyen2022flame}.
Based on the number of classes $L$, the clients are divided into $L$ groups. 
Non-iid degree $q$ determines the heterogeneity of data distribution. 
{For example, if we use MNIST with 10 classes and $q=0.5$, the samples with label} ``0" are allocated to the group ``0" with probability 0.5 (but to other groups with probability $\frac{1-0.5}{10-1}$).\\

\noindent\textbf{Byzantine-attacks settings.} We consider six poisoning attacks aforementioned in Section~\ref{byzantineattack}. For GA, Krum, and Trim attacks, we adopt the default settings in~\cite{fang2020local}.
To achieve a fair comparison, we follow the settings of BA~\cite{nguyen2022flame}, where a white rectangle
with size 6×6 is seen as a trigger embedded on the left side of the image.
The Poisoning Data Rate (PDR) is also aligned with the settings of~\cite{nguyen2022flame}.
Wang \emph{et al.}~\cite{DBLP:conf/nips/WangSRVASLP20} did not provide a dataset for FMNIST. In the experiments, we do not consider
launching EA to FMNIST. To balance the main and attack tasks, we set $\lambda$ as 0.5.
\begin{table}[t]
\centering
\scalebox{0.9}{
\begin{tabular}{@{}c|ccc@{}}
\toprule
Dataset                   & MNIST         & FMNIST                                  & CIFAR-10                          \\ \midrule
\#clients                          & \multicolumn{3}{c}{{[}10, 100, 500{]}}                                                                                    \\
clients subsampling rate           & \multicolumn{3}{c}{1}                                                                                                     \\
non-iid degree                     & \multicolumn{3}{c}{{[}0.1, 0.5, 0.9{]}}                                                                                   \\
\#local epochs                     & \multicolumn{3}{c}{1}                                                                                                     \\
\#global epochs                    & \multicolumn{2}{c}{250}  & \multicolumn{1}{c}{1200}                                                                      \\
learning rate                      & \multicolumn{1}{c}{0.01} & \multicolumn{2}{c}{\begin{tabular}[c]{@{}c@{}}0.01 with $1e^{-5}$\\ weight decay\end{tabular}} \\
proportion of malicious clients $\xi$ & \multicolumn{3}{c}{{[}0.1, 0.6, 0.9{]}}                                                                                   \\
$\lambda$ & \multicolumn{3}{c}{0.5}
\\  $\alpha$ & \multicolumn{3}{c}{$1$} \\
\#edge-case                        & \multicolumn{1}{c}{300}  & \multicolumn{1}{c}{/}                           & 300                                        \\
DP’s $(\epsilon, \delta, \Delta)$                              & \multicolumn{3}{c}{(5, $1e^{-5}$, 5)} \\ \bottomrule                                                                                        
\end{tabular}}
\caption{FL system settings. The parameters' range and default values are in the form of ``[min, default, max]".}
\label{tab:flsetting}
\end{table}
\noindent\textbf{FL system settings.}
Table~\ref{tab:flsetting} gives the detailed parameters.
We follow the parameters setting of~\cite{mcmahan2017communication,bernstein2018signsgd}, set the minibatch size to 128, and use the {Adam optimizer~\cite{kingma2014adam}} for training LeNet and ResNet-18.
In the experiments, all the clients participate in the training from beginning to end. 
By default, we assume that there exist 100 clients splitting the training data with non-iid degree q=0.5; the proportion of malicious clients is set to $\xi$=0.6 (i.e., 60 out of 100 clients are malicious).
The testing accuracy is computed over the whole testing dataset.
{We inject triggers into the whole testing dataset to inspect the ASR of BA.
The Ardis and Southwest airplanes datasets with changed labels are used to inspect the ASR of EA in MNIST and CIFAR-10, respectively.
Note that the main focus of the experiments is to examine the complexity of \texttt{MUDGUARD} and to check if \texttt{MUDGUARD} can effectively fight against Byzantine attacks. 
Thus, we do not further present details for client selection during each round of training, which will not affect the test of Byzantine robustness.}
%
{In the clustering and robustness comparison, we define \texttt{weights-MUDGUARD} as a variant of \texttt{MUDGUARD}, which uses SGD to update models} and takes pairwise adjusted cosine similarity of updates as inputs and $L_2$ norm as clustering metric, without applying any security tools.

\begin{table*}[t]
\centering
\scalebox{1}{
\begin{tabular}{@{}cc|cccccccc@{}}
\toprule
\multicolumn{2}{c|}{Attacks}                    & baseline & GA    & LFA   & Krum  & Trim  & AA          & BA          & EA          \\ \midrule
\multicolumn{1}{c|}{\multirow{5}{*}{$\xi$}} & 0.5 & 0.975     & 0.973 & 0.967 & 0.955 & 0.965 & 0.979 / 0     & 0.972 / 0.002 & 0.966 / 0.03  \\
\multicolumn{1}{c|}{}                    & 0.6 & 0.977     & 0.975 & 0.974 & 0.952 & 0.96  & 0.979 / 0.002 & 0.968 / 0.001 & 0.968 / 0.023 \\
\multicolumn{1}{c|}{}                    & 0.7 & 0.975     & 0.971 & 0.971 & 0.956 & 0.953 & 0.977 / 0     & 0.963 / 0.002 & 0.953 / 0.07  \\
\multicolumn{1}{c|}{}                    & 0.8 & 0.969     & 0.968 & 0.964 & 0.942 & 0.944 & 0.976 / 0.003 & 0.961 / 0.005 & 0.965 / 0.085 \\
\multicolumn{1}{c|}{}                    & 0.9 & 0.969     & 0.968 & 0.968 & 0.943 & 0.937 & 0.971 / 0.005 & 0.963 / 0.002 & 0.963 / 0.093 \\ \midrule
\multicolumn{1}{c|}{\multirow{5}{*}{$n$}}  & 10  & 0.978     & 0.978 & 0.965 & 0.961 & 0.962 & 0.976 / 0     & 0.976 / 0     & 0.975 / 0     \\
\multicolumn{1}{c|}{}                    & 50  & 0.975     & 0.97  & 0.958 & 0.96  & 0.949 & 0.975 / 0     & 0.975 / 0     & 0.967 / 0.02  \\
\multicolumn{1}{c|}{}                    & 100 & 0.977     & 0.975 & 0.974 & 0.952 & 0.96  & 0.979 / 0.002 & 0.968 / 0.001 & 0.968 / 0.023 \\
\multicolumn{1}{c|}{}                    & 200 & 0.962     & 0.962 & 0.948 & 0.951 & 0.943 & 0.963 / 0.002 & 0.961 / 0     & 0.962 / 0.042 \\
\multicolumn{1}{c|}{}                    & 500 & 0.763     & 0.762 & 0.72  & 0.722 & 0.735 & 0.738 / 0.004 & 0.762 / 0.001 & 0.756 / 0.007 \\ \midrule
\multicolumn{1}{c|}{\multirow{5}{*}{$q$}}  & 0.1 & 0.976     & 0.975 & 0.978 & 0.975 & 0.975 & 0.978 / 0     & 0.975 / 0.003 & 0.976 / 0.031 \\
\multicolumn{1}{c|}{}                    & 0.3 & 0.974     & 0.973 & 0.974 & 0.966 & 0.972 & 0.98 / 0      & 0.978 / 0.002 & 0.978 / 0.026 \\
\multicolumn{1}{c|}{}                    & 0.5 & 0.977     & 0.975 & 0.974 & 0.952 & 0.96  & 0.979 / 0.002 & 0.968 / 0.001 & 0.968 / 0.023 \\
\multicolumn{1}{c|}{}                    & 0.7 & 0.898     & 0.894 & 0.872 & 0.887 & 0.906 & 0.89 / 0.013  & 0.876 / 0.011 & 0.883 / 0.039 \\
\multicolumn{1}{c|}{}                    & 0.9 & 0.709     & 0.682 & 0.705 & 0.694 & 0.689 & 0.689 / 0.017 & 0.707 / 0.025 & 0.72 / 0.06   \\ \bottomrule
\end{tabular}}
\caption{Comparison of accuracy with baseline and ASR by an increasing proportion of malicious clients ($\xi\geq 0.5$), \#clients $n$ and non-iid degree $q$, where MNIST is used. The results under targeted attacks are in the form of “testing accuracy / ASR".}
\label{tab:mnist}
\end{table*}

\begin{figure}[t]
    \centering
    \begin{subfigure}[b]{0.235\textwidth}
        \centering
        \includegraphics[width=1.04\textwidth]{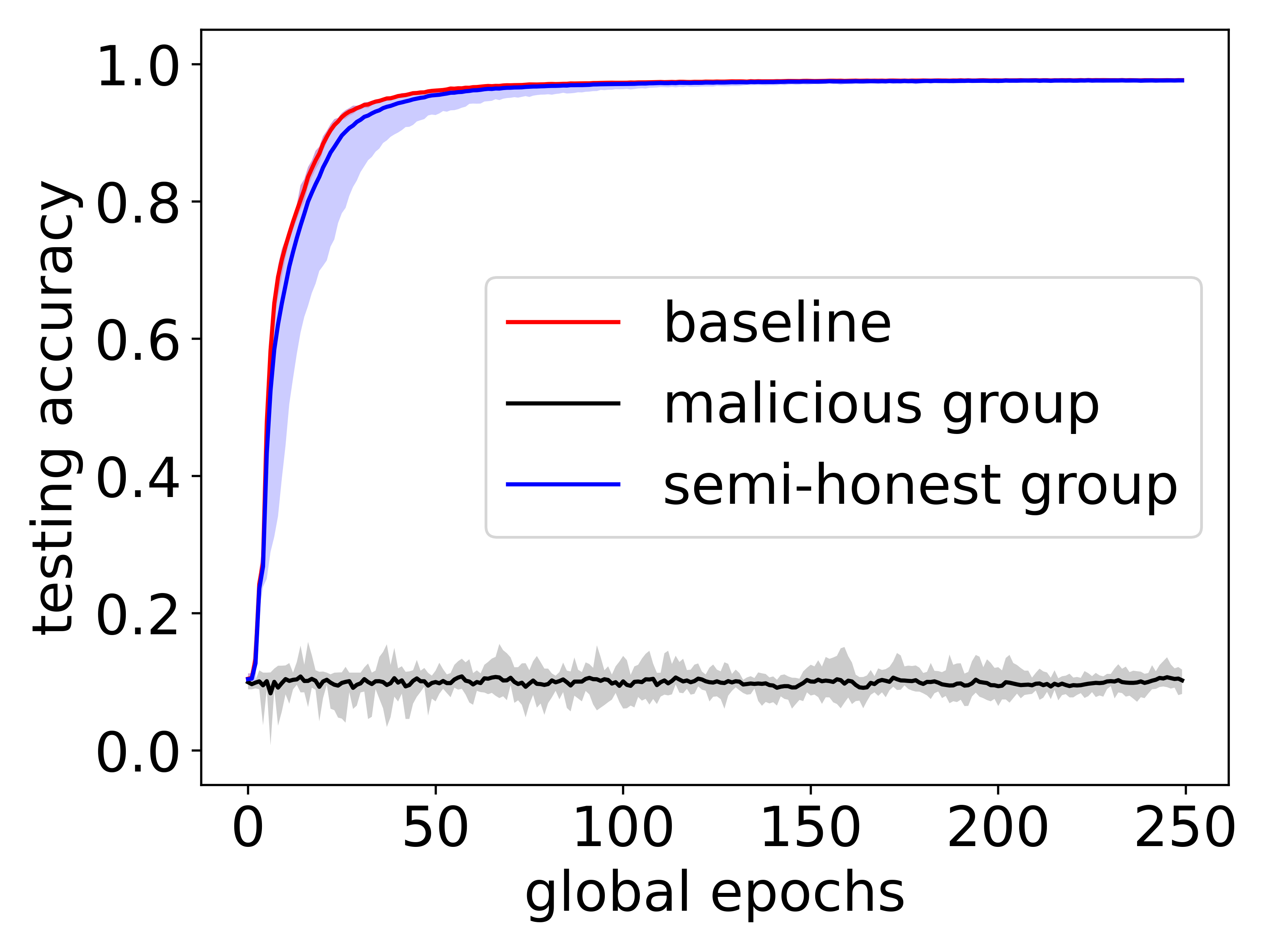}
        \caption{Gaussian Attack}
    \end{subfigure}
    \begin{subfigure}[b]{0.235\textwidth}
        \centering
        \includegraphics[width=1.04\textwidth]{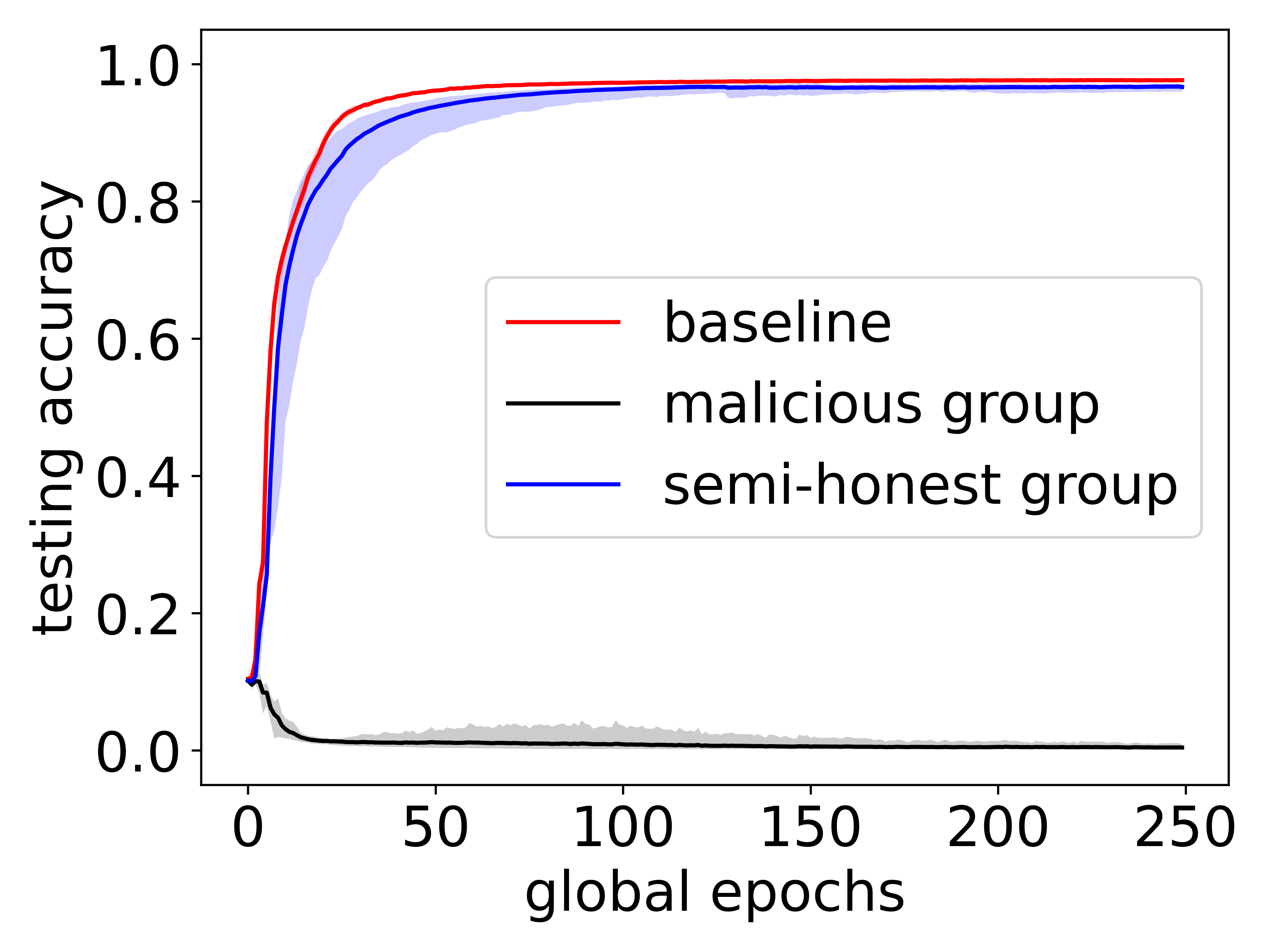}
        \caption{Label Flipping Attack}
    \end{subfigure}    
    \begin{subfigure}[b]{0.235\textwidth}
        \centering
        \includegraphics[width=1.04\textwidth]{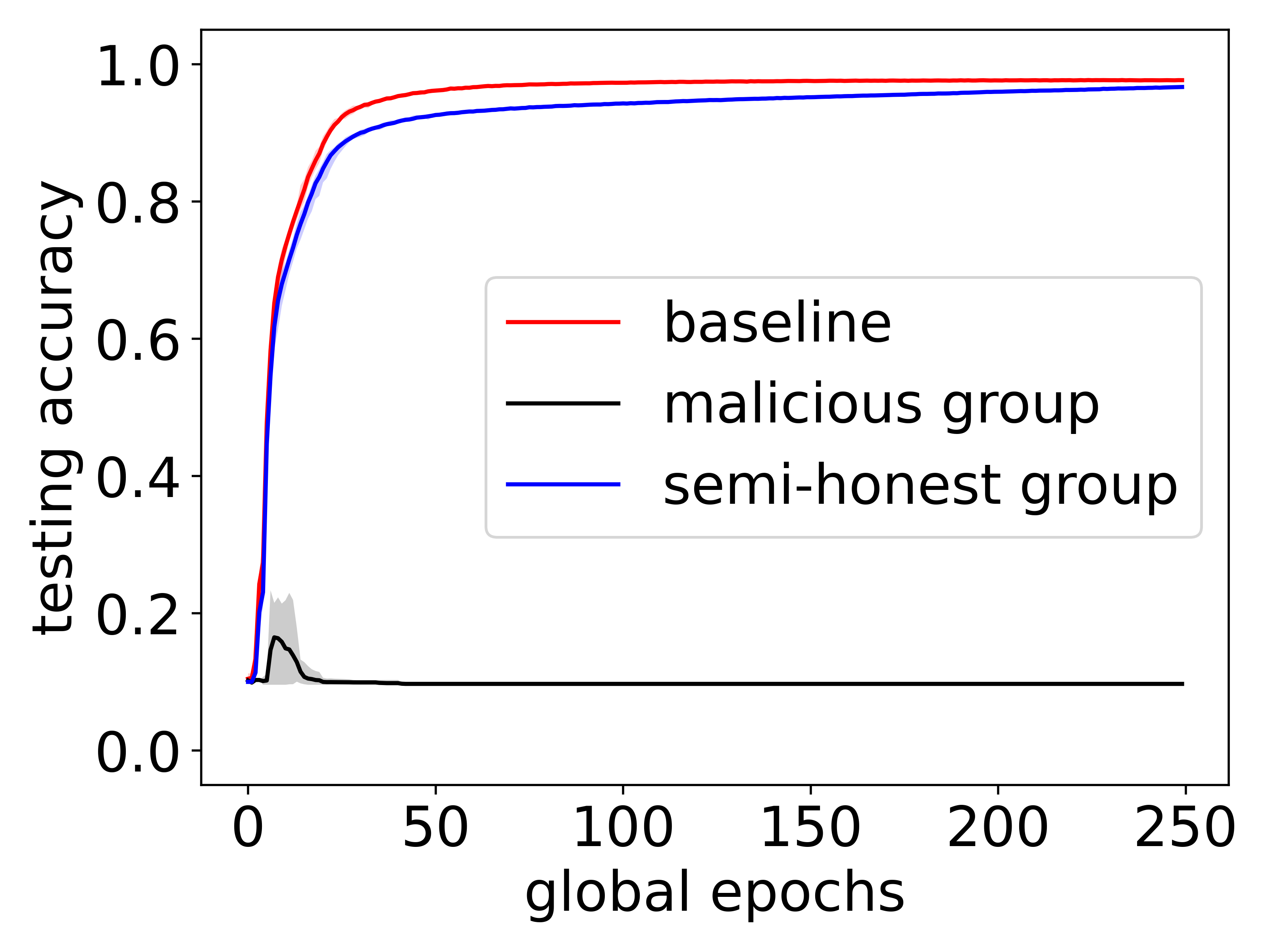}
        \caption{Krum Attack}
    \end{subfigure}    
    \begin{subfigure}[b]{0.235\textwidth}
        \centering
        \includegraphics[width=1.04\textwidth]{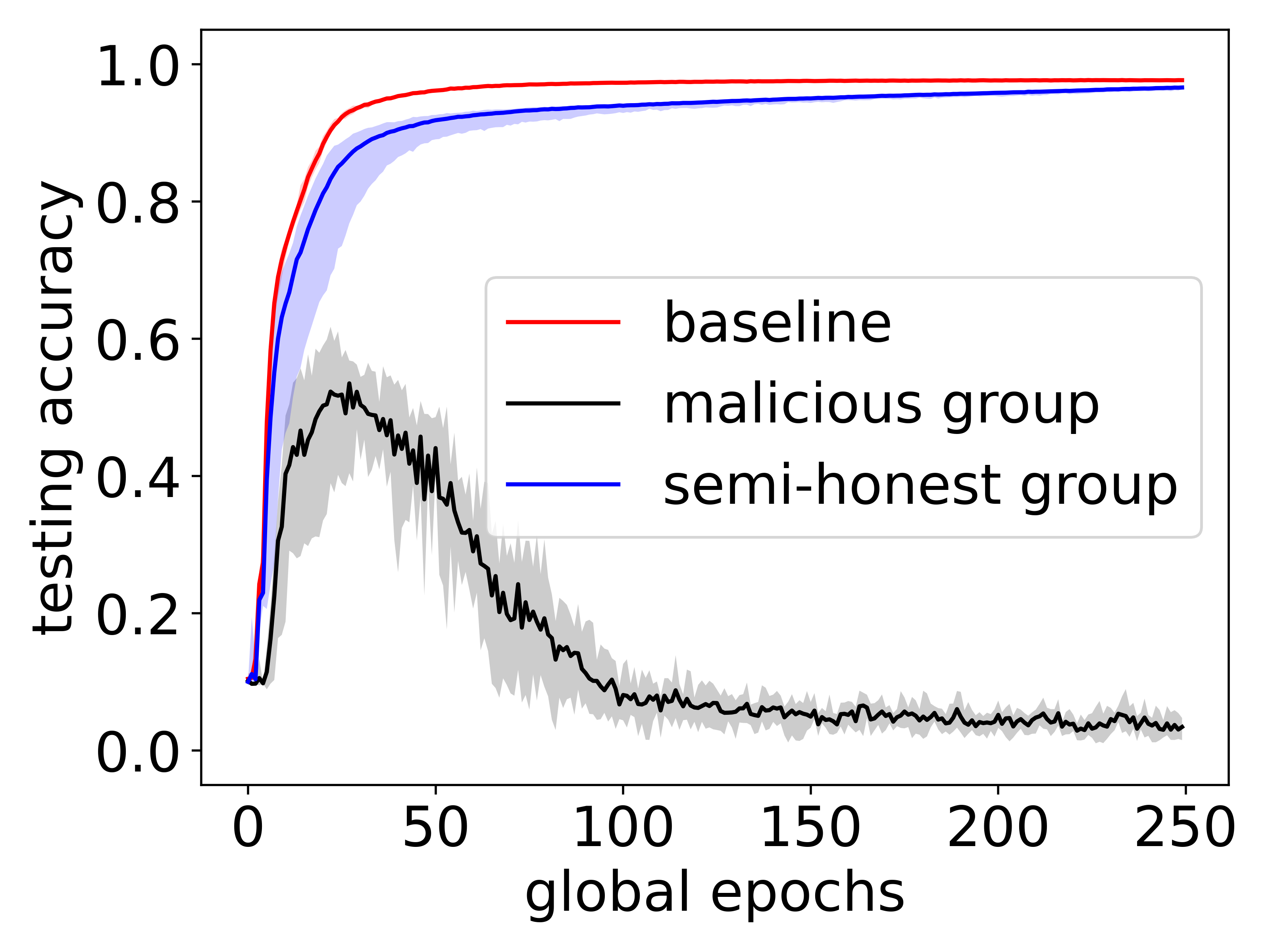}
        \caption{Trim Attack}
    \end{subfigure}  
    \begin{subfigure}[b]{0.235\textwidth}
        \centering
        \includegraphics[width=1.04\textwidth]{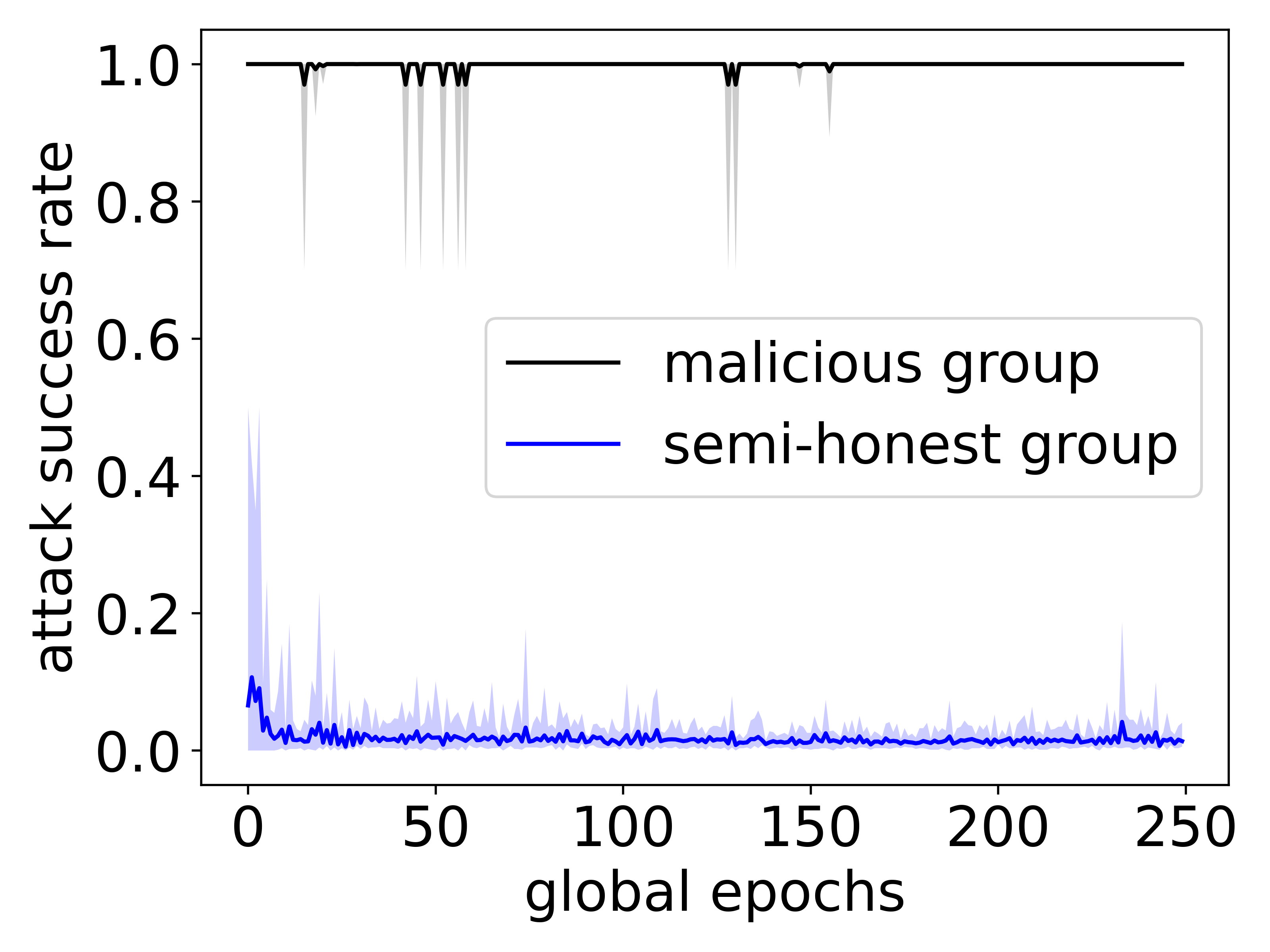}
        \caption{Adaptive Attack}
    \end{subfigure}  
    \begin{subfigure}[b]{0.235\textwidth}
        \centering
        \includegraphics[width=1.04\textwidth]{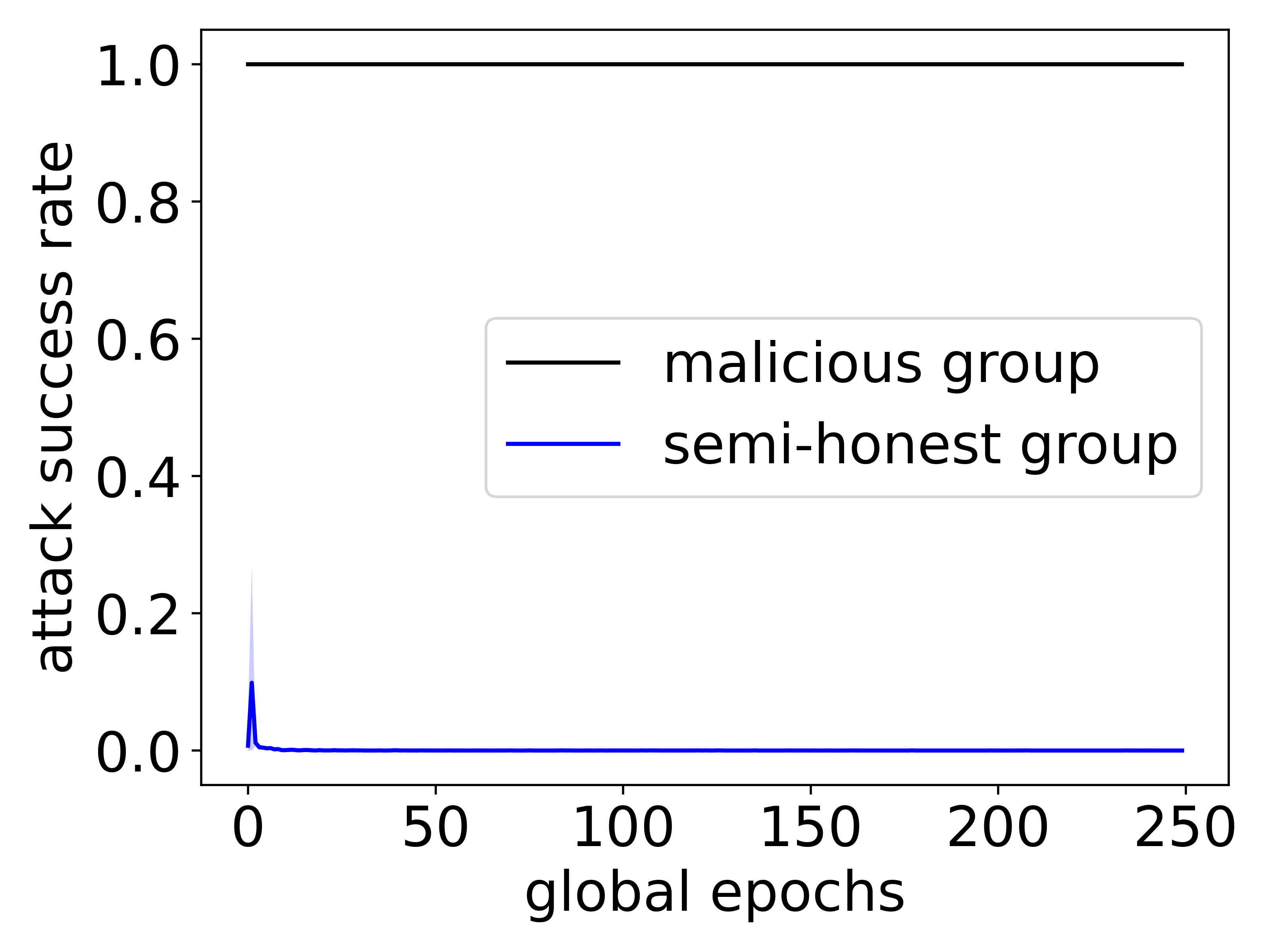}
        \caption{Backdoor Attack}
    \end{subfigure}
    \begin{subfigure}[b]{0.235\textwidth}
        \centering
        \includegraphics[width=1.04\textwidth]{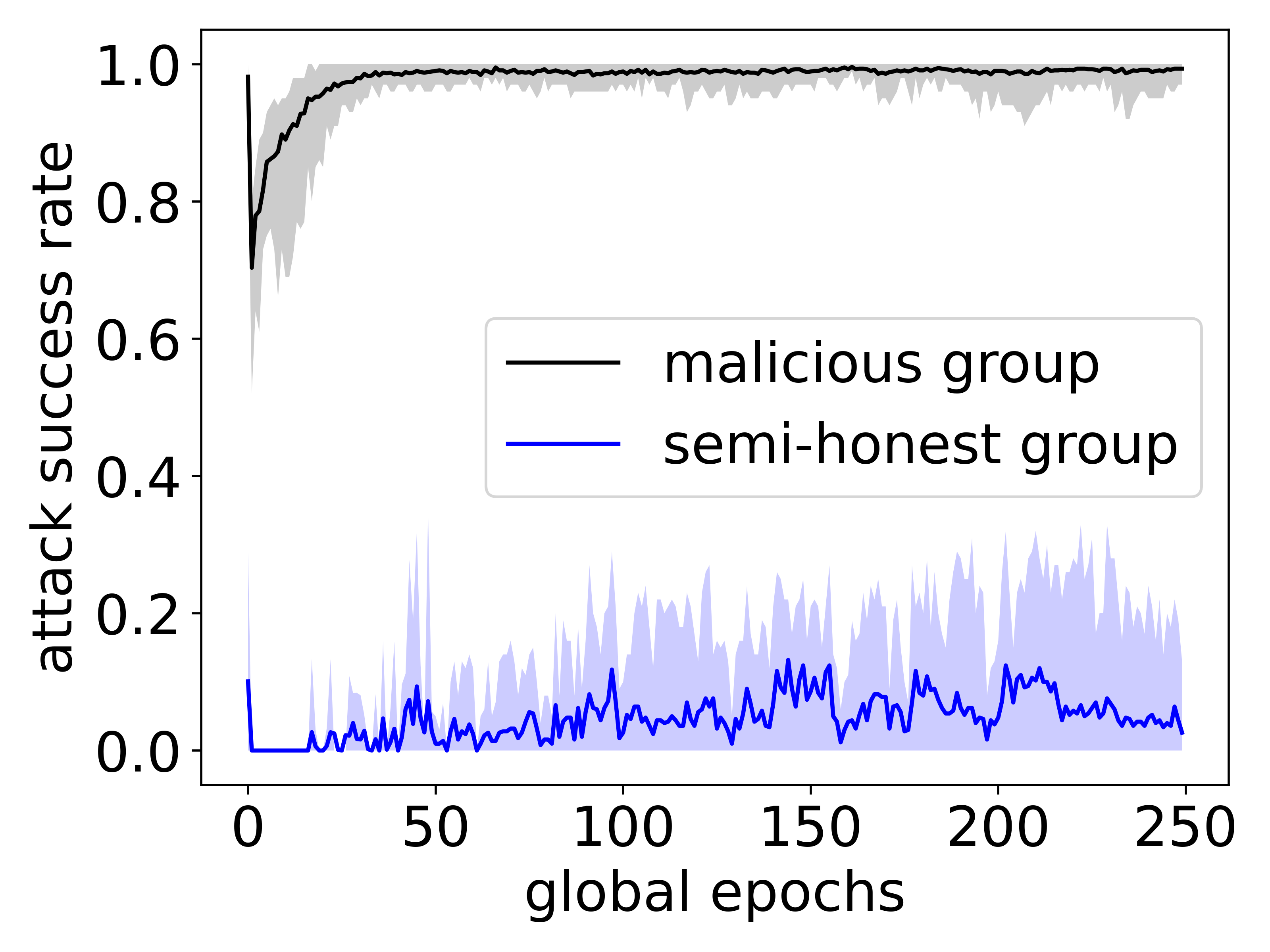}
        \caption{Edge-case Attack}
    \end{subfigure}    
    \caption{Comparison of testing accuracy among baseline, semi-honest, and malicious groups under untargeted attacks (a-d) and ASR between the groups under targeted attacks (e-f), where we train MNIST by the default settings in Table~\ref{tab:flsetting}.}
    \label{fig:xi06}
\end{figure}
\subsection{Evaluation on Accuracy}
\label{sec:expacc}

We set the baseline as a ``no-attack-and-defense" FL, which means it excludes the use of any cryptographic tools as well as Byzantine-robust solutions but only trains with fully honest parties. 
This reaches the highest accuracy and fastest convergence speed for FL training.
We then set \#clients participating in the baseline training equal to \#semi-honest clients in the malicious existence case.
We conduct each experiment for 10 independent trials and further calculate the average to achieve smooth and precise accuracy performance. 
{We evaluate \texttt{MUDGUARD}'s accuracy and ASR by varying the total number of clients, the proportion of malicious clients, and the degree of non-iid; and further compare the performance with the baseline.} 

Table~\ref{tab:mnist} shows that, under GA, AA, BA, and EA, the testing accuracy is on par with the baseline (with only a 0.008 gap on average) in MNIST. 
However, compared with the baseline, the results of \texttt{MUDGUARD} under LFA, Krum, and Trim attacks show slight drops (on average, 0.025 in MNIST). 
This is so because \texttt{MUDGUARD} has slow convergence and large fluctuation.
This is incurred by two factors. 
To reduce the overheads of secure computations, we apply binary SS in SignSGD. 
SignSGD could cause negative impacts on clustering. 
Only taking the signs of the gradients can ignore the effect of the magnitudes of the malicious gradients. 
This makes the clustering a bit prone to inaccuracy. 
The other factor is the LFA and Krum/Trim attacks either poison the training data and further poison updates or the local model to optimize the attacks. 
In the early stage of training, the malicious models do not perfectly fit the poisoned training data and local models yet. Thus, the semi-honest and malicious clients could be classified into the same cluster. 


Figure~\ref{fig:xi06} presents an overview of {the testing accuracy (of baseline and semi-honest and malicious groups) and ASR (of the two groups) under Byzantine attacks in the default settings of Table~\ref{tab:flsetting}, where MNIST is used.}\footnote{The lines refer to average cases, while the shadow outlines the max and min accuracy of each epoch.}
We see that semi-honest clients can obtain comparable accuracy to the baseline at the end of the training.
{In Figure~\ref{fig:xi06}a-d, the accuracy of the semi-honest group and the baseline sharply increase from 0.1 at epoch 0 to around 0.95 at epoch 25, then gradually converge to 0.97.
In the GA, since the malicious group can only receive aggregation of noise, their accuracy always fluctuates around 0.1, equalling a random guess probability.   
As for LFA, the model accuracy gradually drops from 0.1 (at the beginning) to 0. 
This is because their models are trained on label-flipped datasets, while the labels of the testing set are not flipped. 
If the testing set is used to detect a poisoned model, the result should be flipped labels and failing to match the labels in the testing set, which results in 0. Since semi-honest and malicious clients can be classified into the same cluster at the beginning of the training, the accuracy of their models, w.r.t. malicious clients, is larger than 0.1 in some trials.

As shown in Figure~\ref{fig:xi06}b-e, the accuracy of the semi-honest group under these attacks converges slightly slower than the baseline. 
LFA, Krum, and Trim attacks aim to either train poisoned data or optimize local poisoned models to deteriorate the global model's testing accuracy.
Due to the attacks being relatively slow and not as direct as GA, malicious updates cannot deviate 100\% from benign updates at the beginning of the training (which means that malicious and semi-honest clients could be clustered together).
However, with more training rounds, the deviation becomes clearer. Thus, \texttt{MUDGUARD} separates the two groups easily. 

AA, BA, and EA have no impact on the model's testing accuracy since their main purpose is to improve the ASR (nearly equal to 1 without defense). 
Under \texttt{MUDGUARD}, the final ASR is well suppressed.
The ASR of AA and BA are close 
to 0 in MNIST (see Figure~\ref{fig:xi06}f-g). 
But the ASR of EA is much higher than that of AA and BA, reaching an average of 0.041. 
This is because, in EA, the edge-case training sets owned by attackers are very similar to the training sets with the target labels.
If the discriminative capability of the model is not strong enough, the update directions of semi-honest and malicious gradients are also very close, making it difficult for \texttt{MUDGUARD} to distinguish them.}

The experimental results in FMNIST and CIFAR-10 show the same trends as those in MNIST under the tested attacks. Due to space limitations, we present these results in Appendix~\ref{sec:otherexp}.

\noindent\textbf{Impact of the proportion of malicious clients.}
We evaluate testing accuracy and ASR when the proportion of malicious clients $\xi\geq 0.5$. 
In Tables~\ref{tab:mnist},~\ref{tab:fmnist}, and~\ref{tab:cifar}, we can see that all accuracy results show a slightly downward trend with the increase of $\xi$ in three datasets.
For the baseline, the accuracy on average drops 0.008, 0.057, and 0.084 in MNIST, FMNIST, and CIFAR-10, respectively. 
Under GA, AA, BA, and EA, this kind of decline is on par with the baseline, whether in MNIST (0.003-0.009), FMNIST (0.059-0.66), or CIFAR-10 (0.078-0.093). 
Under the LFA, Krum, and Trim attacks, affected by the slow convergence and fluctuation, the testing accuracy of \texttt{MUDGUARD} also declines a bit more than the baseline, which is 0.012-0.028, 0.035-0.055, and 0.089-0.107 in MNIST, FMINST, and CIFAR-10, respectively.
Recall that malicious clients hold a portion of the benign dataset but do not contribute to the global model (note this equals to the case where the portion of the benign dataset is missing).
From this perspective, the accuracy should be related to the number of semi-honest clients, {where the max accuracy we achieve could correspond to the case when the clients are all semi-honest.} 
Beyond the accuracy, the ASR of EA has an upward trend while the number of malicious clients is increasing, rising by 0.005 and 0.048 in MNIST and FMNIST.
Since EA is not perfectly distinguished by \texttt{MUDGUARD}, the ASR naturally grows with the increase in the number of malicious clients. 
\textit{In conclusion, 
\texttt{MUDGUARD} is effective in maintaining accuracy even when the proportion of malicious clients is $\geq$ 50\%. While there is a slight decline in accuracy in some cases, it is on par with the baseline and does not significantly affect the overall performance of the system. }

\noindent\textbf{Impact of the total number of clients.} Tables~\ref{tab:mnist},~\ref{tab:fmnist} and~\ref{tab:cifar}, show the comparable testing accuracy of \texttt{MUDGUARD} under different attacks, as well as ASR of BA and EA when the total number of clients is set from 10 to 500.
We observe that the accuracy appears to fall whilst the client number is increasing, especially when \#clients = 500, it descends by about 0.2, 0.25, and 0.4 in MNIST, FMINST, and CIFAR-10, respectively.
This is caused by a relatively small number of training samples. For example, in CIFAR-10, each client can only be assigned 100 samples, which does not capture one minibatch size, resulting in a ``bad" performance in terms of testing accuracy. 
However, \texttt{MUDGUARD} is not affected by this factor, and it can further defend against all untargeted attacks to maintain accuracy at the same level as the baseline. 
The ASR of AA and BA are controlled to nearly 0\%.
Although EA provides a higher ASR (than AA and BA), it drops to nearly 0 when \#clients = 500, which confirms that its effectiveness relies on how well the model learns.
\textit{Overall, \texttt{MUDGUARD} can maintain a high level of accuracy under different attacks and across a range of client numbers and can effectively defend against untargeted attacks. Even when against EA, \texttt{MUDGUARD} can reduce its ASR to nearly 0\%. }

\noindent\textbf{Impact of the degree of non-iid.}
We further present the testing accuracy and ASR for the cases where the degree of non-iid ranges from 0.1 to 0.9 in Tables~\ref{tab:mnist},~\ref{tab:fmnist}, and~\ref{tab:cifar}.
We can see that in the presence of attacks, \texttt{MUDGUARD} can still remain at the same level of performance as the baseline, dropping only 0.018 on average. 
The largest decrease is 0.067 when $q$ = 0.5, which happens under the Krum attack on training LeNet with FMNIST.
Note the accuracy and the degree of non-iid show a negative correlation with/without attacks, which is also in line with the conclusion of~\cite{mcmahan2017communication} that \texttt{FedAvg} performs not well in the case of heterogeneous data distribution.
The ASR of BA appears to have a slight growth as $q$ ascends in (F)MNIST.
This is because, in the high degree of non-iid, the distances among semi-honest clients also raise. 
For targeted attacks like AA and BA, the directions of updates are closer to those of benign updates than those of untargeted attacks. 
At the beginning of training, there are cases when the distances between malicious clients and semi-honest clients are similar to those between semi-honest clients, making it difficult for \texttt{MUDGUARD} to capture subtle differences. 
For the ASR of EA, as concluded in analyzing the impact of total clients, EA performs poorly when the model's accuracy is low.
\emph{As a general conclusion, \texttt{MUDGUARD} achieves a high robustness. 
Semi-honest clients can get accurate models, while malicious clients fail to attack but also are unable to get the models.}

\noindent\textbf{Effectiveness of clustering.}
To investigate the effectiveness of our clustering approach, we present the impact on True Positives Rate (TPR) and True Negatives Rate (TNR) under all attacks of $\xi=0.6$ in Table~\ref{tab:tptn} and compare against the method of \texttt{FLAME}.
\begin{table}[t]
\centering
\scalebox{0.9}{
\begin{tabular}{@{}cccccccc@{}}
\toprule
\multicolumn{2}{c|}{\multirow{2}{*}{$\xi=$0.6}}                                                                                   & \multicolumn{2}{c}{MNIST} & \multicolumn{2}{c}{FMNIST} & \multicolumn{2}{c}{CIFAR-10} \\ \cmidrule(l){3-4}\cmidrule(l){5-6}\cmidrule(l){7-8} 
\multicolumn{2}{c|}{}                                                                                                       & TPR         & TNR         & TPR          & TNR         & TPR           & TNR          \\ \midrule
\multicolumn{1}{c|}{\multirow{3}{*}{GA}}   & \multicolumn{1}{c|}{\texttt{FLAME}}                                                 & 0.821       & 0.846       & 0.848        & 0.847       & 0.879         & 0.928        \\
\multicolumn{1}{c|}{}                      & \multicolumn{1}{c|}{\begin{tabular}[c]{@{}c@{}}\texttt{weights-}\\ \texttt{MUDGUARD}\end{tabular}} & 1           & 1           & 1            & 1           & 1             & 1            \\
\multicolumn{1}{c|}{}                      & \multicolumn{1}{c|}{\texttt{MUDGUARD}}                                                    & 0.957       & 1           & 0.94         & 1           & 0.966         & 1            \\ \midrule
\multicolumn{1}{c|}{\multirow{3}{*}{LFA}}  & \multicolumn{1}{c|}{\texttt{FLAME}}                                                 & 0.653       & 0.612       & 0.634        & 0.655       & 0.742         & 0.711        \\
\multicolumn{1}{c|}{}                      & \multicolumn{1}{c|}{\begin{tabular}[c]{@{}c@{}}\texttt{weights-}\\ \texttt{MUDGUARD}\end{tabular}} & 0.974       & 0.987       & 0.975        & 0.977       & 0.98          & 0.985        \\
\multicolumn{1}{c|}{}                      & \multicolumn{1}{c|}{\texttt{MUDGUARD}}                                                    & 0.929       & 0.924       & 0.927        & 0.916       & 0.943         & 0.967        \\ \midrule
\multicolumn{1}{c|}{\multirow{3}{*}{Krum}} & \multicolumn{1}{c|}{\texttt{FLAME}}                                                 & 0.587       & 0.622       & 0.521        & 0.63        & 0.527         & 0.578        \\
\multicolumn{1}{c|}{}                      & \multicolumn{1}{c|}{\begin{tabular}[c]{@{}c@{}}\texttt{weights-}\\ \texttt{MUDGUARD}\end{tabular}} & 0.974       & 0.953       & 0.973        & 0.968       & 0.971         & 0.966        \\
\multicolumn{1}{c|}{}                      & \multicolumn{1}{c|}{\texttt{MUDGUARD}}                                                    & 0.916       & 0.929       & 0.96         & 0.933       & 0.967         & 0.959        \\ \midrule
\multicolumn{1}{c|}{\multirow{3}{*}{Trim}} & \multicolumn{1}{c|}{\texttt{FLAME}}                                                 & 0.691       & 0.679       & 0.699        & 0.664       & 0.646         & 0.615        \\
\multicolumn{1}{c|}{}                      & \multicolumn{1}{c|}{\begin{tabular}[c]{@{}c@{}}\texttt{weights-}\\ \texttt{MUDGUARD}\end{tabular}} & 0.976       & 0.964       & 0.975        & 0.965       & 0.973         & 0.988        \\
\multicolumn{1}{c|}{}                      & \multicolumn{1}{c|}{\texttt{MUDGUARD}}                                                    & 0.938       & 0.944       & 0.927        & 0.913       & 0.964         & 0.958        \\ \midrule
\multicolumn{1}{c|}{\multirow{3}{*}{AA}}   & \multicolumn{1}{c|}{\texttt{FLAME}}                                                 & 0.591       & 0.573       & 0.612        & 0.625       & 0.766         & 0.719        \\
\multicolumn{1}{c|}{}                      & \multicolumn{1}{c|}{\begin{tabular}[c]{@{}c@{}}\texttt{weights-}\\ \texttt{MUDGUARD}\end{tabular}} & 0.998       & 0.982       & 0.99         & 0.982       & 0.984         & 0.982        \\
\multicolumn{1}{c|}{}                      & \multicolumn{1}{c|}{\texttt{MUDGUARD}}                                                    & 0.971       & 0.943       & 0.941        & 0.935       & 0.943         & 0.96         \\ \midrule
\multicolumn{1}{c|}{\multirow{3}{*}{BA}}   & \multicolumn{1}{c|}{\texttt{FLAME}}                                                 & 0.777       & 0.763       & 0.794        & 0.83        & 0.856         & 0.897        \\
\multicolumn{1}{c|}{}                      & \multicolumn{1}{c|}{\begin{tabular}[c]{@{}c@{}}\texttt{weights-}\\ \texttt{MUDGUARD}\end{tabular}} & 0.957       & 0.969       & 0.965        & 0.97        & 0.963         & 0.979        \\
\multicolumn{1}{c|}{}                      & \multicolumn{1}{c|}{\texttt{MUDGUARD}}                                                    & 0.936       & 0.928       & 0.926        & 0.931       & 0.947         & 0.928        \\ \midrule
\multicolumn{1}{c|}{\multirow{3}{*}{EA}}   & \multicolumn{1}{c|}{\texttt{FLAME}}                                                 & 0.313       & 0.32        & \_            & \_           & 0.248         & 0.288        \\
\multicolumn{1}{c|}{}                      & \multicolumn{1}{c|}{\begin{tabular}[c]{@{}c@{}}\texttt{weights-}\\ \texttt{MUDGUARD}\end{tabular}} & 0.899       & 0.903       & \_            & \_           & 0.893         & 0.921        \\
\multicolumn{1}{c|}{}                      & \multicolumn{1}{c|}{\texttt{MUDGUARD}}                                                                          & 0.856       & 0.876       & \_            & \_           & 0.827         & 0.83         \\ \bottomrule
\end{tabular}}
\caption{Effectiveness of clustering among FLAME method, \texttt{weights-MUDGUARD}, and \texttt{MUDGUARD}.}
\label{tab:tptn}
\end{table}
The \texttt{FLAME} takes updates as inputs and cosine similarity as a metric for clustering.
Note on the server side, \emph{Model Segmentation} does not need to identify which cluster is malicious/semi-honest. 
We consider false positives to occur if semi-honest clients are grouped with the malicious.
On average, under GA, the TPR and TNR improve from 0.151 and 0.126 in \texttt{FLAME} to 1 in \texttt{weights-MUDGUARD}, respectively.
Since \texttt{MUDGUARD} is based on SignSGD, only the signs of updates are taken. 
Ignoring the magnitude effect, there is a reduction in TPR (an average reduction of 0.046 as compared to \texttt{weights-MUDGUARD}). 
Furthermore, TNR does not drop as we set the appropriate parameters according to Theorem~\ref{the:1}.
The same changes can be captured in the case of LFA: \texttt{weights-MUDGUARD} has an average increase of 0.3 and 0.324 in TPR and TNR, respectively, as compared to \texttt{FLAME}. 
Compared with \texttt{weights-MUDGUARD}, \texttt{MUDGUARD} drops by 0.04 and 0.05. 
We see that under other attacks (LFA, Krum, Trim, AA, BA, and EA), TPR and TNR are lower than the case under GA. 
Because they launch attacks on either training data or optimizing poisoned models, all updates at the beginning of training have high similarities, yielding those updates being clustered together and the cases of misclustering.
The true rates of CIFAR-10 are higher than those of (F)MNIST, 
because we can set more rounds to train ResNet-18. 
After the model converges, the true rates reach almost 100\%.
Thence, \texttt{MUDGUARD} obtains more correct clusters.

{From the above analysis, we conclude that TNR and TPR are related to the number of training rounds, attack type, and the values of updates. 
Because \texttt{MUDGUARD} groups high similarity updates into one cluster and does not need to identify malicious/semi-honest clusters, the performance of clustering is less affected by the proportion of malicious clients. 
Similar results, like Table~\ref{tab:tptn}, can be captured even in the case when $\xi>$0.6.} 
Through Figure~\ref{fig:xi06}, Table~\ref{tab:tptn}, and the above discussion, we state that although TNR and TPR are affected to a certain extent by binary SS,   
from the view of testing accuracy and ASR, \texttt{MUDGUARD} achieves higher TPR and TNR than \texttt{FLAME}.
In terms of other analyses of hypeparameters (i.e., $\alpha$ and $\lambda$), please refer to Appendix~\ref{sec:otherexp}.
Appendix~\ref{sec:convana} shows the detailed convergence analysis of \texttt{MUDGUARD}.

\noindent\textbf{Robustness comparison against other methods.}
We present a comparison among \texttt{MUDGUARD} and SOTA methods (\texttt{FLTrust}, \texttt{FLAME}, \texttt{Zeno++}, and \texttt{EIFFeL}) in terms of robustness, as shown in Figure~\ref{fig:comp}, where MNIST is used. 
Several Byzantine-robust FL systems can easily and directly apply to \texttt{EIFFeL}. 
We select the two of them (please refer to~\cite{roy2022EIFFeL}) for comparison, namely FLTrust and Zeno++. 
For brevity, we refer to them as \texttt{EIFFeL-FLtrust} and \texttt{EIFFeL-Zeno++} hereafter. 
To demonstrate the advantages of \texttt{MUDGUARD} (based on SignSGD), we also compare its robustness with both SignSGD and \texttt{FedAvg} w/o defense.
{To investigate the impact of the cryptographic tools on testing accuracy and ASR, we also compare \texttt{MUDGUARD} with \texttt{weights-MUDGUARD}.} 
%
{One may see that \texttt{MUDGUARD}, countering the case of the malicious majority on the client side, does outperform most existing approaches.} 

In Figure~\ref{fig:comp}a-d, the accuracy of \texttt{(weights-)MUDGUARD} and \texttt{EIFFeL-(FLTrust/Zeno++)} can be maintained at the same level as the baseline (about 0.97).
Due to the impacts of misclustering, \texttt{weights-MUDGUARD} has a 0.02 accuracy gap with \texttt{EIFFeL-FLTrust}.  
\texttt{MUDGUARD} (with DP noise) commits a roughly 0.01 accuracy loss as compared to \texttt{weights-MUDGUARD}.
The accuracy of others decreases with the increase in malicious clients, especially when $\xi\geq0.5$, the accuracy drops abruptly to the same level of \texttt{FedAvg} without defense. 
For the ASR of AA and BA, apart from \texttt{EIFFeL-(FLTrust/Zeno++)}, \texttt{MUDGUARD} and \texttt{weights-MUDGUARD}, all the remaining methods suddenly increase to 1 at $\xi$=0.4/0.5.
Since EA has better attack ability (than AA and BA), \texttt{weights-MUDGUARD} and \texttt{MUDGUARD} suffer from a nearly 0.08 gap to \texttt{EIFFeL-FLTrust}. 
The ASR of others can raise from $\xi=0.1$ and finally reach 1.0 at $\xi=0.5$. 
\texttt{SignSGD} only limits the magnitude of malicious updates rather than filtering them out. Still, it can provide a certain level of defense (Figure~\ref{fig:comp}) when there is a low malicious proportion ($\xi$=0.1-0.2) (compared to \texttt{FedAvg} having an average of 0.3 higher testing accuracy under untargeted attacks, and an average lower ASR of 0.4 under targeted attacks). 
As the number of malicious clients rises, its robustness drop to the level of \texttt{FedAvg} w/o defense.

\texttt{FLAME} indicates that a small-size cluster should be a malicious group. Thus, it is easy to confirm malicious clients via clustering. 
In the case of the malicious majority, it is hard to identify the malicious/semi-honest via group size.
\texttt{FLTrust} assumes that before training, an honest server collects and trains on a small dataset. 
In each round, the server takes the updates trained by this small dataset as the root of trust. 
The ``trusted" results are then compared to the updates sent by the clients. 
If the cosine similarity between them is too small, the updates will be filtered out.
With this approach, the accuracy of the global model remains equivalent to that of the baseline. 
We state that \texttt{MUDGUARD} is on par with \texttt{FLTrust}, but it does not suffer from the restriction that the servers need to collect an auxiliary dataset ahead of training. 
We also see that when the proportion of malicious clients rises, the accuracy of \texttt{MUDGUARD} shows a slight decline. 
When clients upload their updates, \texttt{MUDGUARD} can only aggregate them with similar directions.
If there is only a small percentage of semi-honest clients in the system, we naturally have an incomplete training set, causing a loss in accuracy. Note the same trends as those in MNIST can be seen in FMNIST (Figure~\ref{fig:compfmnist}) and CIFAR-10 (Figure~\ref{fig:compcifar10}).

\begin{figure}[!htbp]
    \centering
    \begin{subfigure}[b]{0.235\textwidth}
        \centering
        \includegraphics[width=1.04\textwidth]{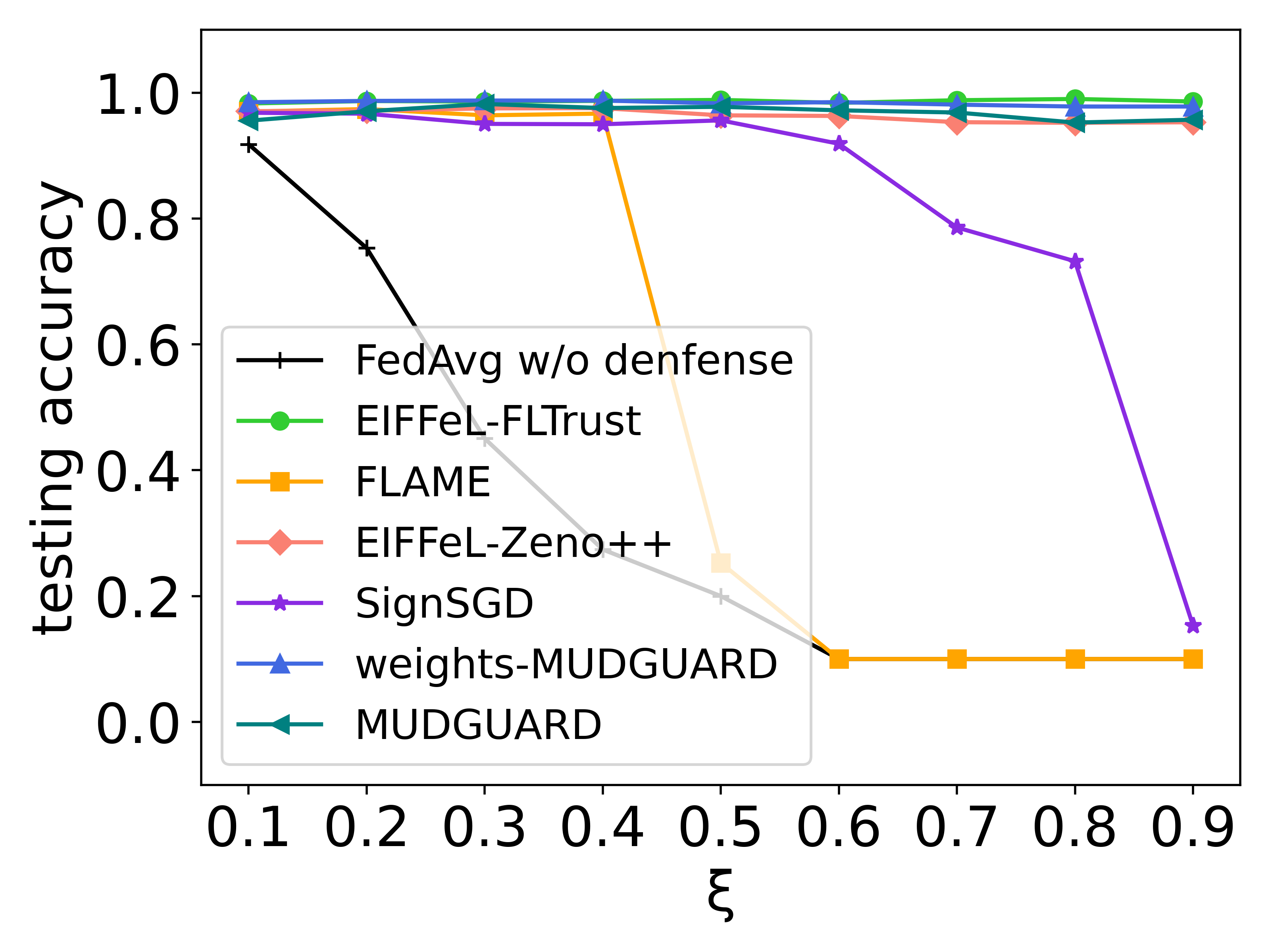}
        \caption{Gaussian Attack}
    \end{subfigure}
    \begin{subfigure}[b]{0.235\textwidth}
        \centering
        \includegraphics[width=1.04\textwidth]{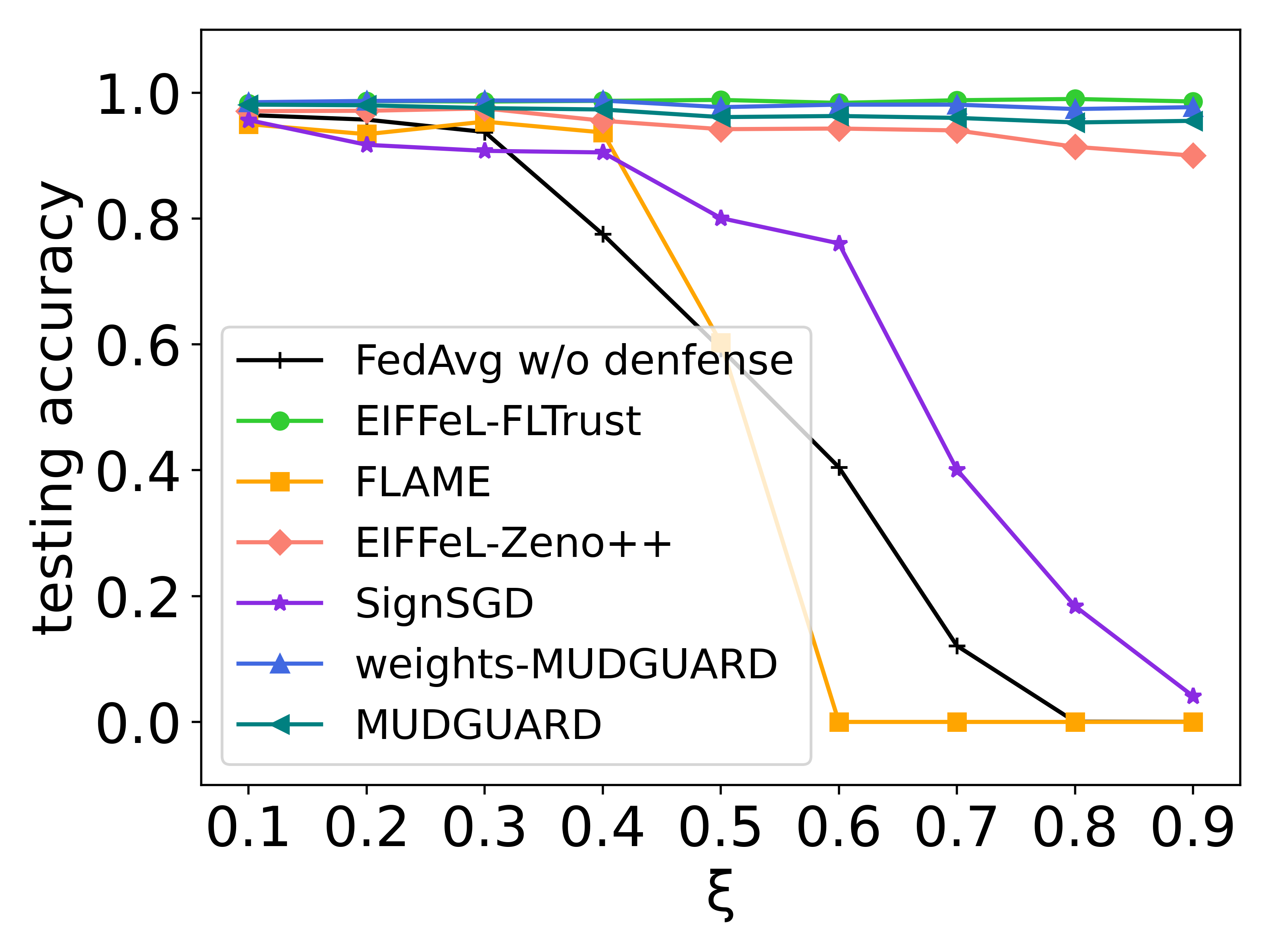}
        \caption{Label Flipping Attack}
    \end{subfigure}    
    \begin{subfigure}[b]{0.235\textwidth}
        \centering
        \includegraphics[width=1.04\textwidth]{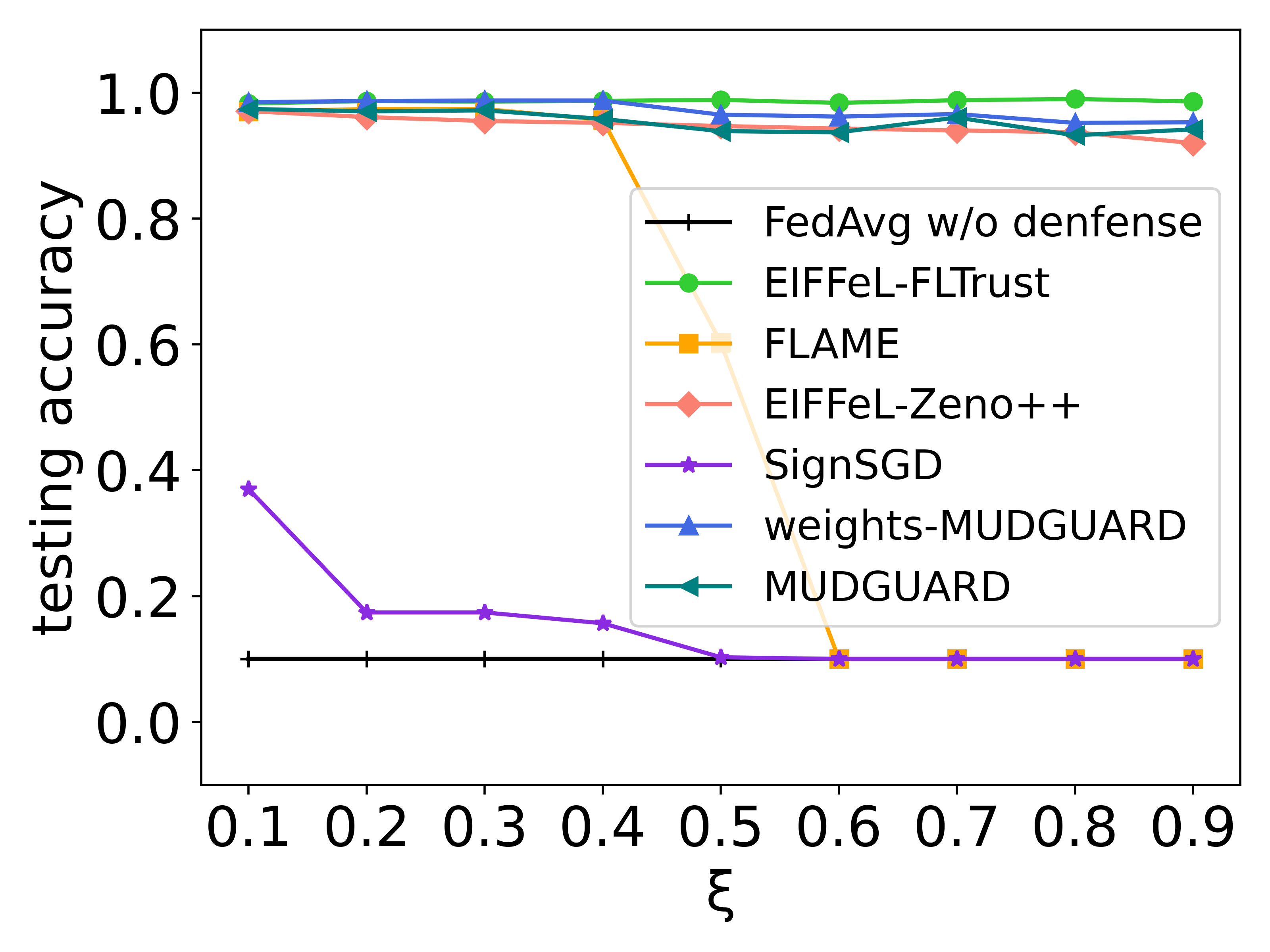}
        \caption{Krum Attack}
    \end{subfigure}    
    \begin{subfigure}[b]{0.235\textwidth}
        \centering
        \includegraphics[width=1.04\textwidth]{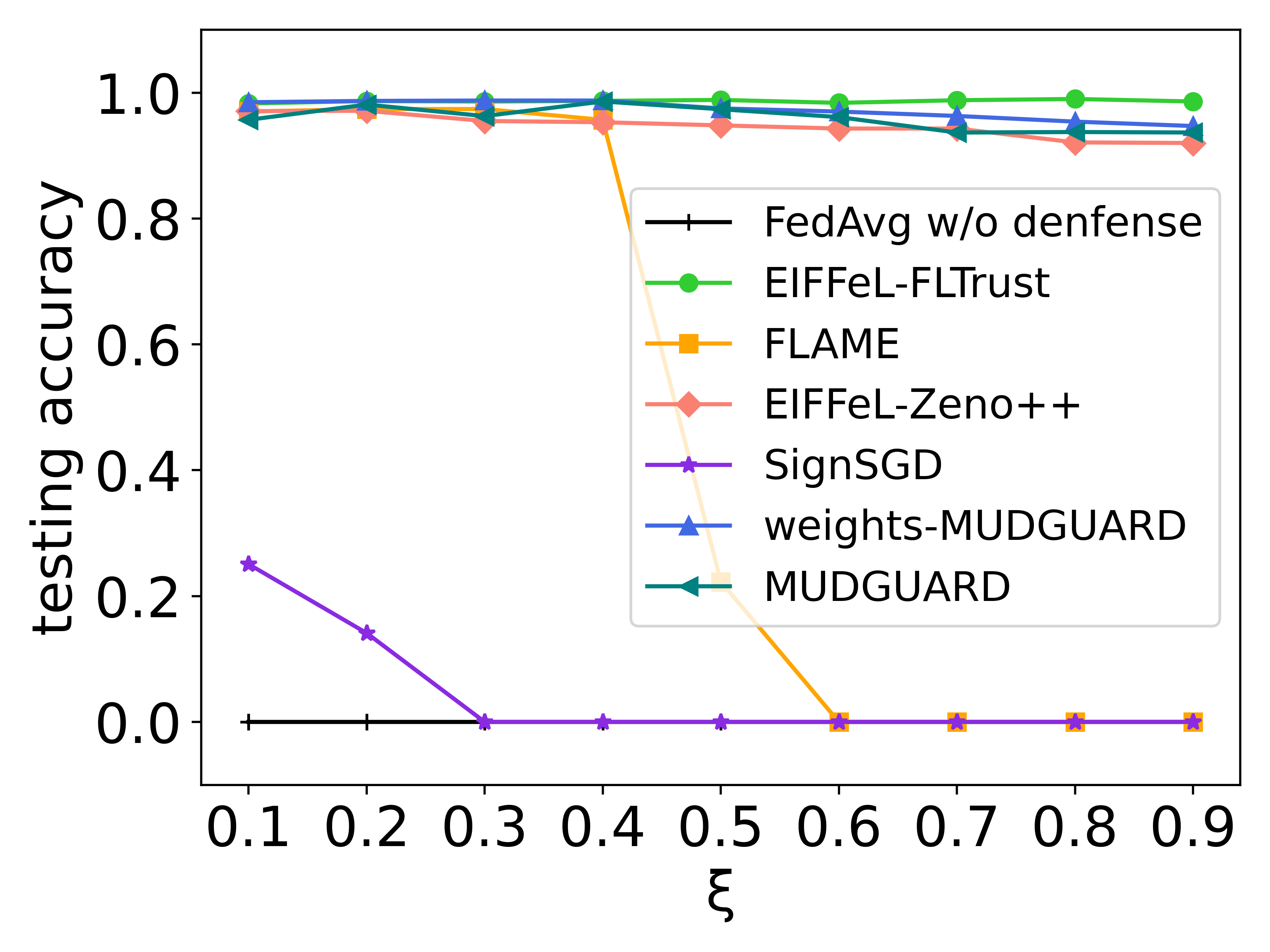}
        \caption{Trim Attack}
    \end{subfigure}  
        \begin{subfigure}[b]{0.235\textwidth}
        \centering
        \includegraphics[width=1.04\textwidth]{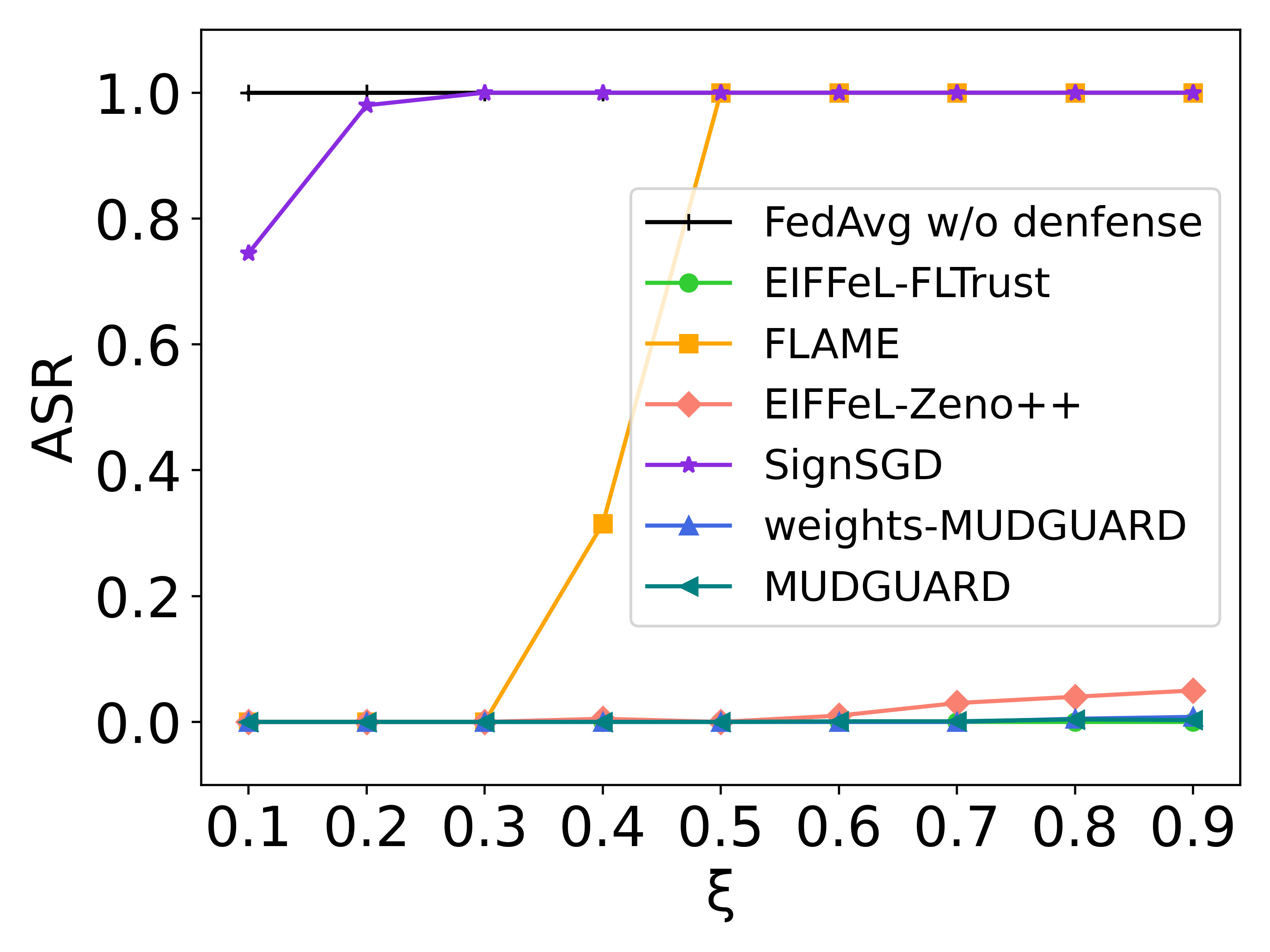}
        \caption{Adaptive Attack}
    \end{subfigure}  
    \begin{subfigure}[b]{0.235\textwidth}
        \centering
        \includegraphics[width=1.04\textwidth]{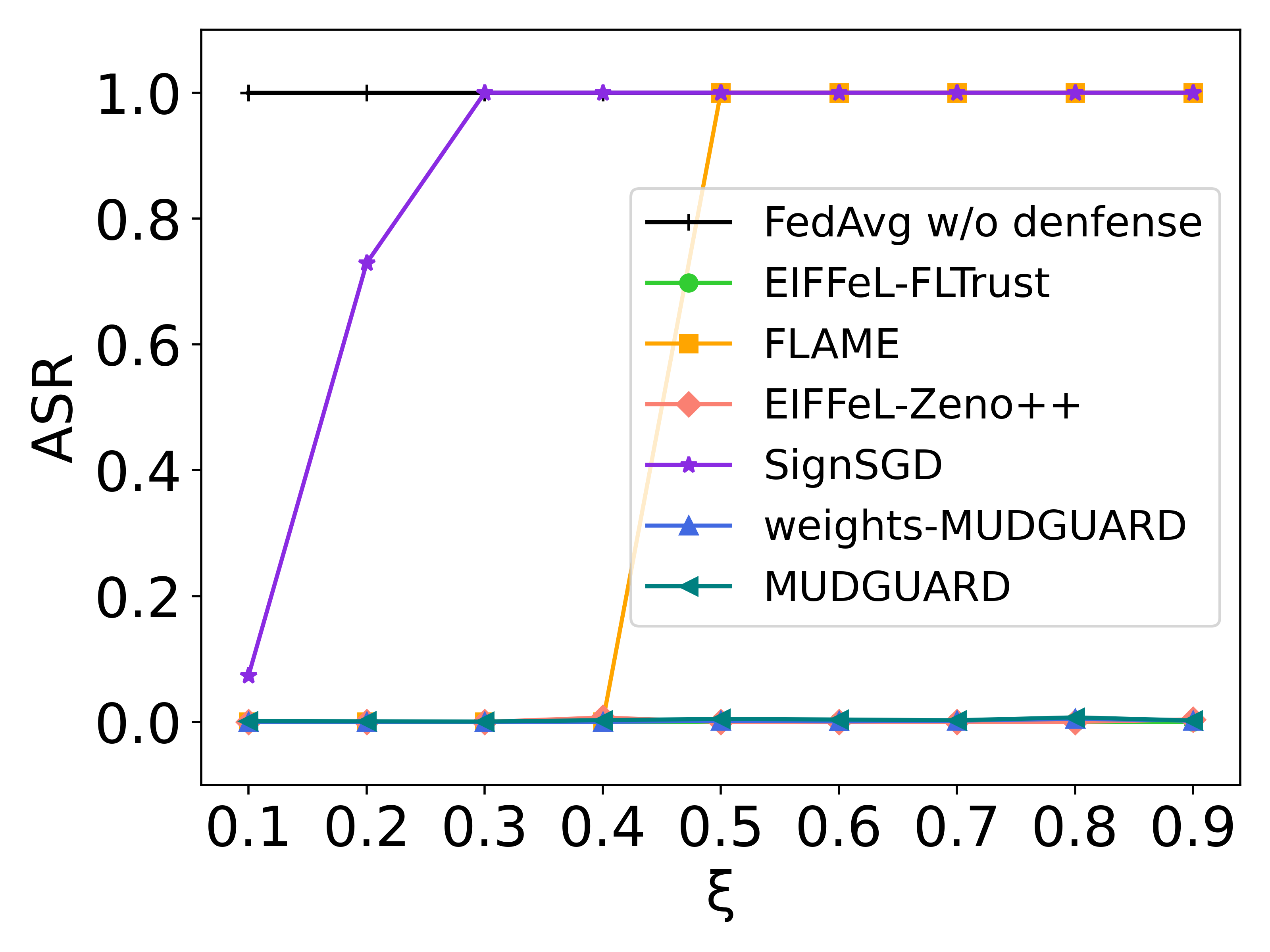}
        \caption{Backdoor Attack}
    \end{subfigure}
    \begin{subfigure}[b]{0.235\textwidth}
        \centering
        \includegraphics[width=1.04\textwidth]{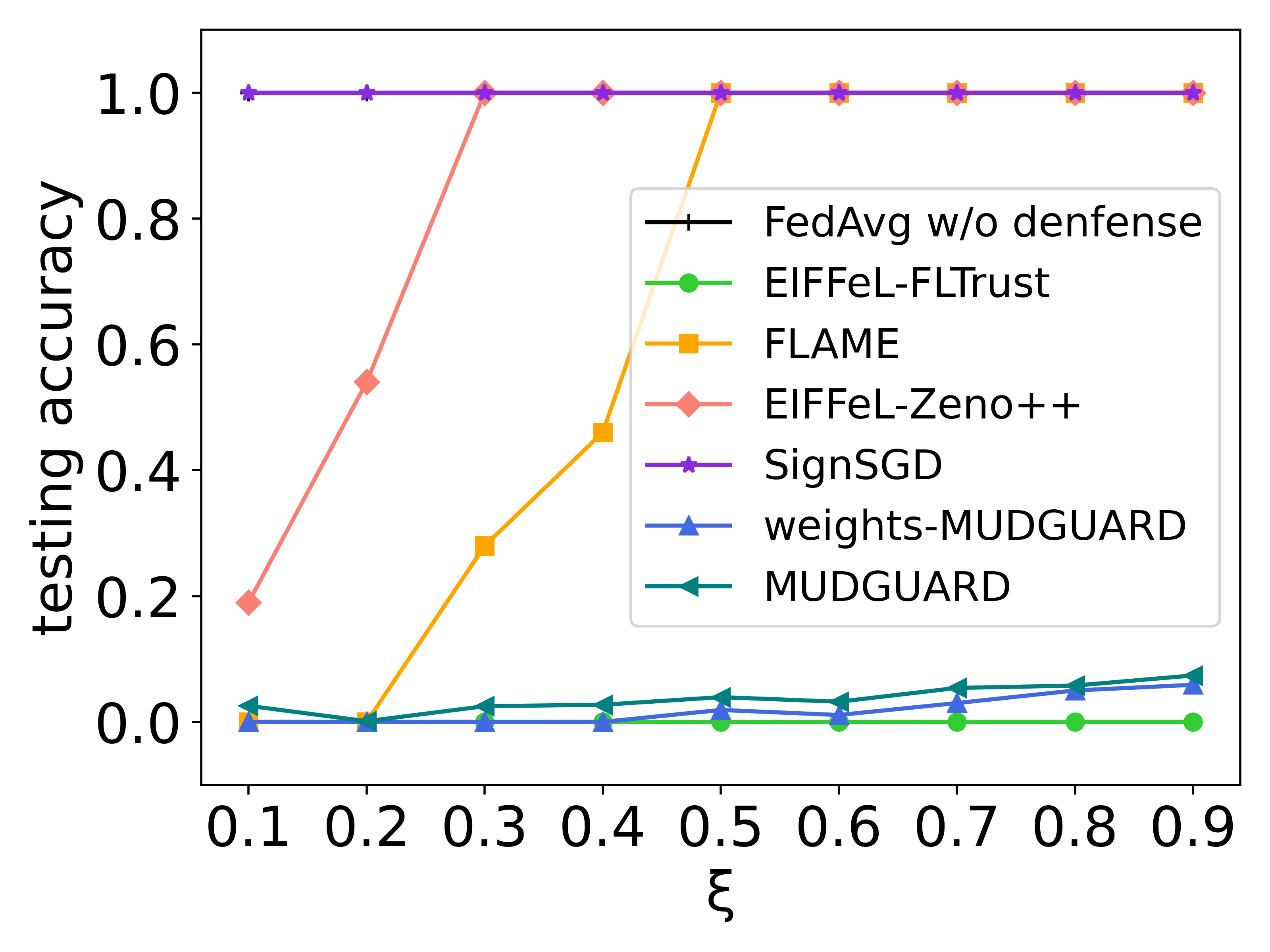}
        \caption{Edge-case Attack}
    \end{subfigure}    
    \caption{Comparison with Byzantine-robust methods by $\xi=0.1 - 0.9$.}
    \label{fig:comp}
\end{figure}

\subsection{Evaluation on Overheads}
\label{sec:overheads}
\begin{table*}[!ht]
\centering
\scalebox{0.9}{
\begin{tabular}{@{}c|cccccccc@{}}
\toprule
\multirow{2}{*}{Threat Model} & \multicolumn{4}{c}{Server-side}                                          & \multicolumn{4}{c}{Client-side}                                          \\ \cmidrule(l){2-5}\cmidrule(l){6-9} 
                              & \multicolumn{2}{c}{Semi-honest} & \multicolumn{2}{c}{Malicious Minority} & \multicolumn{2}{c}{Semi-honest} & \multicolumn{2}{c}{Malicious Majority}\\ \cmidrule(l){2-3}\cmidrule(l){4-5} \cmidrule(l){6-7}\cmidrule(l){8-9}  
Traning Model                 & LeNet         & ResNet-18       & LeNet             & ResNet-18          & LeNet        & ResNet-18        & LeNet            & ResNet-18           \\ \midrule
Runtime(Second)               & 0.43$\pm$0.07/1.3$\pm$0.12      & 1.28$\pm$0.25/3.15$\pm$0.31        & 4.54$\pm$0.84/14.41$\pm$1.49         & 24.03$\pm$2.83/70.74$\pm$2.83         & 14.33$\pm$0.74         & 23.71$\pm$3.45             & 14.56$\pm$1.61              & 23.93 $\pm$4.31               \\
Communication Costs (MB)      & 16.20/53.46         & 34.82/314.45         & 873.23/2776.38          & 5151.18/15572.68         & 16.34         & 758.48             & 16.34            & 758.48               \\ \bottomrule
\end{tabular}}
\caption{Comparison of overheads among different threat models over LeNet $\&$ ResNet-18. The results on the server side are in the form of "optimized/unoptimized".}
\label{tab:overheads}
\end{table*}
We conduct overheads assessment together with the evaluation of accuracy. 
The overheads presented in Table~\ref{tab:overheads} capture the runtime and communication costs incurred by the implemented cryptographic tools on the server side and model training on the client side.
Recall that we propose an optimization in  Figure~\ref{fig:optimization} (Section~\ref{sec:MUDGUARD}). 
We present the average overheads of each round of training so as to illustrate the optimized and unoptimized results in terms of different training models and honest/malicious contexts on the server side.
We use LeNet and ResNet as models, and the overheads are related to their dimensionality (instead of the training data). 

\noindent\textbf{Runtime.}
In general, we see that providing robustness in the malicious context should require more runtime than in the semi-honest.
This is because extra operations for verification are taken, {e.g., using HHF to verify whether a received aggregation is correct}.
The unoptimized ResNet-18 takes 3.15s per round in the semi-honest context while costing 70.74s (approx. an increase of 22 times) in the malicious minority. 
ResNet-18 has more model parameters than LeNet, leading to extra computational operations on cryptographic tools, which can be seen, in the malicious context, 70.74s v.s. 14.41s. 
By binary SS and polynomial transformation,  Table~\ref{tab:overheads} shows that the runtime of LeNet and ResNet-18 are reduced by 68.75\% (4.54s) and 66.05\% (24.03s), respectively, under malicious-minority servers.

\noindent\textbf{Communication costs.}
Similar to runtime, malicious-minority servers consume a considerable amount of communication cost compared to semi-honest ones. 
Table~\ref{tab:overheads} shows that after optimization, the communication costs drop to 33\% in the malicious minority and 10.83\% in the semi-honest with ResNet-18. 
In the worse case, we consume 15,572.68MB  bandwidth per round under malicious minority, but we optimize the cost to 5,151.18MB.
In the semi-honest context, LeNet achieves the best performance, requiring 16MB with optimization, which is 30.19\% of the unoptimized cost (53.46MB).

Under the same contexts, we present the overheads of the client side for FL training in Table~\ref{tab:overheads}. Note the use of advanced FL techniques, such as those outlined in~\cite{pmlr-v119-hamer20a,NEURIPS2020_e32cc80b}, can be employed to enhance computing and communication efficiency in \texttt{MUDGUARD}. 
Since applying those is straightforward, we will not go into further detail on this matter.



%% file: sections/conclusion.tex
\section{Conclusion}
\label{sec:conclusions}

We propose a novel Byzantine-robust and privacy-preserving FL system.
To defend against malicious majority clients, we introduce \emph{Model Segmentation}.
Leveraging cryptographic tools and DP, our design enables training to be performed correctly without breaching privacy.
Our experimental results demonstrate that the proposed protocol effectively deals with various malicious settings and outperforms most existing solutions.
In addition to theoretical and empirical analysis, we also provide extensive discussions on model accuracy, advantages of \texttt{MUDGUARD}, and its limitations. Due to page limits, please refer to Appendix~\ref{sec:otherdis}.


%% file: sections/notation.tex
\subsection{Notation}
\label{sec:notation}

The frequently used notations are in Table~\ref{tab:notations}.
\begin{table}[t]
\centering
\begin{tabular}{@{}c|l@{}}
\toprule
 \textbf{Notation}&  \textbf{Description}  \\ \midrule
 $\gradients_{t}^{i}$ & gradients of $i$-th client at $t$-th round \\ 
 $\weights_{t}^{i}$ & weights of $i$-th client at $t$-th round  \\ 
 $T$ & the number of rounds \\  
 $n$&the number of clients  \\  
 $S$ & the number of servers \\  
 $m$ & the number of malicious clients \\  
 $k_i$ & the number of data instances of $i$-th client \\  
 $c$ & the number of clusters \\  
 $l$ & the cluster labels \\  
 $E$ & the number of epochs \\  
 $\mathcal{D}_i$ & the dataset of $i$-th client\\  
 $\eta$ & learning rate \\  
 $\aggradient^z$ & the aggregation of gradients of $z$-th cluster\\  
 $[n]$ & a set of numbers ranging from 1 to n \\  
 $[\![\cdot]\!]$ & secret shared format  \\  
 $\mathsf{CosM}$ & pairwise adjusted cosine similarity  matrix \\  
 $\mathsf{EudM}$ & pairwise $L_2$ distance matrix \\  
 $\mathsf{IndM}$ & indicator matrix\\  
 $\delta, \phi$ & secret keys of homomorphic hash function\\  
 $\Delta, \sigma, \epsilon$ & parameters of differential privacy \\  
 $\mathcal{N}$ & Gaussian noise \\  
 $\alpha$ & density parameter \\  
 $\encoding(\cdot)$ & encoding algorithm\\  
 $\decoding(\cdot)$ & decoding algorithm \\  
 \bottomrule
\end{tabular}
\caption{Notation summary}
\label{tab:notations}
\end{table}

%% file: sections/algorithm.tex
\subsection{Implementation Algorithms}
\label{sec:IA}

In the evaluation, we implement the proposed \texttt{MUDGUARD} (with optimization) mainly based on Algorithms~\ref{al:fl} and~\ref{al:mpc}.  

\begin{algorithm}[!ht]
\small
\SetAlgoLined
\SetKwInput{KwInput}{Input}
\SetKwInput{KwOutput}{Output}
\DontPrintSemicolon
    \KwInput{
training dataset $\mathcal{D}=\bigcup\limits_{i=1}^n \mathcal{D}_{i}$\\}
    \KwOutput{
    global models $\{\weights^{i} \mid i \in [n] \}$\\
    }
\textbf{ServerAggregation:}\\

Initialize global model $\weights_{0}$\\

\For{each global epoch $t$ = 1,2,$\cdots$,$T$}{
\For{client $i \in n$ \textbf{in parallel}}{
$[\![\gradients^{i}_{t}]\!]\leftarrow$ ClientUpdate$(i, \weights^{i}_{t-1})$
}
$[\![\aggradient_t^j]\!]\leftarrow$\textbf{Algorithm 2}\\
return $[\![\aggradient_t^j]\!], \hash_{\delta, \phi}([\![\aggradient_t^j]\!])$ to clients
}
\textbf{ClientUpdate($i$, $\weights^{i}_{t-1}$):}\\
$\mathcal{B}\leftarrow$(split $\mathcal{D}_{i}$ into batches of size $b$)\\
$\indm\leftarrow$MajorityVote($\{\indm_i\mid i \in [s]\}$)\\
  $l\leftarrow$DBSCAN($\indm$)\\
  reconstruct $\aggradient^{i}_{t-1}$ by $[\![\aggradient^{i}_{t-1}]\!]$ \\
  \eIf{$\prod\hash_{\delta, \phi}(\hat{\gradients}_{t-1}^i)=\hash_{\delta, \phi}(\aggradient_{t-1}^i)$}{accept and continue
  }{refuse and break}
 $\weights_{t}^{i} \leftarrow \weights_{t-1}^{i} - \eta\cdot\signn(\aggradient^{i}_{t-1})$\\
  $\gradients^{i}_{t}\leftarrow$LocalTraining$(\weights_{t}^{i}; batch; loss)$\\
  $\widetilde{\gradients}_t^i\leftarrow \gradients_t^i/\max(1, |\!|\gradients_t^i|\!|_2/\Delta) + \mathcal{N}(0, \Delta^2\sigma^2)$\\
   $\hat{\gradients}_t^i\leftarrow \ks(\widetilde{\gradients}_t^i, \mathcal{N})\cdot \widetilde{\gradients}_t^i$\\
   $\overline{\gradients}_t^i\leftarrow\encoding(\signn(\hat{\gradients}_t^i))$\\
  send $[\![\overline{\gradients}_t^i]\!]$ to servers \\
  broadcast $\hash_{\delta, \phi}(\signn(\hat{\gradients}_t^i))$
\caption{\texttt{MUDGUARD}.}
\label{al:fl}
\end{algorithm}

\begin{algorithm}[!ht]
\small
\SetAlgoLined
\SetKwInput{KwInput}{Input}
\SetKwInput{KwOutput}{Output}
\DontPrintSemicolon
    \KwInput{
shares of gradients: $\{[\![\overline{\gradients}_{t}^{i}]\!]\mid i\in [n]\}$}
    \KwOutput{
    shares of aggregation $\{[\![\aggradient_{t}^{z}]\!]\mid z\in c\}$ \\
    }

\For{each $i,j\in n$}{
 $[\![dot\_product_{ij}]\!]\leftarrow[\![\overline{\gradients}_{t}^{i}]\!]\xor [\![\overline{\gradients}_{t}^{j}]\!]$\\
 convert binary sharing ($[\![dot\_product_{ij}]\!], [\![\overline{\gradients}_t^i]\!]$) to arithmetic sharing by B2A\\
 $[\![\cosm_{ij}]\!]\leftarrow1-\frac{2}{np}\sum[\![dot\_product_{ij}]\!]$
}
\For{each $i,j\in n$}{
 $[\![x_{ij}]\!]\leftarrow([\![\cosm_{i}]\!]-[\![\cosm_{j}]\!])^2$\\
 $[\![\eucm_{ij}]\!]\leftarrow1+\frac{[\![x_{ij}]\!]-1}{2}-\frac{([\![x_{ij}]\!]-1)^2}{8}+\frac{([\![x_{ij}]\!]-1)^3}{16}$\\
 \eIf{$\eucm_{ij}\leq\alpha$}{$[\![\indm_{ij}]\!]==[\![1]\!]$}{$[\![\indm_{ij}]\!]==[\![0]\!]$}
}
reconstruct $\indm$\\
each server broadcasts $\indm$\\
$l \leftarrow$ DBSCAN($\indm$)\\
\For{each $z \in c$ \textbf{all servers in parallel}}{
$[\![\aggradient_{t}^z]\!]\leftarrow\sum_{l_i= z}\decoding([\![\overline{\gradients}_{t}^i]\!]), i\in n$\\
return $[\![\aggradient_{t}^z]\!]$ to clients $\{i\mid l_i=z\}$
}
\caption{Secure clustering.}
\label{al:mpc}
\end{algorithm}

%% file: sections/preliminaries.tex
\subsection{Tools}
\label{sec:background}

\subsubsection{Federated Learning}

Federated Learning (FL) enables $n$ clients to train a global model $\weights$ collaboratively without revealing local datasets. Unlike centralized learning,  
FL requires clients to upload the weights of local models ($\{\weights^i \mid i \in n\}$) to a parametric server. 
It aims to optimize a loss function:
\begin{math}
    \mathop {\arg\min }\limits_\weights\sum\limits_{i = 1}^n {\frac{{{k_i}}}{K}{\mathcal{L}_i}(} \weights,\mathcal{D}_i ),
    \notag
\end{math}
where ${\mathcal{L}_i}(\cdot)$ and $k_i$ are the loss function and local data size of $i$-th client.
At $t$-th round, the training of FL can usually be divided into the following steps.\\
$\bullet$ \textit{Global model download}: The server selects partial clients engaging in training. All connected clients download the global model $\weights_{t-1}$ from the server.\\
$\bullet$ \textit{Local training}: Each client updates its local model by training with its own dataset: $\gradients_{t-1}^i \leftarrow \frac{\partial \mathcal{L}(\weights_{t-1}, \mathcal{D}_i)}{\partial \weights_{t-1}} $.\\
$\bullet$ \textit{Aggregation}: After the local updates $\{\gradients^i_{t-1} \mid i \in n\}$ are uploaded, the server updates the global model by aggregation: $\weights_{t} \leftarrow \weights_{t-1}- \eta\sum\limits_{i = 1}^n \frac{{{k_i}}}{K} \gradients_{t-1}^i$, where $\eta$ refers to the learning rate.

\subsubsection{DBSCAN}

Unlike traditional clustering algorithms (e.g., k-means, k-means++, bi-kmeans), which need to pre-define the number of clusters, Density-Based Spatial Clustering of Applications with Noise (DBSCAN)~\cite{ester1996density} is proposed to cluster data points dynamically.  
Based on density-based clustering, DBSCAN guarantees that clusters of any shape can always be identified.
Besides, it can recognize noise points effectively.
Basically, after setting the density parameter ($\alpha$) and the minimum cluster size ($mPts$), DBSCAN can conduct effective clustering. 
%
{We note HDBSCAN~\cite{campello2013density} could be also used for clustering. 
Its main difference from DBSCAN is the multiple densities clustering.
In this work, we assume that malicious clients may only conduct one kind of attack during the whole training, e.g., a group of malicious clients conducting a Label Flipping Attack together.
The malicious updates could only derive one density. 
Like~\cite{nguyen2022flame}, we may apply HDBSCAN in the clustering.  
But DBSCAN, in general, requires less computational complexity than HDBSCAN in algorithmic constructions. 
And further, we will conduct the clustering with cryptographic operations. 
Considering efficiency, we choose DBSCAN over HDBSCAN.} 
%
%


\subsubsection{Cryptographic Tools}

The secure Multiparty Computation (MPC) framework aims to enable multiple parties to evaluate a function over ciphertexts securely. 
The parties conducting MPC can access inputs via protection approaches, e.g., in a secret-shared format.
It does not leak any information besides the final output unless these shares are combined to derive plaintexts. 

\noindent \textbf{Secret Sharing (SS).}
It refers to a type of tool for splitting a secret among multiple parties, each of whom is assigned a share of the secret. 
The security of an SS scheme guarantees that one can distinguish shares and randoms with a negligible probability. 
Apart from that, no one can reconstruct the secret unless holding all (or a subset) of shares.
Let us consider Shamir Secret Sharing (SSS)~\cite{shamir1979share} $(t, n)$-threshold scheme as an example. Assume one chooses a polynomial $f(x)=\sum_{i=0}^{t-1}a_{i}\cdot x^i$ over $\mathbb{Z}_q$ and a secret $a_0=f(0)$. 
The secret can be split into $n$ shares by randomly selecting $n$ values: $\{r_j\leftarrow \mathbb{Z}_q^* \mid j\in n\}$, and then calculating shares $\{f(r_j)\mid j\in n\}$.
Given a subset of any $t$ out of $n$ shares, the secret can by reconstructed by Lagrange interpolation~\cite{berrut2004barycentric}:$f(0)=\sum_{j=0}^{t-1}f(r_j)\cdot \prod_{z=0, z\neq j}^{t-1}\frac{r_{z}}{r_{z}-r_{j}}$.
Except for SSS, other schemes like additive SS and replicated SS are used in the MPC framework~\cite{keller2020mp, damgaard2013practical}. 
Note these schemes have a linear property. 
Even if each party performs linear combinations locally with shares, the combined secret matches the result obtained by these linear calculations.
This saves significant communication costs in the FL context, where servers are only required to aggregate shares of gradients.\\
\noindent\textbf{Homomorphic Hash Functions (HHF).}
Given a message $x \in \mathbb{Z}_q$, a collision-resistant HHF~\cite{kim2012device} \hash: $\mathbb{G}_1 \times \mathbb{G}_2 \leftarrow \mathbb{Z}_q$ can be indicated as $\hash(x)=(g^{\hash^{'}_{\delta, \phi}(x)}, h^{\hash^{'}_{\delta, \phi}(x)})$, where $\delta$ and $\phi$ are secret keys randomly and independently selected from $\mathbb{Z}_q$. $\hash^{'}$ is a hash function, and $\mathbb{G}_1 $ and $\mathbb{G}_2$ are two different groups. 
Similar to other one-way hash functions, the security of the HHF requires that one can find a collision with a negligible probability. 
Based on additive homomorphism: $\hash(x_1+x_2) \leftarrow (g^{\hash^{'}_{\delta, \phi}(x_1) + \hash^{'}_{\delta, \phi}(x_2)}, h^{\hash^{'}_{\delta, \phi}(x_1) + \hash^{'}_{\delta, \phi}(x_2)})$, in this work, we will use this tool as a verification of the correctness of aggregation. 

\noindent\textbf{Homomorphic Encryption (HE).}
This tool is an interesting privacy-preserving technology enabling users to evaluate polynomial computations on ciphertexts without revealing underlying plaintexts. 
An encryption scheme is called partial HE if it only supports addition~\cite{paillier1999public} or multiplication~\cite{elgamal1985public}, while fully HE~\cite{gentry2009fully} can support both. 
An HE scheme usually includes the following steps. \\
$\bullet$ \textit{Key Generation:} $(\pk ,\sk)$ $\leftarrow \kgen$ ($1^\lambda$), where based on security level parameter $\lambda$, public key $\pk$ and secret key $\sk$ can be generated.\\
$\bullet$ \textit{Encryption:} $(c_1, c_2) \leftarrow \enc(\pk, m_1, m_2)$. By using $\pk$, the probabilistic algorithm encrypts messages $m_1, m_2$ to ciphertexts $c_1, c_2$.\\
$\bullet$ \textit{Homomorphic evaluation:} $\eval(c_1, c_2)=c_1\circ c_2=\enc(\pk,m_1)\circ \enc(\pk, m_2)=\enc(\pk, m_1\circ m_2)$, where $\circ$ refers to an operator, e.g., addition or multiplication. \\
$\bullet$ \textit{Decryption:} $m_1\circ m_2\leftarrow \dec(\sk, \enc(\pk, m_1\circ m_2))$. Using $\sk$, the operational results of $m_1$ and $m_2$ can be derived.

\noindent\textbf{Oblivious Transfer (OT).} OT~\cite{rabin2005exchange} is one of the crucial building blocks for MPC. 
In an OT protocol (involving two parties), a sender holds $n$ different strings $s_i, i=1\cdots n,$ and a receiver has an index ($ind$) and wants to learn $s_{ind}$. 
At the end of the protocol, the receiver cannot get information about strings rather than $s_{ind}$, while the sender learns nothing about $ind$ selected by the receiver. 
For example, 
a 1-out-of-2 OT protocol only inputs two strings and a 1-bit index. 

\noindent{\textbf{Garbled Circuits (GC)~\cite{bellare2012foundations}.}} The protocol is run between two parties called the garbler and evaluator. The garbler generates the GC corresponding to the Boolean circuit to be evaluated securely by associating two random keys per wire representing the bit values {0, 1}. The garbler then sends the GC together with the keys for the inputs to the evaluator. The evaluator obliviously obtains the keys for his inputs via OT and evaluates the circuit to obtain the output key. Finally, the evaluator maps the output key to the real output.

\subsubsection{Differential Privacy}

Differential Privacy (DP)~\cite{dwork2008differential} is a data protection approach enabling one to publish statistical information of datasets while keeping individual data private. 
The security of DP requires that adversaries cannot statistically distinguish the changes between two datasets where an arbitrary data point is different. 
The most widely used DP mechanism is called $(\epsilon, \delta)$-DP defined  below, requiring less injection noise than the $\epsilon$-DP but standing at the same privacy level. 
\begin{definition}[$(\epsilon, \delta)$ - Differential Privacy~\cite{dwork2008differential}]
 Given two real positive numbers $(\epsilon, \delta)$ and a randomized algorithm $\mathcal{A}$: $\mathcal{D}^{n} \rightarrow \mathcal{Y}$, the algorithm $\mathcal{A}$ provides $(\epsilon, \delta)$ - DP if for all data sets $\mathbf{D}, \mathbf{D}^{'} \in \mathcal{D}^{n}$ differing in only one data sample, and all  $\mathcal{S} \subseteq \mathcal{Y}$:
 $   Pr[\mathcal{A}(\mathbf{D})\in \mathcal{S}] \leq exp(\epsilon)\cdot Pr[\mathcal{A}(\mathbf{D}^{'})\in \mathcal{S}] + \delta.\notag$
\label{dp}
\end{definition}

Note that the Gaussian noise $\mathcal{N}\sim N(0, \Delta^{2}\sigma^{2})$ should be added to the output of the algorithm, where $\Delta$ is $L_{2}$ sensitivity of $\mathbf{D}$ and $\sigma$=$\sqrt{2\ln(1.25/\delta)}$~\cite{abadi2016deep}. 
The robustness of post-processing guarantees for any probabilistic/deterministic functions $\mathcal{F}$, if $\mathcal{A}$ satisfies $(\epsilon, \delta)$-DP, so does $\mathcal{F(A)}$.


%% file: sections/complexity_analysis.tex
\subsection{Complexity Analysis}\label{sec:complexityana}
We use $d$ to denote the dimension of the model. $n$, $S$ and $c$ refer to the number of clients, servers and clusters, respectively.\\
$\bullet$ \textbf{Computation cost.} Each client's computation cost can be computed as binary SS with Tiny OT $-$ $O(d)$. The server's computation cost consists of 4 parts: (1) computing pairwise XOR $-$ $O(n^2)$; (2) bit to arithmetic conversion $-$ $O(d)$; (3) multiplication for $L_2$ distance $-$ $O(n^3)$; (4) comparison with $\alpha$ $-$ $O(n^2)$; (5) calculating results of HHF based on the number of clusters $c$ $-$ $O(nc)$. Thus, the total computation complexity of each server is $O(n^{3})$.\\
$\bullet$ \textbf{Communication cost.} For a client in \textit{MUDGUARD}, the communication cost can be divided into 2 parts: (1) sending updates to $S$ servers with binary SS and Tiny OT $-$ $O(Sd)$; (2) broadcasting hash results of updates to the rest of parties $-$ $O(n+S)$. Thus, we have communication complexity $-$ $O(Sd+n)$ for each client. The servers' communication costs include (1) receiving correlated randomness and doubly-authenticated bits for converting a boolean shared matrix to arithmetic one $-$ $O(n^{2})$; (2) receiving triples for multiplications $-$ $O(n^{2})$; (3) receiving correlated randomness for element-wise comparison $-$ $O(n^2)$; (4) sending shares and a random bit to other servers for reconstruction $-$ $O(Sn^2)$; (5) sending aggregated shares and values of HHF to all clients $-$ $O(nd)$. Overall, the communication cost for every server is $O(Sn^2)$. For detailed experimental results, refer to Section~\ref{sec:overheads}.

%% file: sections/security_analysis.tex
\subsection{Security Analysis}
\label{sec:analysis}

We first define the ideal functionality $ \mathcal{F}_\textsf{MUDGUARD} $ to execute a byzantine-robust privacy-preserving FL, and then show that the proposed protocol $ \Pi_\textsf{MUDGUARD} $ securely realizes the functionality. 
Our security is based on the random oracle model where the homomorphic hash function outputs a uniformly random value for a new query and the same value for a previously answered query, hence we prove the UC security in $ \mathcal{F}_\textsf{RO} $-hybrid model. 
Besides, the security is also based on the existence of a secret sharing protocol where the clients derive shares indistinguishable with randoms which securely realizes $ \mathcal{F}_\textsf{SS} $ (a combination of $ \mathcal{F}_\textsf{triples}, \mathcal{F}_\textsf{share}, \mathcal{F}_\textsf{reconst} $~\cite{furukawa2017high}), and a bit-to-arithmetic conversion protocol that securely realizes $ \mathcal{F}_\textsf{B2A} $ (noted as $ \mathcal{F}_\textsf{PREP} $ in~\cite{cryptoeprint:2019:207}). 

\begin{framed}
\label{fshare}
\begin{center}
	\textbf{Ideal Functionality $ \mathcal{F}_\textsf{share} $} 
\end{center}
The functionality $ \mathcal{F}_\textsf{share} $ interacts with a dealer party $ \textsf{P}_j $, and a corrupted party $ \textsf{P}_i $. 

Upon receiving $ (t_i, s_i) $ from the corrupted party $ \textsf{P}_i $, and receiving $ v $ from the dealer $ \textsf{P}_j $, the functionality $ \mathcal{F}_\textsf{share} $ computes $ (t_{j+1}, s_{j+1}) $ and $ (t_{j+2}, s_{j+2}) $ from $ (t_i, s_i) $ and $ v $, and sends the honest $ \textsf{P}_{i-1} $ and $ \textsf{P}_{i+1} $ their respective shares. 

\end{framed}

\begin{framed}
\label{freconst}
\begin{center}
	\textbf{Ideal Functionality $ \mathcal{F}_\textsf{reconst} $} 
\end{center}

The functionality $ \mathcal{F}_\textsf{reconst} $ interacts with an adversary \textsf{Sim} and a corrupted party $ \textsf{P}_i $, and receives information from $ \textsf{P}_{i+1} $ and $ \textsf{P}_{i+2} $. 

Upon receiving $ (t_{i+1}, s_{i+1}, j) $ from $ \textsf{P}_{i+1} $ and $ (t_{i+2}, s_{i+2}, j) $ from $ \textsf{P}_{i+2} $, $ \mathcal{F}_\textsf{reconst} $ computes $ v=s_{i+2}\oplus t_{i+1} $ and sends $ v $ to $ \textsf{P}_j $. 
In addition, the functionality $ \mathcal{F}_\textsf{reconst} $ sends $ (t_i, s_i) $ to the adversary \textsf{Sim}, where $ (t_i, s_i) $ is $ \textsf{P}_i $'s share as defined by the shares received from the honest parties. 

\end{framed}

\begin{framed}
\label{ftriples}

\begin{center}
	\textbf{Ideal Functionality $ \mathcal{F}_\textsf{triples} $} 
\end{center}

The functionality $ \mathcal{F}_\textsf{triples} $ interacts with a corrupted party $ \textsf{P}_i $, and receive information from $ \textsf{P}_1, \textsf{P}_2, \textsf{P}_3 $. 

Upon receiving $ N $ triples of pairs $ \{(t_{a_i}^j, s_{a_i}^j), (t_{b_i}^j, s_{b_i}^j), (t_{c_i}^j, s_{c_i}^j)\}_{j=1}^N $ from $ \textsf{P}_i $, the functionality first $ \mathcal{F}_\textsf{triples} $ chooses random $ a_j,b_j\in\{0, 1\} $ and computes $ a_j b_j $, and then defines a vector of sharings $ \textbf{\textit{d}}=([a_j],[b_j],[c_j]) $, for $ j=1, ..., N $. 
The sharings are computed from $ [(t_{a_i}^j, s_{a_i}^j), (t_{b_i}^j, s_{b_i}^j), (t_{c_i}^j, s_{c_i}^j)] $ provided by $ \textsf{P}_i $ and the chosen $ a_j, b_j, c_j $. 
Next, $ \mathcal{F}_\textsf{triples} $ sends the generated shares to each corresponding party. 
\end{framed}

\begin{framed}
\label{fprep}
\begin{center}
	\textbf{Ideal Functionality $ \mathcal{F}_\textsf{Prep}$ ($ \mathcal{F}_\textsf{B2A} $)} 
\end{center}

Independent copies of $ \mathcal{F}_\textsf{MPC} $ are identified via session identifiers \textsf{sid}. 
For each instance, $ \mathcal{F}_\textsf{Prep} $ maintains a dictionary $ \textsf{Dic}_\textsf{sid} $. 
If a party provides input with an invalid \textsf{sid}, the $ \mathcal{F}_\textsf{Prep} $ outputs \textbf{reject} to all parties and await another message. 

Upon receiving (\textbf{Init}, $ \mathbb{F} $, \textsf{sid}) from all parties, initialize a new database of secrets $ \textsf{Dic}_\textsf{sid} $ indexed by a set $ \textsf{Dic}_\textsf{sid}.\textsf{Keys} $ and store the field $ \mathbb{F}  $as $ \textsf{Dic}_\textsf{sid}.\textsf{Field} $,  if \textsf{sid} is a new session identifier. 
Set the flag $ \textsf{Abort}_\textsf{sid}=\textsc{false} $. 

Upon receiving (\textbf{Input}, $ i $, \textsf{id}, $ x $, \textsf{sid}) from a party $ \textsf{P}_i $ and (\textbf{Input}, $ i $, \textsf{id}, $ \perp $, \textsf{sid}) from all other parties, if $ \textsf{id}\notin \textsf{Dic}_\textsf{sid}.\textsf{Keys} $ then insert it and set $ \textsf{Dic}_\textsf{sid}[\textsf{id}]=x $. 
Then execute the procedure \textbf{Wait}. 

Upon receiving $ (\textbf{Add}, \textsf{id}_x, \textsf{id}_y, \textsf{id}, \textsf{sid}) $, set $ \textsf{Dic}_\textsf{sid}[\textsf{id}]=\textsf{Dic}_\textsf{sid}[\textsf{id}_x]+\textsf{Dic}_\textsf{sid}[\textsf{id}_y] $ if $ \textsf{id}_x, \textsf{id}_y\in \textsf{Dic}_\textsf{sid}.\textsf{Keys} $. 

Upon receiving $ (\textbf{Mult}, \textsf{id}_x, \textsf{id}_y, \textsf{id}, \textsf{sid}) $, set $ \textsf{Dic}_\textsf{sid}[\textsf{id}]=\textsf{Dic}_\textsf{sid}[\textsf{id}_x]\cdot \textsf{Dic}_\textsf{sid}[\textsf{id}_y] $, if $ \textsf{id}_x, \textsf{id}_y \in \textsf{Dic}_\textsf{sid}.\textsf{Keys} $. 
Then execute the procedure \textbf{Wait}. 

Upon receiving $ (\textbf{RanEle}, \textsf{id}, \textsf{sid}) $, set $ \textsf{Dic}_\textsf{sid}[\textsf{id}] $ to a random element in $ \textsf{Dic}_\textsf{sid}.\textsf{Field} $, if $ \textsf{id}\notin \textsf{Dic}_\textsf{sid}.\textsf{Keys} $. 
Then execute the procedure \textbf{Wait}. 

Upon receiving $ (\textbf{RanBit}, \textsf{id}, \textsf{sid}) $, set $ \textsf{Dic}_\textsf{sid}[\textsf{id}] $ to a random bit if $ \textsf{id}\notin \textsf{Dic}_\textsf{sid}.\textsf{Keys} $. 
Then execute the procedure \textbf{Wait}. 

Upon receiving $ (\textbf{Open}, $ i $, \textsf{id}, \textsf{sid}) $ from all parties, if $ \textsf{id}\in \textsf{Dic}_\textsf{sid}.\textsf{Keys} $: 
1) if $ i=0 $, send $ \textsf{Dic}_\textsf{sid}[\textsf{id}] $ to the adversary and executes \textbf{Wait}. If the answer is $ (\textbf{OK}, \textsf{sid}) $, await an error $ \epsilon $ from the adversary. Send $ \textsf{Dic}_\textsf{sid}[\textsf{id}]+\epsilon $ to all honest parties. 
If $ \epsilon\neq 0 $, set the flag $ \textbf{Abort}_\textsf{sid}=\textsc{true} $. 
2) if $ i\in A $, then send $ \textsf{Dic}_\textsf{sid}[\textsf{id}] $ to the adversary. 
Then execute \textbf{Wait}. 
3) if $ i\in[n]\backslash A $, execute \textbf{Wait}. If not already halted, then await an error $ \epsilon $ from the adversary. 
Send $ \textsf{Dic}_\textsf{sid}[\textsf{id}]+\epsilon $ to party $ \textsf{P}_i $. 
If $ \epsilon\neq 0 $, set the flag $ \textbf{Abort}_\textsf{sid}=\textsc{true} $. 

Upon receiving $ (\textbf{Check}, \textsf{sid}) $ from all parties, execute the procedure \textbf{Wait}. If not already halted and $ \textbf{Abort}_\textsf{sid}=\textsc{true} $, send $ (\textbf{Abort}, \textsf{sid}) $ to the adversary and all honest parties, and ignore further messages to $ \mathcal{F}_\textsf{MPC} $ with the same \textsf{sid}. 
Otherwise, send $ (\textbf{OK}, \textsf{sid}) $ and continue. 

Upon receiving $ (\textbf{daBits}, \textsf{id}_1, ..., \textsf{id}_l, \textsf{sid}_1, \textsf{sid}_2) $ from all parties where $ \textsf{id}_i\notin \textsf{Dic}_\textsf{sid}.\textsf{Keys} $ for all $ i\in[l] $, await a message \textbf{OK} or \textbf{Abort} from the adversary. If \textbf{OK} is received, sample a set of random bit $ \{b_j\}_j\in[l] $, and for each $ j\in[l] $ set $ \textsf{Dic}_{\textsf{sid}_1}[\textsf{id}_j]=b_j $ and $ \textsf{Dic}_{\textsf{sid}_2}[\textsf{id}_j]=b_j $, and insert the set $ \{\textsf{id}_i\}_{i\in[l]} $ into $ \textsf{Dic}_{\textsf{sid}_1}.\textsf{Keys} $ and $ \textsf{Dic}_{\textsf{sid}_2}.\textsf{Keys} $. 
Otherwise, send $ (\textbf{Abort}, \textsf{sid}_1) $ and $ (\textbf{Abort}, \textsf{sid}_2) $ to the adversary and all honest parties, and ignore all further messages to $ \mathcal{F}_\textsf{MPC}) $ with the same $ \textsf{sid}_1 $ and $ \textsf{sid}_2 $. 

Procedure \textbf{Wait}: 
Await a message $ (\textbf{OK}, \textsf{sid}) $ or $ (\textbf{Abort}, \textsf{sid}) $ from the adversary. 
If \textbf{OK} is received, then continue. 
Otherwise, send $ (\textbf{Abort}, \textsf{sid}) $ to all honest parties, and ignore all further messages to $ \mathcal{F}_\textsf{MPC} $ with the same $ \textsf{sid} $. 
	
\end{framed}

\begin{remark}
	We use $ \mathcal{F}_\textsf{share} $ as the secret share generating algorithm, $ \mathcal{F}_\textsf{reconst} $ as the reconstructing algorithm, $ \mathcal{F}_\textsf{triples} $ as the secret share multiplication algorithm, and $ \mathcal{F}_\textsf{B2A} $ (see $\mathcal{F}_\textsf{PREP}$ in~\cite{cryptoeprint:2019:207}) as the bit to arithmetic conversion algorithm. 
\end{remark}

We formally define $ \mathcal{F}_\textsf{\textit{MUDGUARD}} $ as follows. 

\begin{framed}
	\label{fMUDGUARD}
	\begin{center}
		\textbf{Ideal Functionality $ \mathcal{F}_\textsf{\textit{MUDGUARD}} $}
	\end{center}
	The functionality $ \mathcal{F}_\textsf{\textit{MUDGUARD}} $ is parameterized with a DBSCAN algorithm with corresponding parameters, a local training SGD algorithm with appropriate variables, Gauss noise parameters $ \Delta $ and $ \sigma $, and the density parameter $ \alpha $.  
	The functionality $ \mathcal{F}_\textsf{\textit{MUDGUARD}} $ interacts with $ n $ clients $ \textsf{P}_1, ..., \textsf{P}_n $, $ s $ remote servers $ \textsf{S}_1, ..., \textsf{S}_s $, and an ideal adversary \textsf{Sim}. 
	
	Upon receiving $ (\textbf{Init}, \{w_0^i\}_{i\in [n]}, \{\aggradient_0^i\}_{i\in [n]}) $ from the adversary \textsf{Sim}, send $ (\textbf{Init}, w_0^i, \aggradient_0^i) $ to each $ \textsf{P}_i $. 
	
	Upon receiving $ (\textbf{Update}, t, w_{t-1}^{i}, \mathcal{D}_i) $ from each honest client $ \textsf{P}_i $, calculate $ w_{t}^i $, $ \gradients_t^i $, $ \tilde{\gradients}_t^i $, $ \hat{\gradients}_t^i $, $ \bar{\gradients}_t^i $, and store $ (t, [\![\bar{\gradients}_t^i]\!]) $ for each server, and notify \textsf{Sim} with $ (\textbf{Update}, t, \textsf{P}_i) $.
	If $ t = T $, terminate the protocol. 
	Later, when \textsf{Sim} replies with $ (\textbf{Update-data}, t) $, send $ (\textbf{Update}, t) $ to each server $ \textsf{S}_j $ for each $ j\in[s] $. Upon receiving $ (\textbf{Update}, t, \{[\![\bar{\gradients'}_t^i]\!]\}_{i\in\mathcal{I}}) $ from \textsf{Sim} for all corrupted client index $ i $, where $ \mathcal{I}\subset [n] $, store $ (t, [\![\bar{\gradients'}_t^i]\!]) $ for each honest server. 
	
	%
	%
	
	Upon receiving $ (\textbf{Update}, t) $ from a server $ \textsf{S}_j $, if $ (t, [\![\bar{\gradients}_t^i]\!]) $ is stored for each $ i\in[n] $ and for each server, calculate \textsf{IndM} and $ \{[\![\aggradient_{t}^z]\!]\}_{z\in[c]} $ for each server. 
	Then send $ (\textbf{Model}, t, \textsf{IndM}, \{[\![\aggradient_{t}^z]\!]\}_{z\in[c]}) $ to each server. 
	If $ \textsf{S}_j $ is honest, upon receiving $ (\textbf{Update-model}, t) $ from the simulator \textsf{Sim}, send $ (\textbf{Update-model}, t, [\![\aggradient_{t}^z]\!]) $ to the corresponding client. 
	Otherwise, upon receiving $ (\textbf{Update-model}, t, \{[\![\aggradient_{t}^{'z}]\!]\}) $ from the simulator \textsf{Sim}, send $ (\textbf{Update-model}, t, [\![\aggradient_{t}^{'z}]\!]) $ to corresponding client. 
	
	Upon receiving $ (\textbf{Abort}) $ from either the adversary or any client, send $ \perp $ to all parties and terminate. 
	
	
\end{framed}

\begin{remark}
	According to the ideal functionality, we could capture not only privacy but also soundness against malicious corruption of servers. 
	However, the differential attack is based on the output of each epoch, which is published roundly and could be obtained legally. 
	Hence, the discussion on differential privacy is not included in this security definition. 
	Detailed proof for differential privacy will be given later. 
	Moreover, soundness against malicious corruption of clients is also not captured by the previous definition since such security is protected by the clustering technique, which is not the concern of cryptography.   
\end{remark}

\begin{definition} [Universally Composable security]
	A protocol $ \Pi $ $ \textsf{UC-realizes} $ ideal functionality $ \mathcal{F} $ if for any PPT adversary $ \mathcal{A} $ there exists a PPT simulator $ \mathcal{S} $ such that, for any PPT environment $ \mathcal{E} $, the ensembles $ \mathsf{EXEC}_{\Pi, \mathcal{A}, \mathcal{E}} $ and $ \mathsf{EXEC}_{\mathsf{IDEAL}_\mathcal{F}, \mathcal{S}, \mathcal{E}} $ are indistinguishable. 
\end{definition}

\begin{definition} [UC security of \texttt{MUDGUARD}]
	A protocol $ \Pi_\textsf{\textit{MUDGUARD}} $ is UC-secure if $ \Pi_\textsf{\textit{MUDGUARD}} $ UC-realizes $ \mathcal{F} $, against malicious-majority clients and malicious-minority servers, considering arbitrary collusion between malicious parties. 
\end{definition}

\begin{theorem} [UC security of \texttt{MUDGUARD}]
	Suppose the existence of a homomorphic hash function in a random oracle model, our protocol is UC-secure in $ (\mathcal{F}_\textsf{RO}, \mathcal{F}_\textsf{SS}, \mathcal{F}_\textsf{B2A}) $-hybrid world. 
\end{theorem}

\begin{proof}We show the validity of the theorem by proving that the protocol $ \Pi_\textsf{\textit{MUDGUARD}} $ securely realizes $ \mathcal{F} $ in the $ (\mathcal{F}_\textsf{RO}, \mathcal{F}_\textsf{SS}, \mathcal{F}_\textsf{B2A}) $-hybrid world against any corruption pattern. 
	We construct a simulator $ \textsf{Sim} $ for any non-uniform PPT environment $ \mathcal{E} $ such that $ \textsf{EXEC}_{\Pi_\textsf{\textit{MUDGUARD}}, \mathcal{A}, \mathcal{E}}^{\mathcal{F}_\textsf{RO}, \mathcal{F}_\textsf{SS}, \mathcal{F}_\textsf{B2A}}\approx\textsf{EXEC}_{\mathcal{F}_\textsf{\textit{MUDGUARD}}, \textsf{Sim}, \mathcal{Z}} $. 
	The \textsf{Sim} is constructed as follows. 
	
	It writes on $ \mathcal{A} $'s input tape upon receiving an input value from $ \mathcal{E} $, as if coming from $\mathcal{E}$, and writes on $ \mathcal{Z} $'s output tape upon receiving an output value from $ \mathcal{A} $, as if from $ \mathcal{A} $. 
	
	\noindent\textit{Case 1: If all clients and servers are not corrupted. }
	Since we assume private channels between client and server, \textsf{Sim} could just simply randomly choose all intermediate values. 
	There is no distinguisher who could tell the difference between random values and real transcripts. 
	
	\noindent\textit{Case 2: If corrupted clients exist. }
	We note the corrupted subset as $ \mathcal{I}\subset [n] $. 
	We need to simulate the adversary's view, which is the secret share and its corresponding hash value. 
	For $ t\in[T] $ and $ i\in \mathcal{I} $, the simulator randomly chosen $ \bar{\gradients'}_t^i\leftarrow \{0, 1\} $, and internally executes $ \mathcal{F}_\textsf{SS}$ to obtain $ [\![\bar{\gradients'}_t^i]\!] $. 
	Then, \textsf{Sim} internally executes $ \mathcal{F}_\textsf{RO} $ to obtain $ \hash_t^i $. 
	Because of the UC security of $ \mathcal{F}_\textsf{SS} $ and $ \mathcal{F}_\textsf{RO} $, there is no distinguisher that could tell the difference between $ ([\![\bar{\gradients'}_t^i]\!], \hash_t^i) $ and $ ([\![\bar{\gradients}_t^i]\!], \textsf{H}_{\delta, \phi}(\textsf{sign}(\hat{g}_t^i))) $. 
	
	\noindent\textit{Case 3: If corrupted servers exist.  }
	We not the corrupted subset as $ \mathcal{I}\subset [T] $. 
	We need to simulate the adversary's view, including all secret shares in the protocol. 
	It is worth noticing that we should not only guarantee the indistinguishability between two groups of shares but also the relationship among elements within each group. 
	After obtaining \textsf{IndM}, the \textsf{Sim} executes DBSCAN protocol on \textsf{IndM} and acquires cluster labels $ l $, and executes the functionality $ \mathcal{F}_\textsf{SS} $ to obtain $ [\![\textsf{IndM}_{ij}]\!] $ for each $ i,j\in [n] $. 
	Then, \textsf{Sim} randomly chosen $ |l| $ secret sharing values $ [\![\bar{\gradients'}_t^i]\!] $ such that the summation $ \Sigma_{l_i=z}\textsf{DCD}([\![\bar{\gradients'}_t^i]\!]) $ equals to the given share $ [\![\aggradient_{t}^z]\!] $. 
	This procedure could be easily achieved by first randomly choosing the first $ |l|-1 $ values and then calculating the last value. 
	For each $ i\notin c $, the simulator \textsf{Sim} simply chooses the shares of gradients randomly since those values are irrelevant to the calculation. 
	After acquiring all the shares of gradients $ [\![\bar{\gradients'}_t^i]\!] $, the simulation \textsf{Sim} pairwisely calculate $ [\![dot\_product'_{ij}]\!] $, and convert it to arithmetic sharing by executing the functionality $ \mathcal{F}_\textsf{B2A} $, and then calculate the adjusted cosine similarity matrix share $ [\![\textsf{CosM}'_{ij}]\!] $. 
	Next, as  in Algorithm 2, \textsf{Sim} calculates $ [\![x'_{ij}]\!], [\![\textsf{EucM}'_{ij}]\!] $ for each $ i, j\in[n] $. 
	We claim that all shares that were previously generated are interdeducible, except between $ [\![\textsf{EucM}'_{ij}]\!] $ and $ [\![\textsf{IndM}'_{ij}]\!] $, since the latter two are the input/output pair of the element-wise comparison algorithm computed by a secure comparison algorithm in our protocol. 
	Fortunately, the privacy of a secure comparison algorithm guarantees the indistinguishability between real and ideal input/output pairs. 
	Hence, we claim that if there exists a distinguisher that could tell the difference between the real and ideal world, it contradicts either the UC security of $ \mathcal{F}_\textsf{SS} $ or the privacy of secure comparison protocol. 
	
	\noindent\textit{Case 4: If there exists both corrupted clients and servers. }
	The situation, in this case, is simply the combination of Case 2 and 3 since there is no extra view needed to simulate. 
	
	In summary, for any PPT adversary $ \mathcal{A} $ we could construct a $ \textsf{Sim} $, so that for any PPT environment $ \mathcal{E} $, the $ \mathsf{EXEC}_{\Pi, \mathcal{A}, \mathcal{E}} $ and $ \mathsf{EXEC}_{{\mathcal{F}_\textsf{\textit{MUDGUARD}}}, \mathcal{S}, \mathcal{E}} $ are indistinguishable. 
\end{proof}


UC framework captures attacks on input and intermediate data. 
On the contrary, differential privacy prevents the adversary from inferring about private information from outputs or updates, and such information might also be utilized by malicious clients. 
When false positives clustering exists, or malicious clients pretend to be honest, local updates have a chance to be revealed to the adversary. 
The following theorem shows that these updates do not leak any individual data due to differential privacy. 
\begin{theorem}
	No adversary in corrupted client set $\adv^c \subset \mathcal{C}$, where $|\adv^c|\leq n-1$, can retrieve the individual values of honest clients.
\end{theorem}

\begin{proof}
	Since we apply differential privacy~\cite{dwork2008differential}, the local updates cannot leak information regarding the inputs. According to Def.~\ref{dp}, the added differentially private noise guarantees that the aggregation is indistinguishable whether an individual update participates or not.
	Therefore, it guarantees the security of individual local updates while aggregation can be calculated. 
\end{proof}

%% file: sections/converge_analysis.tex
\subsection{Convergence Analysis}
\label{sec:convana}

Let $M$ be the total number of clients in a semi-honest majority client cluster. Semi-honest clients and malicious clients are indexed by $\{1,\cdots,h\}$ and $\{h+1,\cdots,h+m\}$, respectively, where $M=h+m$ and $h>m$ if TNR is greater than 50\%. The component $j$ of stochastic gradient and of true gradient are denoted as $\{\tilde{g}_{i,j}\}_{i=1}^M$ and $g_j$ respectively. An error probability is shown as follows.

\begin{lemma}[The bound of error probability with malicious clients]
\label{lm:ErrorProbBound}
If the TNR (h/M) of the clustering is relatively high, then we have  the error probability
    $\mathbb{P}\left[\sign\left[\sum_{i=1}^M \sign(\tilde{g}_{i,j})\right]\not=\sign(g_j)\right]\leq \mathbb{P}_{h}\cdot \mathcal{O}(\sqrt{M/h})$,
where $\mathbb{P}_{h}$ is the bound for the error probability without malicious clients.
\end{lemma}
\begin{proof}
Every client is a Bernoulli trial with success probability $p_h$ for semi-honest clients and $p_m$ for malicious clients, respectively, to receive the true gradient signs. Let $Z_h$ be the number of semi-honest clients with true signs, which therefore equals the sum of $h$ independent Bernoulli trials, so we know $Z_h$ follows the binomial distribution $B(h,p_h)$. Similarly, we know the number of malicious clients with correct signs $Z_m$ follows the binomial distribution $B(m,p_m)$. Denote $q_h=1-p_h$ and $q_m=1-p_m$.

Let $Z$ be the total number of clients with true gradient signs, so $Z=Z_h+Z_m$. We use the Gaussian distribution to simplify the analysis. Notice that $B(h,p_h)\sim N(hp_h,hp_hq_h)$ and $B(m,p_m)\sim N(mp_m,mp_mq_m)$, so we get 
$Z\sim N\left(hp_h+mp_m, hp_hq_h+mp_mq_m\right).$
The event $\sign\left[\sum_{i=1}^M \sign(\tilde{g}_{i,j})\right]\not=\sign(g_j)$ is equivalent to event $Z \leq M/2$.
Then the error probability equals $\mathbb{P}(Z\leq M/2)$. By using Cantelli's inequality, we know 
\begin{equation}
\begin{aligned}
    &\mathbb{P}\left[Z\leq M/2 \right] = 
    \mathbb{P}[Z\geq 2(hp_h+mp_m)-M/2]\\
    &=\mathbb{P}[Z-(hp_h+mp_m)\geq (hp_h+mp_m)-M/2]\\
    &\leq 
    \frac{1}{1+\frac{[(hp_h+mp_m)-M/2]^2}{hp_hq_h+mp_mq_m}}
    \leq
    \frac{\sqrt{hp_hq_h+mp_mq_m}}{2|(hp_h+mp_m)-M/2|}\\
    &=
    \frac{\sqrt{Mp_hq_h}}{2M(p_h-1/2)}
     \cdot\frac{\sqrt{h/M+m/M\cdot p_mq_m/p_hq_h}}
    {(h/M\cdot p_h + m/M\cdot p_m-1)/(p_h-1/2)}\\
    &=\mathbb{P}_{h}\cdot\mathcal{O}(\sqrt{M/h})
    \end{aligned}
\label{eq:errorProbBound}
\end{equation}

where the second inequality holds since $\frac{1}{x^2+1}\leq \frac{1}{2x}$ for $x>0$, and the last two equalities hold since $p_h>1/2$ by~\cite{bernstein2018signsgd} and we assume $hp_h>M/2$ with overwhelming probability for a sufficient large TNR. The assumption is reasonable because $hp_h=Mp_h>M/2$ if $TNR=100\%$. The first factor in Eq.~\eqref{eq:errorProbBound} is the bound for the error probability without malicious clients, so we get the error probability less than a $ \mathcal{O}(\sqrt{M/h})$ factor of that in the case of without malicious clients.
\end{proof}

Let $\textbf{\textit{L}}$ and $\bm{\sigma}$ be non-negative losses and standard deviation of stochastic gradients $\tilde{g}$ respectively. 
$\forall x$, the objective values $f(x)$ are bounded by constants $f_*$ (i.e. $f(x)\geq f_*$). The objective value of 0-$th$ round is referred to as $f_0$.  Under the above conditions, the results are the following.

\begin{theorem}[Non-convex convergence rate of \texttt{MUDGUARD}]
     If the TNR of the clustering is relatively high, then the global model generated in the semi-honest cluster converges at a rate
\begin{equation}
\begin{aligned}
    \mathbb{E}\left[\frac{1}{T} \sum_{t=0}^{T-1}\left\|g_{t}\right\|_{1}\right]^{2} 
\leq &\frac{1}{\sqrt{N}}\Bigg[\sqrt{\|\textbf{L}\|_{1}}\left(f_{0}-f_{*}+\frac{1}{2}\right)\\
&+\frac{2}{\mathcal{O}(\sqrt{h})}\|\bm{\sigma}\|_{1}\Bigg]^2,
\notag
\end{aligned}
\end{equation}
where $N$ is the cumulative number of stochastic gradient calls up to round T $(i.e., N=\mathcal{O}(T^2))$.
Therefore, the higher the rate is, the closer the convergence speed is to the case without malicious clients.
\end{theorem}
\begin{proof}
    Following the results of Theorem 2 in~\cite{bernstein2018signsgd}, in the distributed SignSGD with a majority vote, we can get the non-convex convergence rate without malicious clients at 
\begin{equation}
    \begin{aligned}
            \mathbb{E}\left[\frac{1}{T} \sum_{t=0}^{T-1}\left\|g_{t}\right\|_{1}\right]^{2} 
\leq &\frac{1}{\sqrt{N}}\Bigg[\sqrt{\|\textbf{L}\|_{1}}\left(f_{0}-f_{*}+\frac{1}{2}\right)\\
&+\frac{2}{\sqrt{M}}\|\bm{\sigma}\|_{1}\Bigg]^2
    \end{aligned}
    \label{eq3}
\end{equation}
from 
\begin{equation}
    \begin{aligned}
            |g_t|\mathbb{P}\left[\sign\left[\sum_{i=1}^M \sign(\tilde{g}_{i,j})\right]\not=\sign(g_j)\right]\leq \frac{\sigma_i}{\sqrt{M}}.
    \end{aligned}
    \label{eq4}
\end{equation}
As Lemma~\ref{lm:ErrorProbBound} proved, in the existence of the malicious clients, we get~(\ref{eq4})$\leq\frac{\sigma_i}{\sqrt{M}}\cdot\mathcal{O}(\sqrt{M/h})=\frac{\sigma_i}{\mathcal{O}(\sqrt{h})}$. By plugging the result into~(\ref{eq3}), we have the convergence rate of \texttt{MUDGUARD}:  
\begin{equation}
\begin{aligned}
    \mathbb{E}\left[\frac{1}{T} \sum_{t=0}^{T-1}\left\|g_{t}\right\|_{1}\right]^{2} 
\leq &\frac{1}{\sqrt{N}}\Bigg[\sqrt{\|\textbf{L}\|_{1}}\left(f_{0}-f_{*}+\frac{1}{2}\right)\\
&+\frac{2}{\sqrt{M}}\|\bm{\sigma}\|_{1}\cdot\mathcal{O}(\sqrt{M/h})\Bigg]^2\\
&=\frac{1}{\sqrt{N}}\Bigg[\sqrt{\|\textbf{L}\|_{1}}\left(f_{0}-f_{*}+\frac{1}{2}\right)\\
&+\frac{2}{\mathcal{O}(\sqrt{h})}\|\bm{\sigma}\|_{1}\Bigg]^2.
\end{aligned}
\end{equation}
\end{proof}

\subsection{Proof of Theorem 1}
\label{sec:pot1}
%
\begin{proof}
We denote $a=(a_1,\cdots, a_{d})$ and $b=(b_1,\cdots,b_{d})$ are two vectors uploaded by malicious clients, where $np$ refers to the number of model parameters:
\begin{equation}
    Pr(a_i, b_i=\signn)=
    \begin{cases}
    \frac{1}{2},& \signn=+1\\
    \frac{1}{2},& \signn=-1
    \end{cases},
    \ \forall i \in [d].
    \notag
\end{equation}

The adjusted cosine similarity can be computed as:
\begin{equation}
\begin{aligned}
        COS\_similarity&= \frac{a_1b_1+\cdots+ a_{d}b_{d}}{\sqrt{a_1^2+\cdots+a_{d}^2}\cdot \sqrt{b_1^2+\cdots+b_{d}^2}}\\
    &=\frac{a_1b_1+\cdots+ a_{d}b_{d}}{d}.
\end{aligned}\notag
\end{equation}
Since $a_i$ and $b_i$ are relatively independent, we have:
\begin{equation}
    Pr(a_i\cdot b_i=\signn)=
    \begin{cases}
    \frac{1}{2},& \signn=+1\\
    \frac{1}{2},& \signn=-1
    \end{cases},
    \ \forall i \in [d].
    \notag
\end{equation}

According to the \textit{Law of large numbers}, $E(COS\_similarity)\sim E(a_ib_i)=0.$ This conclusion can be generalized to any two malicious clients, and malicious clients have the same distance as a semi-honest client. Therefore, if we calculate the adjusted cosine similarity vector of two malicious clients, there should be only two elements of difference. The $L_2$ distance of these two vectors is $\sqrt{2}$. 
\end{proof}

%% file: sections/other_experiment_results.tex
\subsection{Other Experimental Setup and Experimental Results}

\subsubsection{Other Experiment Setup}\label{sec:otherexpsetup}
We implement \texttt{MUDGUARD} in C++ and Python. We use MP-SPDZ library~\cite{keller2020mp} to implement secure computations and Pytorch framework~\cite{paszke2019pytorch} for training.
All the experiments are conducted on a cluster of machines with Intel(R) Xeon(R) CPU E5-2620 v4 @ 2.10GHz and NVIDIA 1080 Ti GPU, with 32GB RAM in a local area network. 
As for the cryptographic tools, all the parameters are set to a 128-bit security level.\\ 
\noindent\textbf{Datasets.} We use MNIST and FMNIST datasets for the image classification task.\\
$\bullet$ \textbf{MNIST~\cite{lecun-mnisthandwrittendigit-2010}.} It consists of 60,000 training samples and 10,000 testing samples, where each sample is a 28$\times$28 gray-scale image of handwritten digital (0-9).\\
$\bullet$ \textbf{FMNIST~\cite{xiao2017fashion}.} It contains article images from Zalando and has the same size as MNIST, where each image is a 28$\times$28 gray-scale image associated with a label from 10 classes.\\
$\bullet$ \textbf{CIFAR-10~\cite{krizhevsky2009learning}.} It offers 50,000 training samples and 10,000 test samples, where each is a 32$\times$32 color image in a label from 10 different objectives, and there are 6,000 images for each class.\\
\noindent\textbf{Classifiers.} We use LeNet and ResNet-18 to perform training and classification of the datasets.\\
$\bullet$ \textbf{LeNet~\cite{lecun1989backpropagation}.} Containing 6 layers (including 3 convolution layers, 2 pooling layers, and 1 fully connected layer), LeNet aims to train 44,426 parameters for image classification. \\
$\bullet$ \textbf{ResNet-18~\cite{He_2016_CVPR}.} It provides 18 layers with 11 million trainable parameters to train color images. We use a light vision of ResNet with approx. 2.07 million parameters and complete the experiments with the CIFAR-10 dataset.\\

\subsubsection{Other Experimental Results}\label{sec:otherexp}
\begin{table*}[]
\centering
\begin{tabular}{@{}cc|ccccccc@{}}
\toprule
\multicolumn{2}{c|}{Attacks}                   & Baseline & GA    & LFA   & Krum  & Trim  & AA          & BA          \\ \midrule
\multicolumn{1}{c|}{\multirow{5}{*}{$\xi$}} & 0.5 & 0.811    & 0.803 & 0.772 & 0.751 & 0.763 & 0.793 / 0     & 0.797 / 0     \\
\multicolumn{1}{c|}{}                    & 0.6 & 0.783    & 0.772 & 0.77  & 0.761 & 0.757 & 0.784 / 0.002 & 0.801 / 0     \\
\multicolumn{1}{c|}{}                    & 0.7 & 0.769    & 0.767 & 0.747 & 0.723 & 0.726 & 0.776 / 0.003 & 0.777 / 0.005 \\
\multicolumn{1}{c|}{}                    & 0.8 & 0.754    & 0.752 & 0.731 & 0.743 & 0.754 & 0.771 / 0.001 & 0.758 / 0.001 \\
\multicolumn{1}{c|}{}                    & 0.9 & 0.755    & 0.737 & 0.73  & 0.718 & 0.724 & 0.753 / 0.002 & 0.731 / 0.008 \\ \midrule
\multicolumn{1}{c|}{\multirow{5}{*}{n}}  & 10  & 0.846    & 0.844 & 0.824 & 0.829 & 0.831 & 0.836 / 0     & 0.845 / 0     \\
\multicolumn{1}{c|}{}                    & 50  & 0.834    & 0.829 & 0.829 & 0.829 & 0.827 & 0.827 / 0     & 0.836 / 0     \\
\multicolumn{1}{c|}{}                    & 100 & 0.783    & 0.772 & 0.77  & 0.761 & 0.757 & 0.784 / 0.002 & 0.801 / 0     \\
\multicolumn{1}{c|}{}                    & 200 & 0.774    & 0.77  & 0.763 & 0.771 & 0.766 & 0.747 / 0.002 & 0.771 / 0.002 \\
\multicolumn{1}{c|}{}                    & 500 & 0.61     & 0.601 & 0.61  & 0.599 & 0.602 & 0.615 / 0.015 & 0.602 / 0.003 \\ \midrule
\multicolumn{1}{c|}{\multirow{5}{*}{q}}  & 0.1 & 0.787    & 0.787 & 0.772 & 0.789 & 0.786 & 0.783 / 0     & 0.787 / 0.004 \\
\multicolumn{1}{c|}{}                    & 0.3 & 0.788    & 0.777 & 0.765 & 0.773 & 0.784 & 0.783 / 0     & 0.782 / 0.002 \\
\multicolumn{1}{c|}{}                    & 0.5 & 0.783    & 0.772 & 0.77  & 0.761 & 0.757 & 0.784 / 0.002 & 0.801 / 0     \\
\multicolumn{1}{c|}{}                    & 0.7 & 0.65     & 0.637 & 0.657 & 0.639 & 0.65  & 0.642 / 0.008 & 0.649 / 0.006 \\
\multicolumn{1}{c|}{}                    & 0.9 & 0.566    & 0.542 & 0.545 & 0.542 & 0.548 & 0.546 / 0.001 & 0.55 / 0.007  \\ \bottomrule
\end{tabular}
\caption{Comparison of accuracy with baseline and ASR by an increasing proportion of malicious clients ($\xi\geq 0.5$), \#clients $n$ and non-iid degree $q$, where FMNIST is used. The results under targeted attacks are in the form of “testing accuracy / ASR".}
\label{tab:fmnist}
\end{table*}
\begin{table*}[]
\centering
\begin{tabular}{@{}cc|cccccccc@{}}
\toprule
\multicolumn{2}{c|}{Attacks}                   & Baseline & GA    & LFA   & Krum  & Trim  & AA          & BA          & EA          \\ \midrule
\multicolumn{1}{c|}{\multirow{5}{*}{$\xi$}} & 0.5 & 0.573    & 0.57  & 0.574 & 0.557 & 0.562 & 0.568 / 0.006 & 0.572 / 0.007 & 0.571 / 0.019 \\
\multicolumn{1}{c|}{}                    & 0.6 & 0.562    & 0.559 & 0.559 & 0.547 & 0.534 & 0.521 / 0.011 & 0.567 / 0     & 0.568 / 0.03  \\
\multicolumn{1}{c|}{}                    & 0.7 & 0.54     & 0.52  & 0.506 & 0.513 & 0.515 & 0.531 / 0.011 & 0.524 / 0     & 0.531 / 0.031 \\
\multicolumn{1}{c|}{}                    & 0.8 & 0.519    & 0.508 & 0.489 & 0.494 & 0.488 & 0.482 / 0.003 & 0.52 / 0.004  & 0.501 / 0.05  \\
\multicolumn{1}{c|}{}                    & 0.9 & 0.489    & 0.492 & 0.474 & 0.45  & 0.483 & 0.475 / 0.006 & 0.484 / 0.008 & 0.478 / 0.059 \\ \midrule
\multicolumn{1}{c|}{\multirow{5}{*}{n}}  & 10  & 0.677    & 0.672 & 0.668 & 0.665 & 0.668 & 0.658 / 0     & 0.659 / 0.001 & 0.66 / 0.037  \\
\multicolumn{1}{c|}{}                    & 50  & 0.641    & 0.637 & 0.635 & 0.653 & 0.634 & 0.639 / 0     & 0.647 / 0.001 & 0.641 / 0.068 \\
\multicolumn{1}{c|}{}                    & 100 & 0.562    & 0.559 & 0.559 & 0.547 & 0.534 & 0.521 / 0.011 & 0.567 / 0     & 0.568 / 0.03  \\
\multicolumn{1}{c|}{}                    & 200 & 0.46     & 0.453 & 0.474 & 0.457 & 0.45  & 0.468 / 0.003 & 0.446 / 0     & 0.458 / 0.028 \\
\multicolumn{1}{c|}{}                    & 500 & 0.27     & 0.26  & 0.254 & 0.274 & 0.252 & 0.276 / 0.004 & 0.262 / 0     & 0.242 / 0     \\ \midrule
\multicolumn{1}{c|}{\multirow{5}{*}{q}}  & 0.1 & 0.573    & 0.574 & 0.566 & 0.562 & 0.554 & 0.555 / 0     & 0.572 / 0.001 & 0.558 / 0.058 \\
\multicolumn{1}{c|}{}                    & 0.3 & 0.567    & 0.567 & 0.556 & 0.535 & 0.553 & 0.561 / 0     & 0.569 / 0     & 0.543 / 0.064 \\
\multicolumn{1}{c|}{}                    & 0.5 & 0.562    & 0.559 & 0.559 & 0.547 & 0.534 & 0.521 / 0.011 & 0.567 / 0     & 0.568 / 0.03  \\
\multicolumn{1}{c|}{}                    & 0.7 & 0.426    & 0.417 & 0.394 & 0.435 & 0.424 & 0.44 / 0.013  & 0.41 / 0.004  & 0.429 / 0.015 \\
\multicolumn{1}{c|}{}                    & 0.9 & 0.229    & 0.227 & 0.216 & 0.224 & 0.229 & 0.219 / 0.024 & 0.217 / 0.013 & 0.214 / 0.018 \\ \bottomrule
\end{tabular}
\caption{Comparison of accuracy with baseline and ASR by an increasing proportion of malicious clients ($\xi\geq 0.5$), \#clients $n$ and non-iid degree $q$, where CIFAR-10 is used. The results under targeted attacks are in the form of “testing accuracy / ASR".}
\label{tab:cifar}
\end{table*}

In Figures~\ref{fig:xi06fmnist} and~\ref{fig:xi06cifar}, we present the comparison of testing accuracy among baseline, semi-honest and malicious groups under targeted attacks and ASR between the groups under untargeted attacks. In addition, we also give the experimental results by varying proportion of malicious clients ($\xi\geq 0.5$), \#clients $n$ and non-iid degree $q$ in Tables~\ref{tab:fmnist} and~\ref{tab:cifar}.
We see that the results are consistent with the analysis in Section~\ref{sec:expacc}: the untargeted attacks nearly have no impact on the accuracy of the final model (only slightly decreasing the speed of convergence). 
Since GA may directly upload noise, it can be easily detected from the beginning of the training to the end, resulting in the same convergence as the baseline.
For LFA, Krum, Trim, and AA Attacks, \texttt{MUDGUARD} is also difficult to distinguish between semi-honest and malicious clients at the beginning.   Thus, the speed of convergence is slightly decreased.
The main difference is that: LeNet achieves around $78\%$ in FMNIST, while ResNet-18 provides approx. $56\%$ accuracy in CIFAR-10.

We also provide comparisons with the state-of-the-art Byzantine-robust methods in Figure~\ref{fig:compfmnist} and~\ref{fig:compcifar10} on FMNIST and CIFAR-10 respectively. 
Similar to Figure~\ref{fig:comp}, under the malicious majority of untargeted attacks, the testing accuracies of \texttt{FLTrust}, \texttt{MUDGUARD}, and \texttt{weights-MUDGUARD} are maintained at the same level of the baseline. 
Under the targeted attacks, \texttt{FLTrust}, \texttt{MUDGUARD}, and \texttt{weights-MUDGUARD} can restrain ASR to about 0\%-10\%. For more detailed explanations, please refer to Section~\ref{sec:expacc}.

\begin{figure}[t]
    \centering
    \begin{subfigure}[b]{0.235\textwidth}
        \centering
        \includegraphics[width=1.04\textwidth]{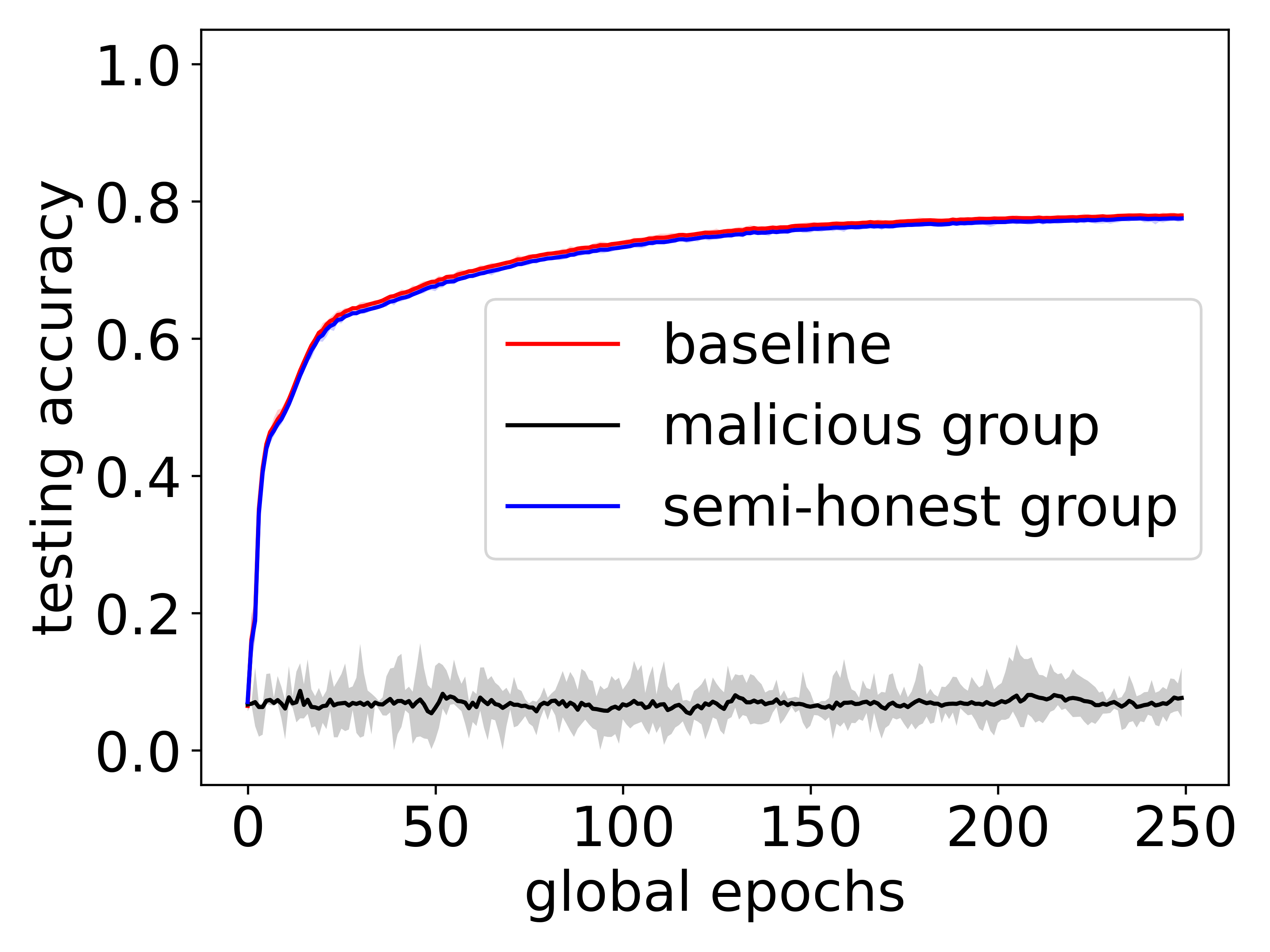}
        \caption{Gaussian Attack}
    \end{subfigure}
    \begin{subfigure}[b]{0.235\textwidth}
        \centering
        \includegraphics[width=1.04\textwidth]{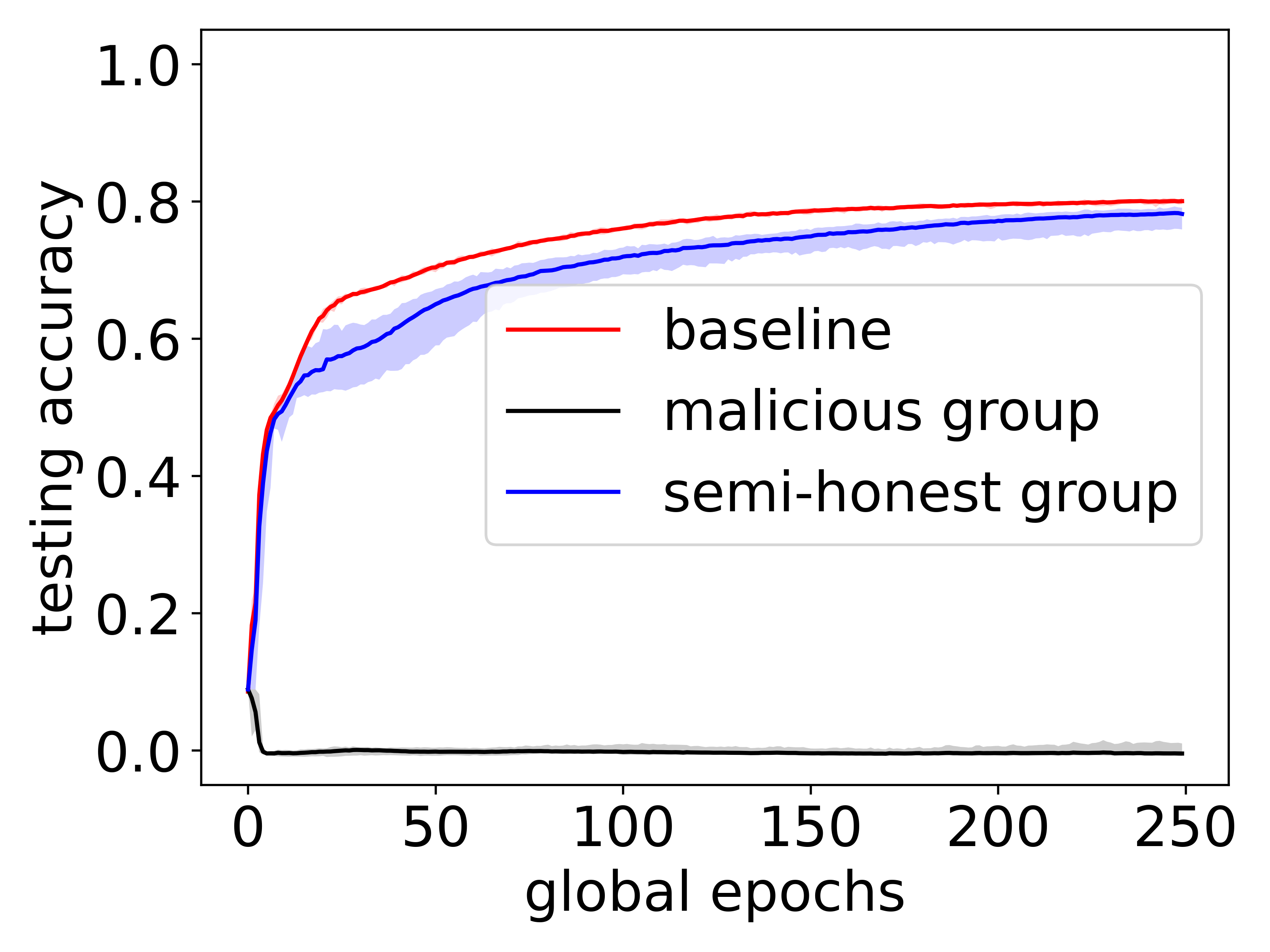}
        \caption{Label Flipping Attack}
    \end{subfigure}    
    \begin{subfigure}[b]{0.235\textwidth}
        \centering
        \includegraphics[width=1.04\textwidth]{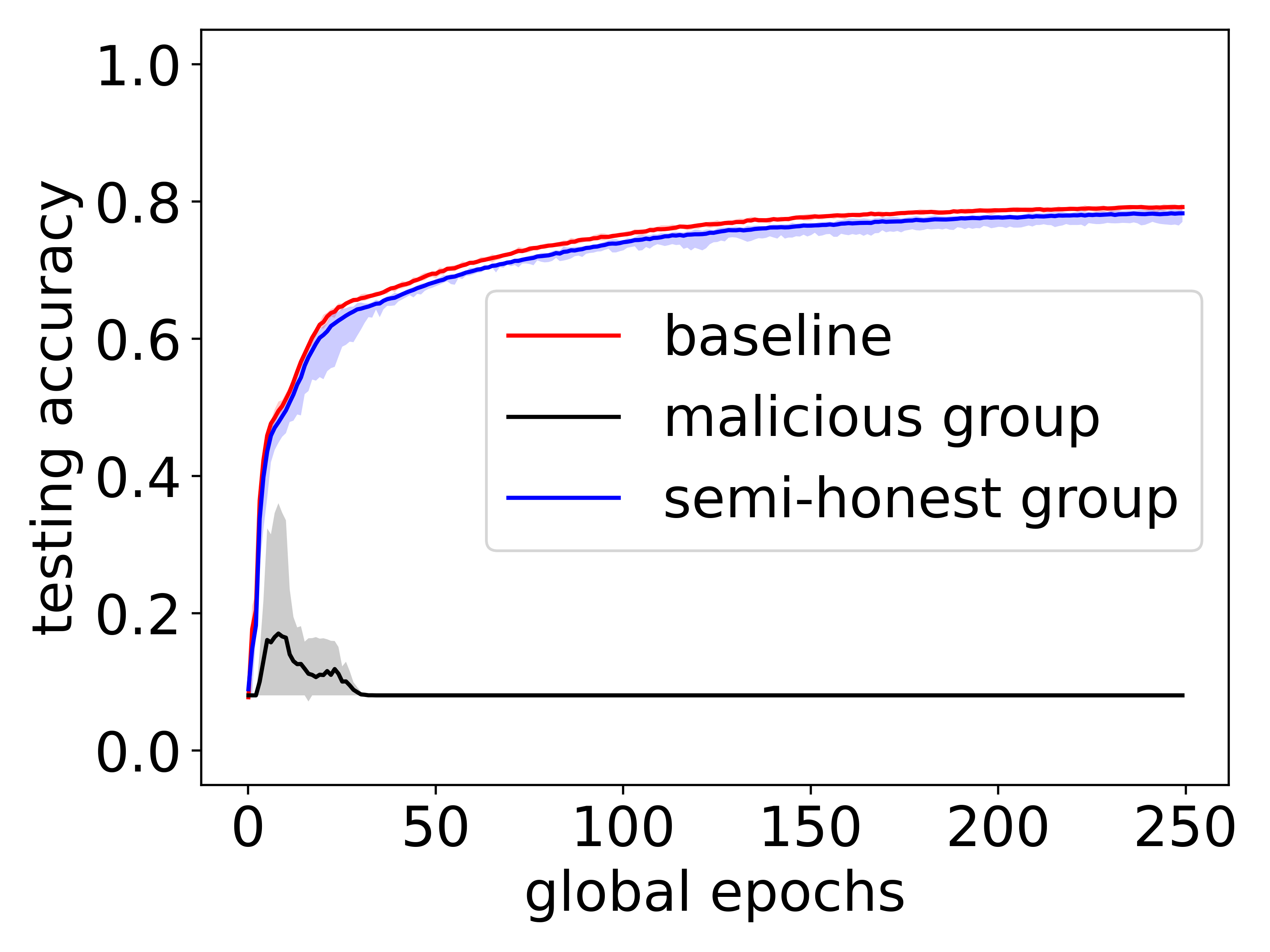}
        \caption{Krum Attack}
    \end{subfigure}    
    \begin{subfigure}[b]{0.235\textwidth}
        \centering
        \includegraphics[width=1.04\textwidth]{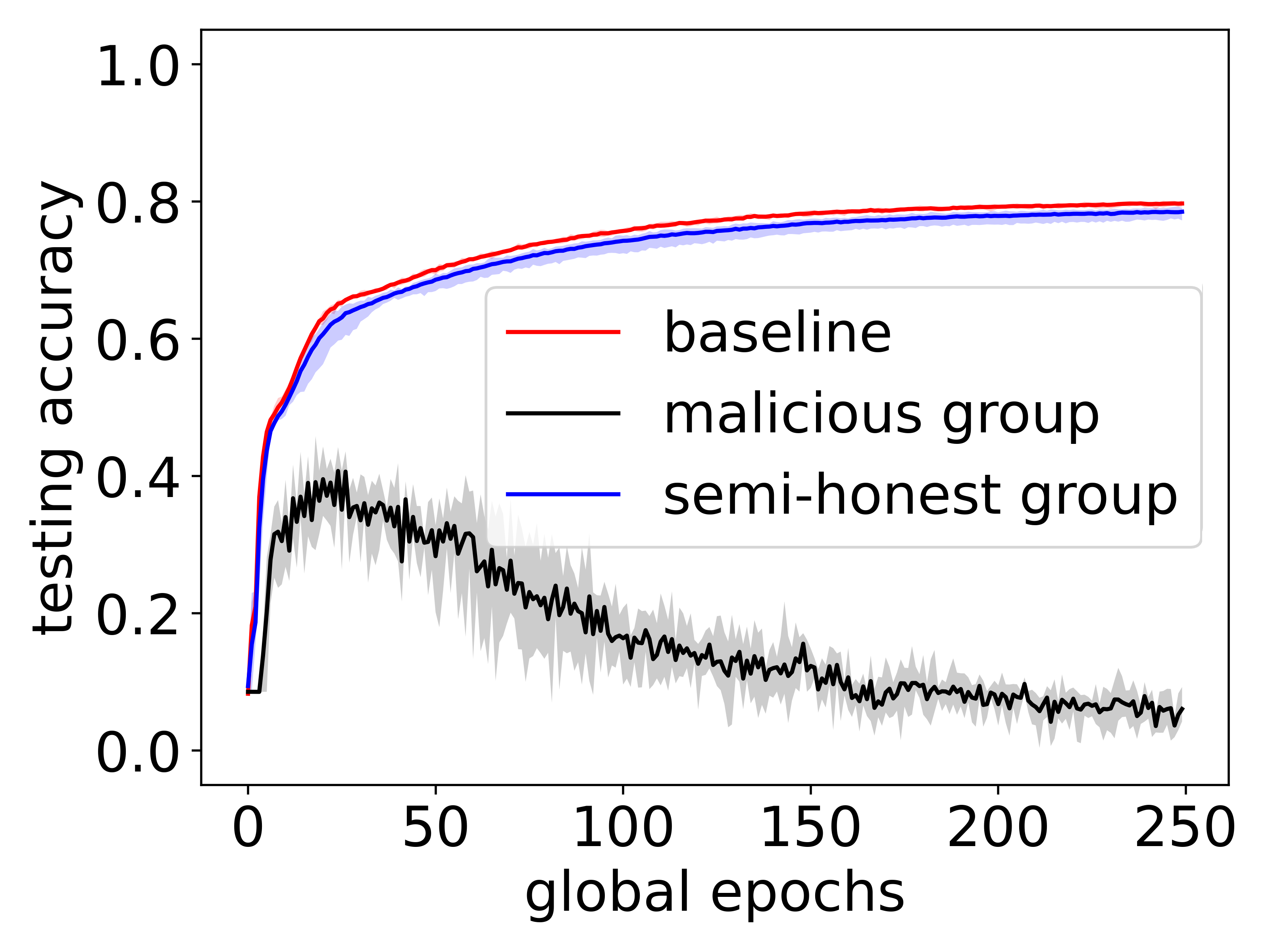}
        \caption{Trim Attack}
    \end{subfigure}  
        \begin{subfigure}[b]{0.235\textwidth}
        \centering
        \includegraphics[width=1.04\textwidth]{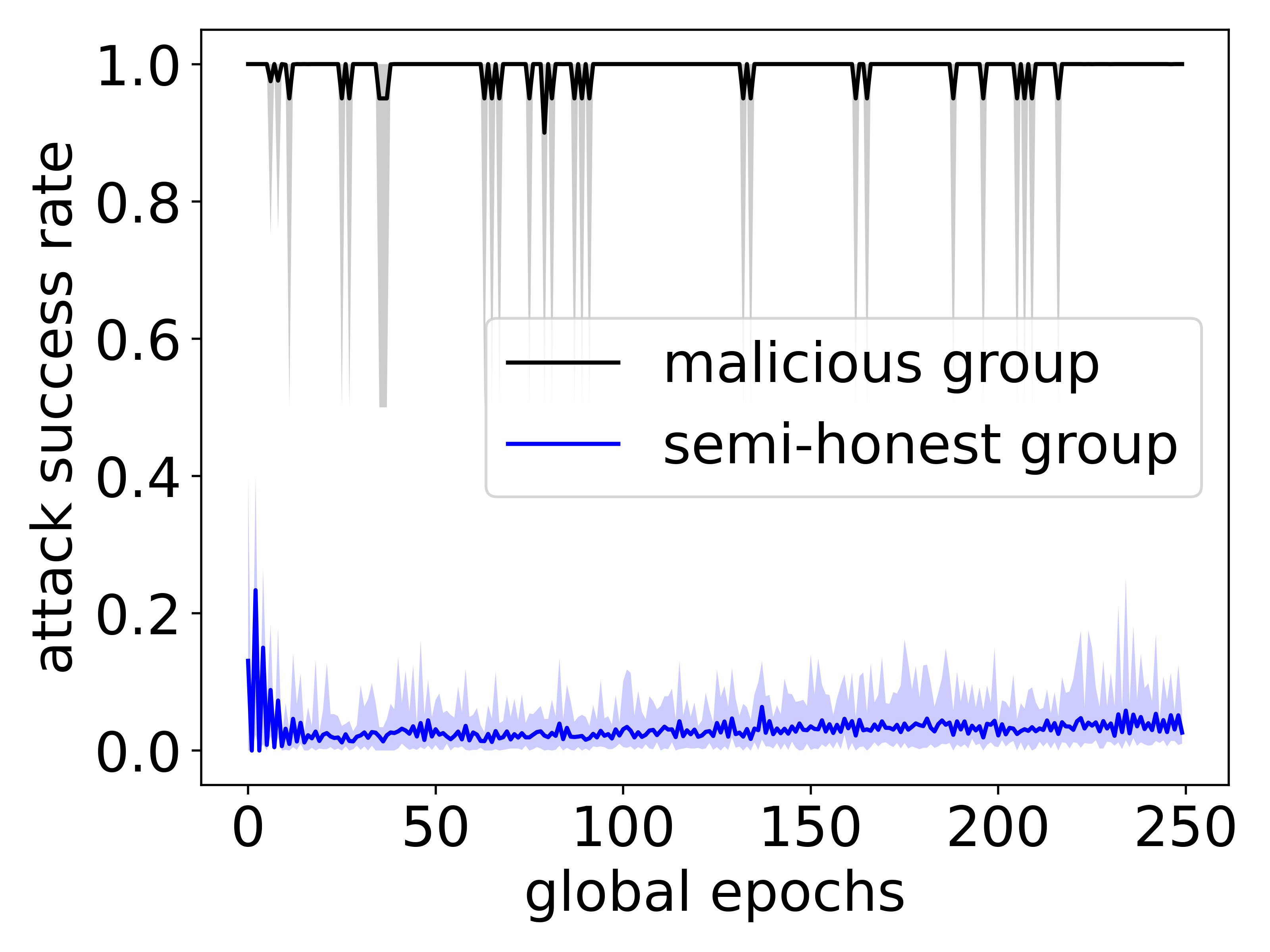}
        
        \caption{Adaptive Attack}
    \end{subfigure}  
    \begin{subfigure}[b]{0.235\textwidth}
        \centering
        \includegraphics[width=1.04\textwidth]{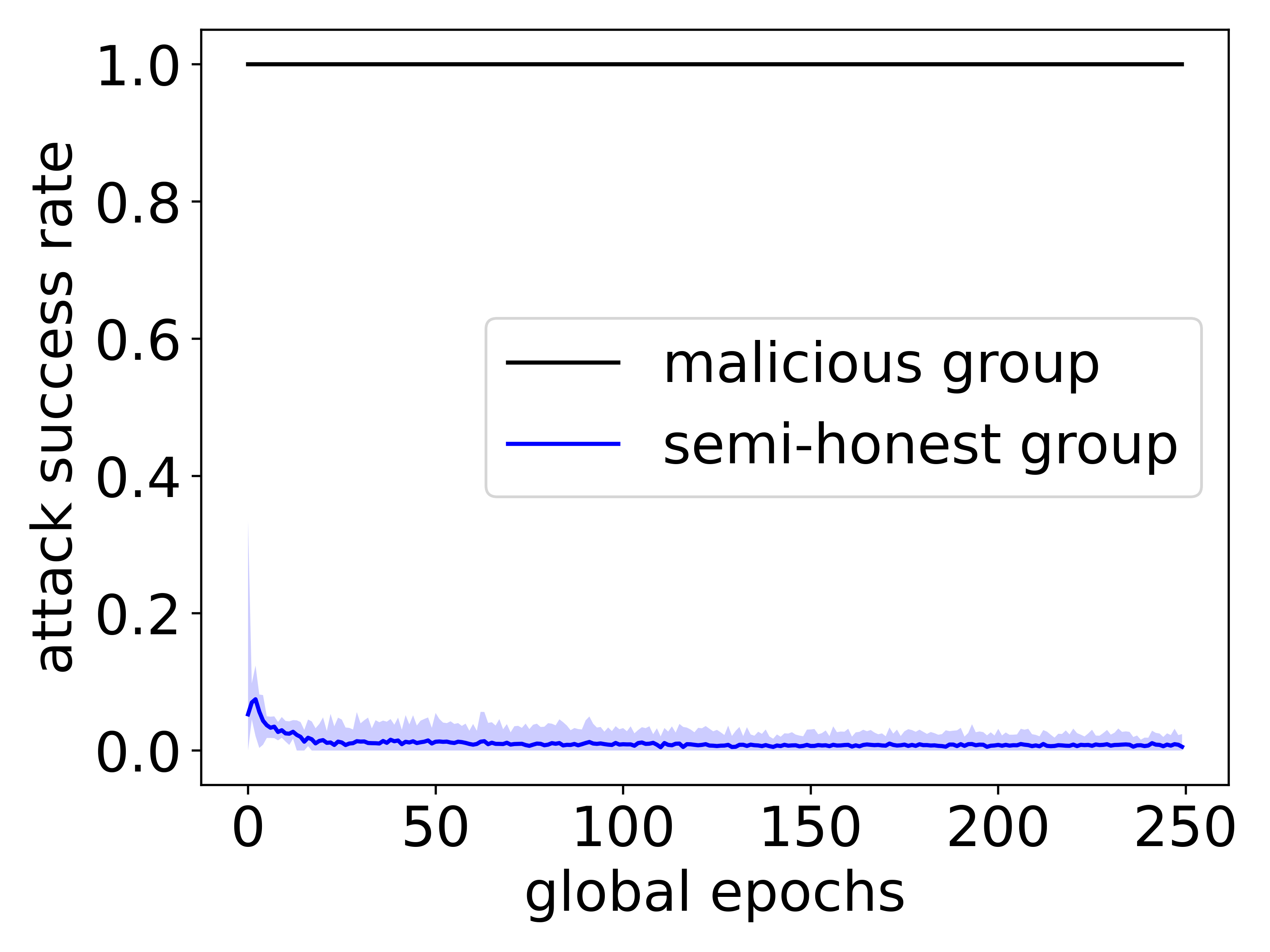}
        \caption{Backdoor Attack}
    \end{subfigure}
    \caption{Comparison of testing accuracy among baseline, semi-honest, and malicious groups under targeted attacks (a-e) and ASR between the groups under untargeted attack (f), where we use LeNet to train FMNIST by default settings in Table~\ref{tab:flsetting}.}
    \label{fig:xi06fmnist}
\end{figure}

\begin{figure}[t]
    \centering
    \begin{subfigure}[b]{0.235\textwidth}
        \centering
        \includegraphics[width=1.04\textwidth]{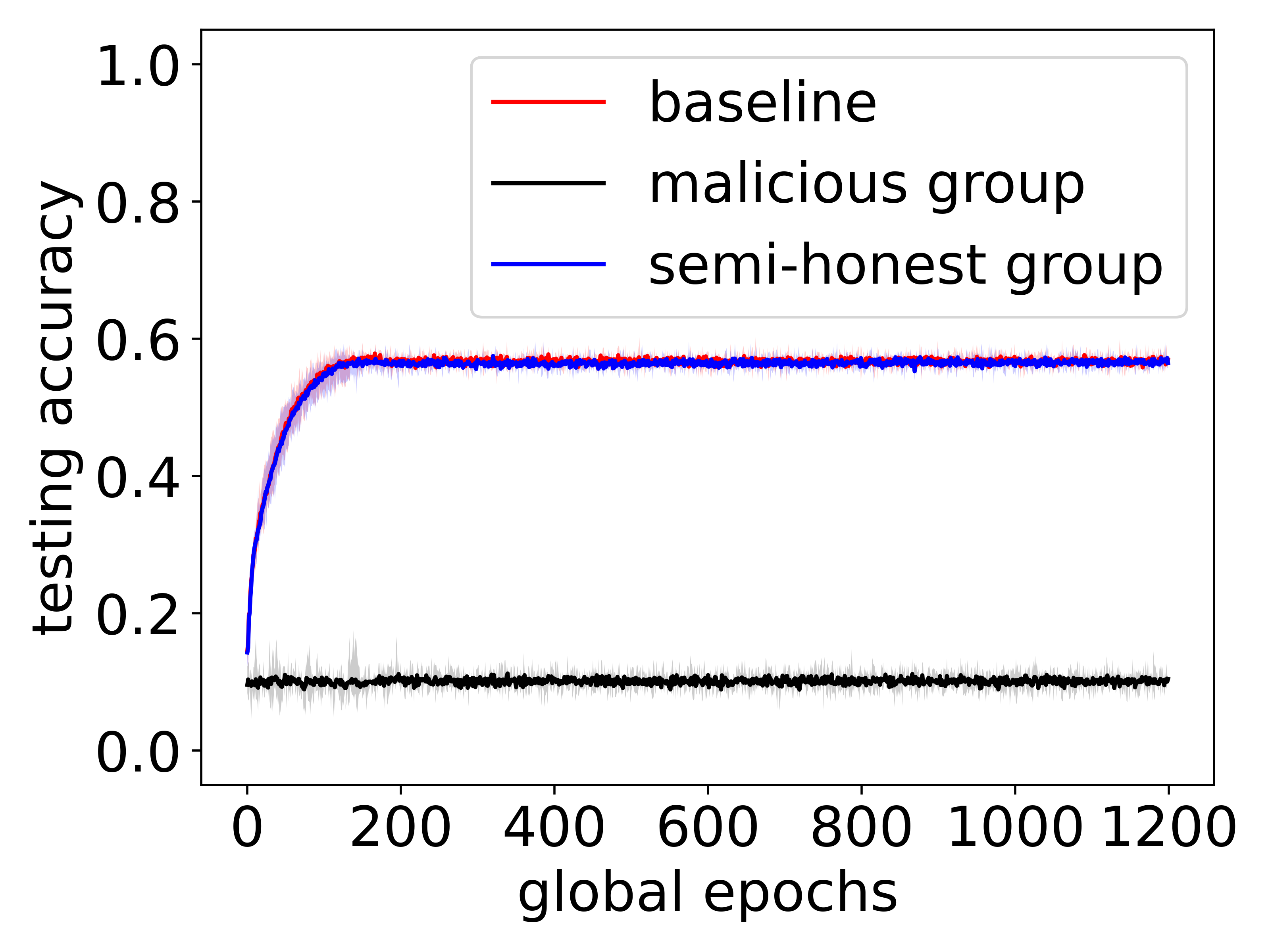}
        \caption{Gaussian Attack}
    \end{subfigure}
    \begin{subfigure}[b]{0.235\textwidth}
        \centering
        \includegraphics[width=1.04\textwidth]{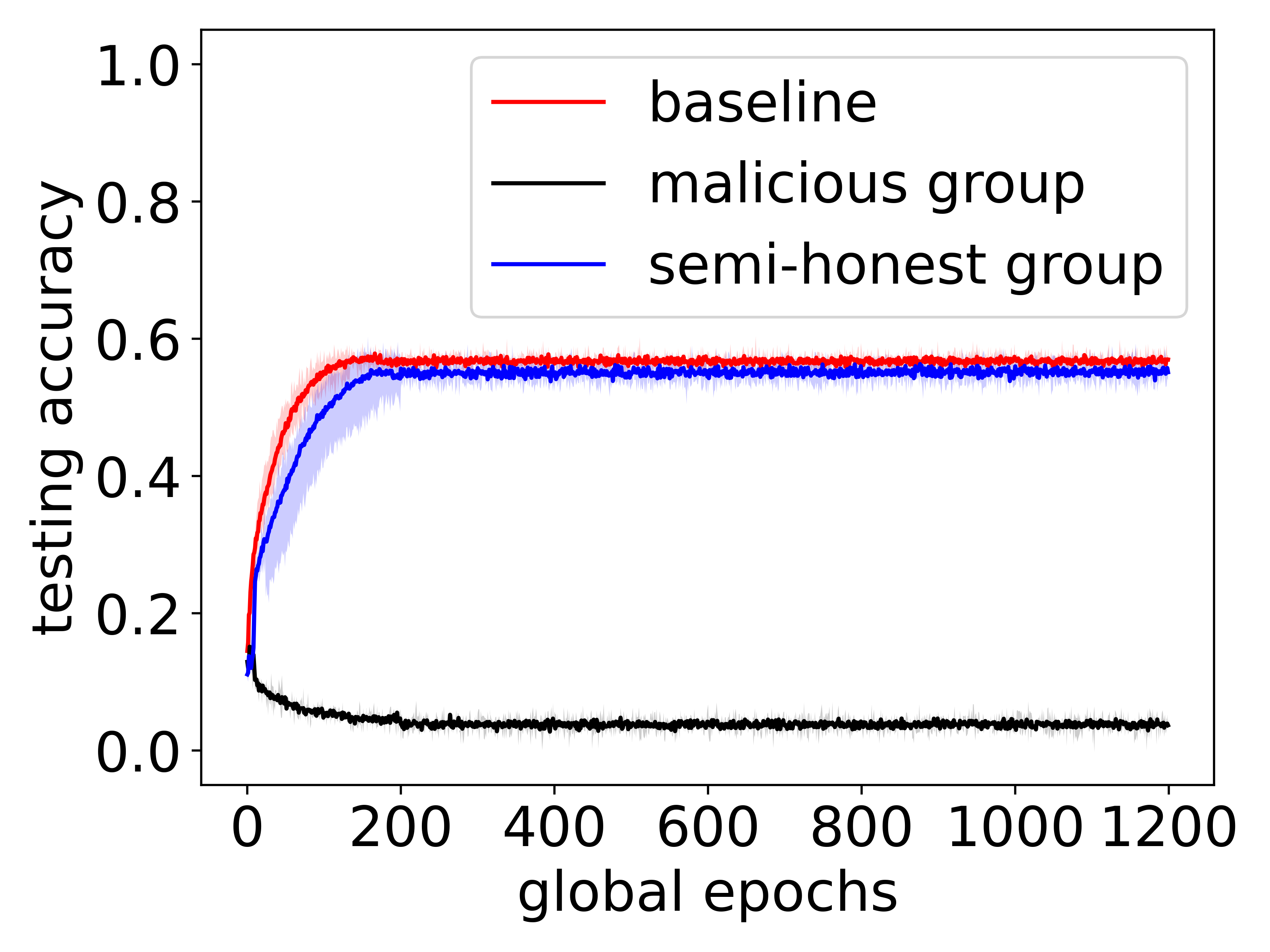}
        \caption{Label Flipping Attack}
    \end{subfigure}    
    \begin{subfigure}[b]{0.235\textwidth}
        \centering
        \includegraphics[width=1.04\textwidth]{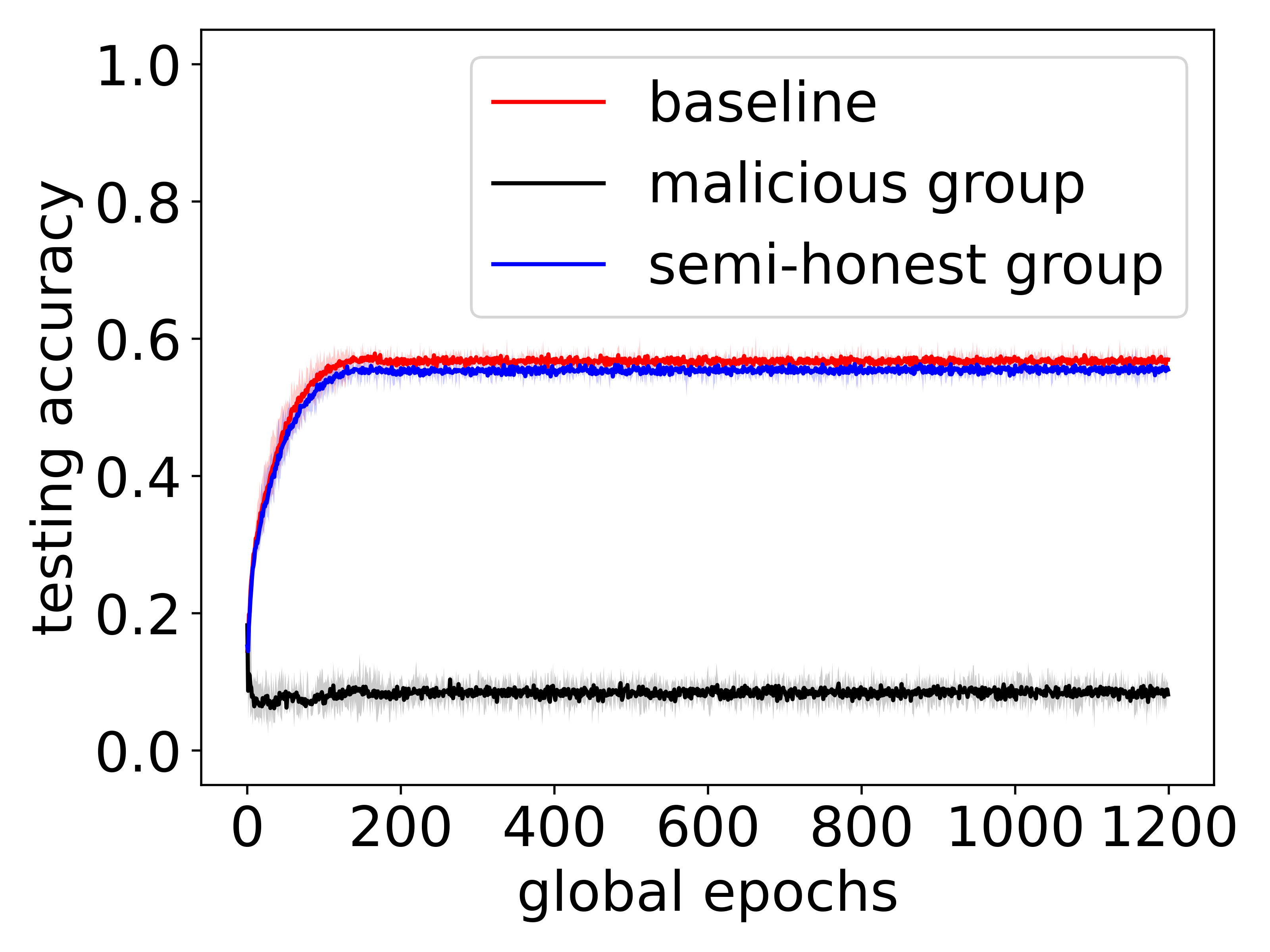}
        \caption{Krum Attack}
    \end{subfigure}    
    \begin{subfigure}[b]{0.235\textwidth}
        \centering
        \includegraphics[width=1.04\textwidth]{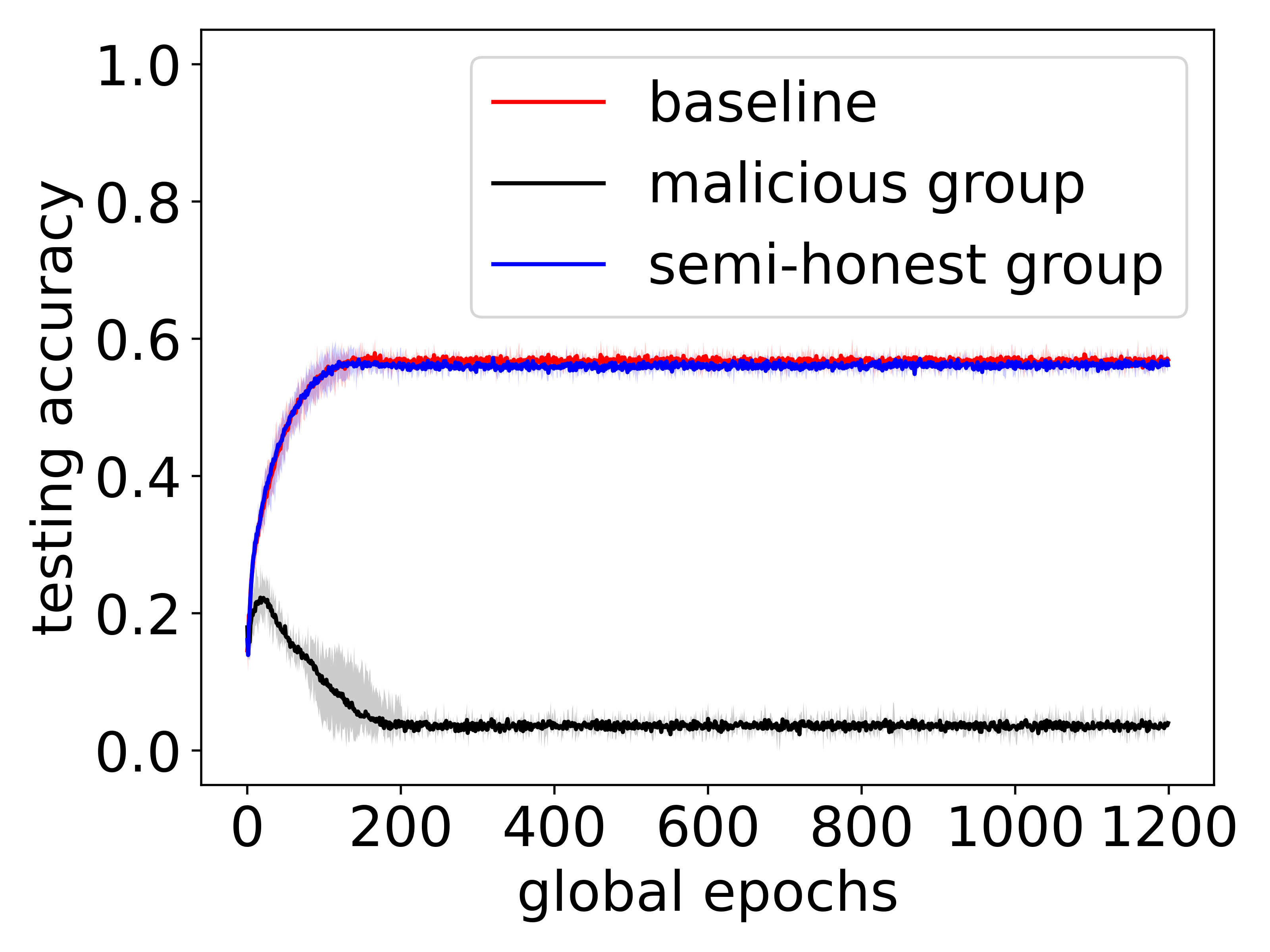}
        \caption{Trim Attack}
    \end{subfigure}  
    \begin{subfigure}[b]{0.235\textwidth}
        \centering
        \includegraphics[width=1.04\textwidth]{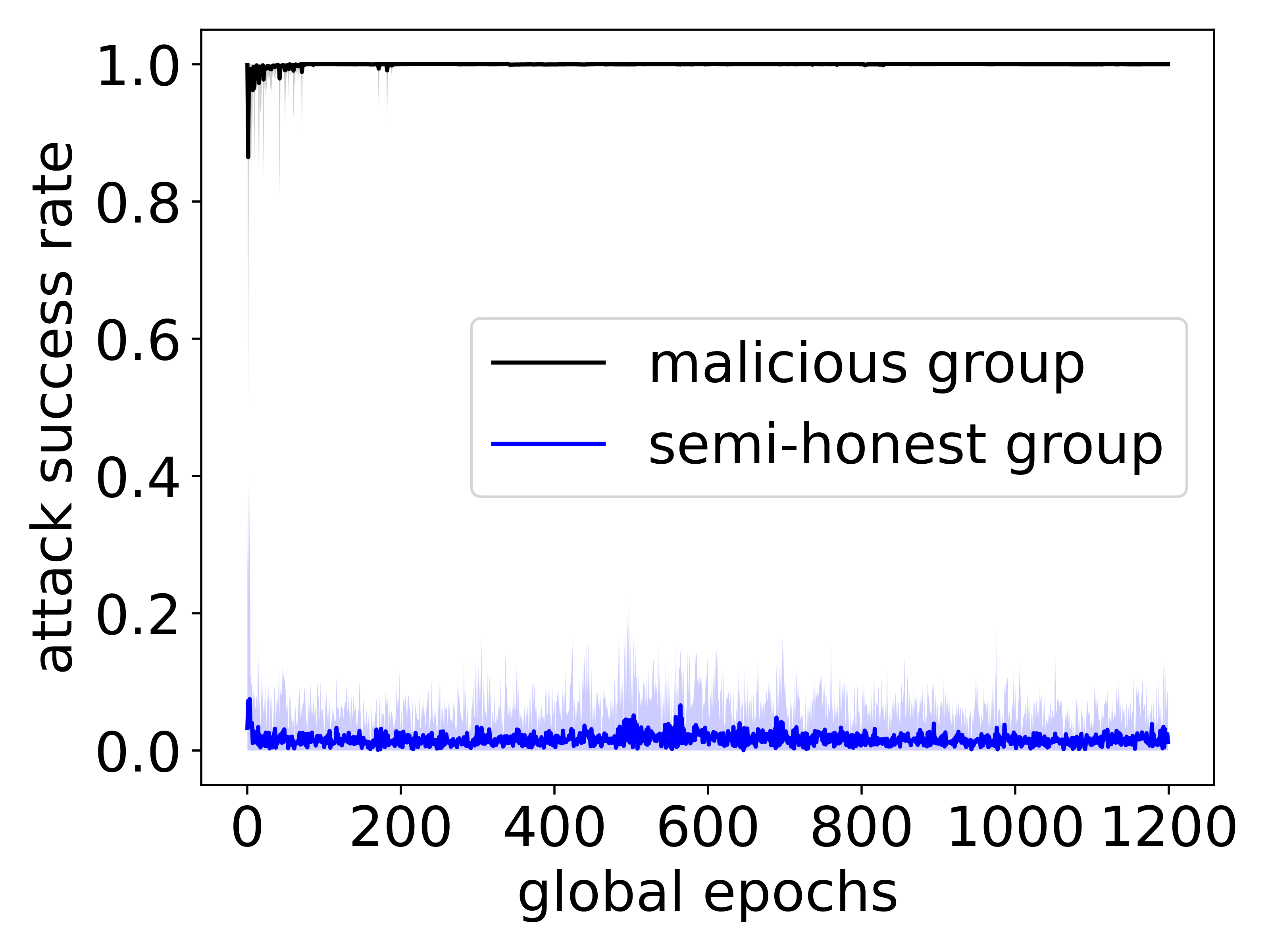}
        
        \caption{Adaptive Attack}
    \end{subfigure}  
    \begin{subfigure}[b]{0.235\textwidth}
        \centering
        \includegraphics[width=1.04\textwidth]{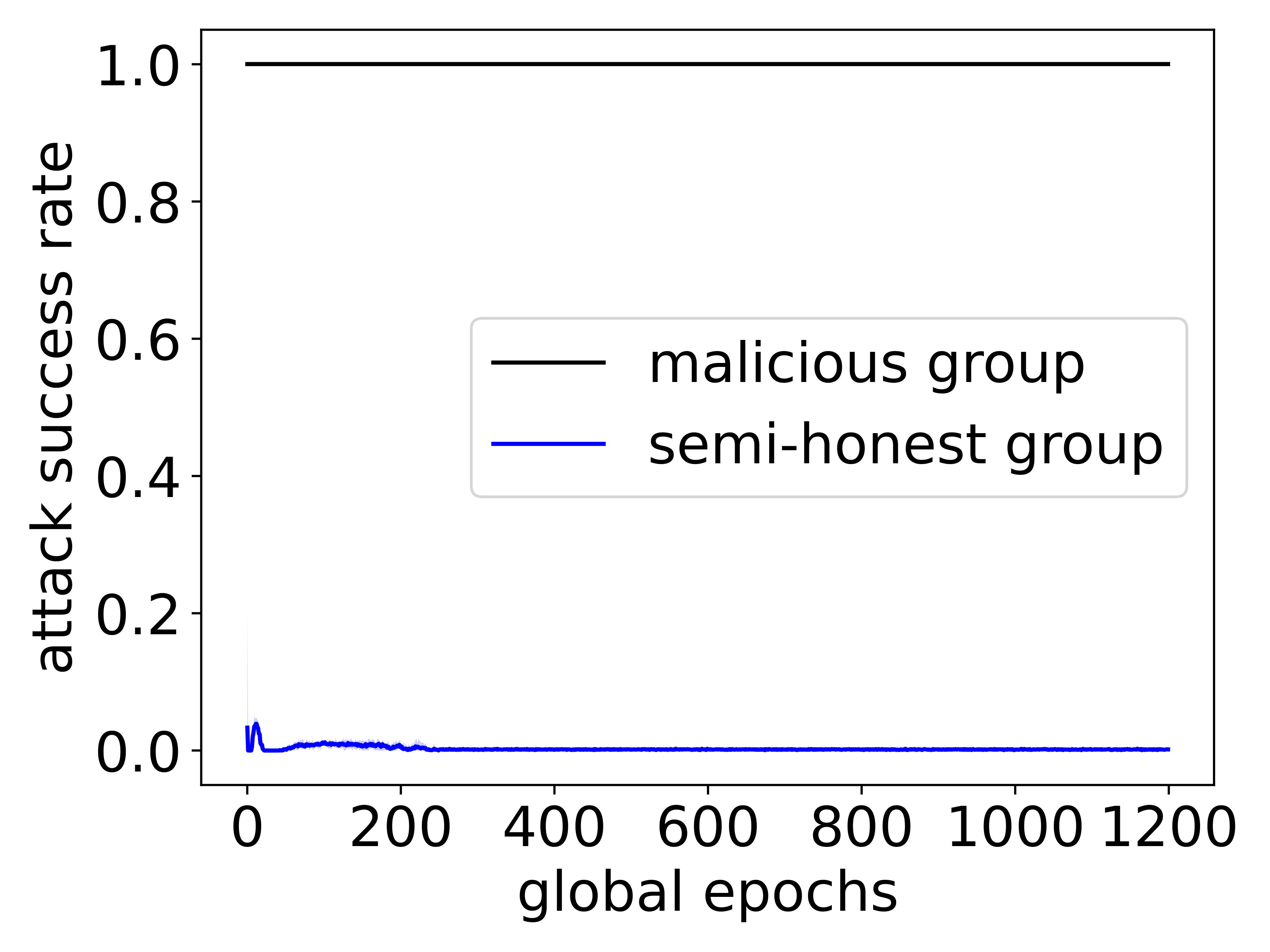}
        \caption{Backdoor Attack}
    \end{subfigure}
    \begin{subfigure}[b]{0.235\textwidth}
        \centering
        \includegraphics[width=1.04\textwidth]{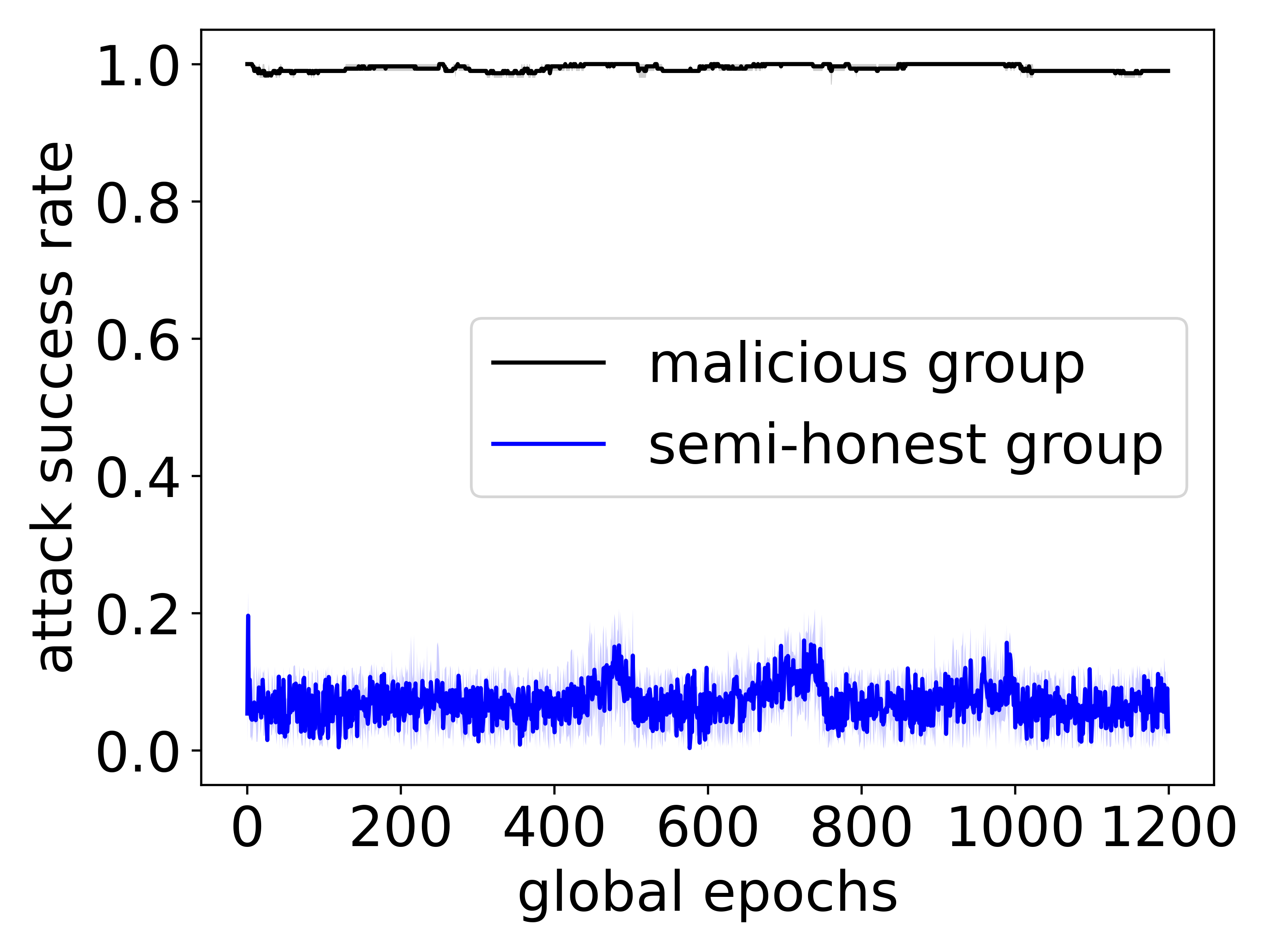}
        \caption{Edge-case Attack}
    \end{subfigure}    
    \caption{Comparison of testing accuracy among baseline, semi-honest, and malicious groups under targeted attacks (a-e) and ASR between the groups under untargeted attacks (f-g), where we use ResNet-18 to train CIFAR-10 by the default settings in Table~\ref{tab:flsetting}.}
    \label{fig:xi06cifar}
\end{figure}

\begin{figure}[!ht]
    \centering
    \begin{subfigure}[b]{0.235\textwidth}
        \centering
        \includegraphics[width=1.04\textwidth]{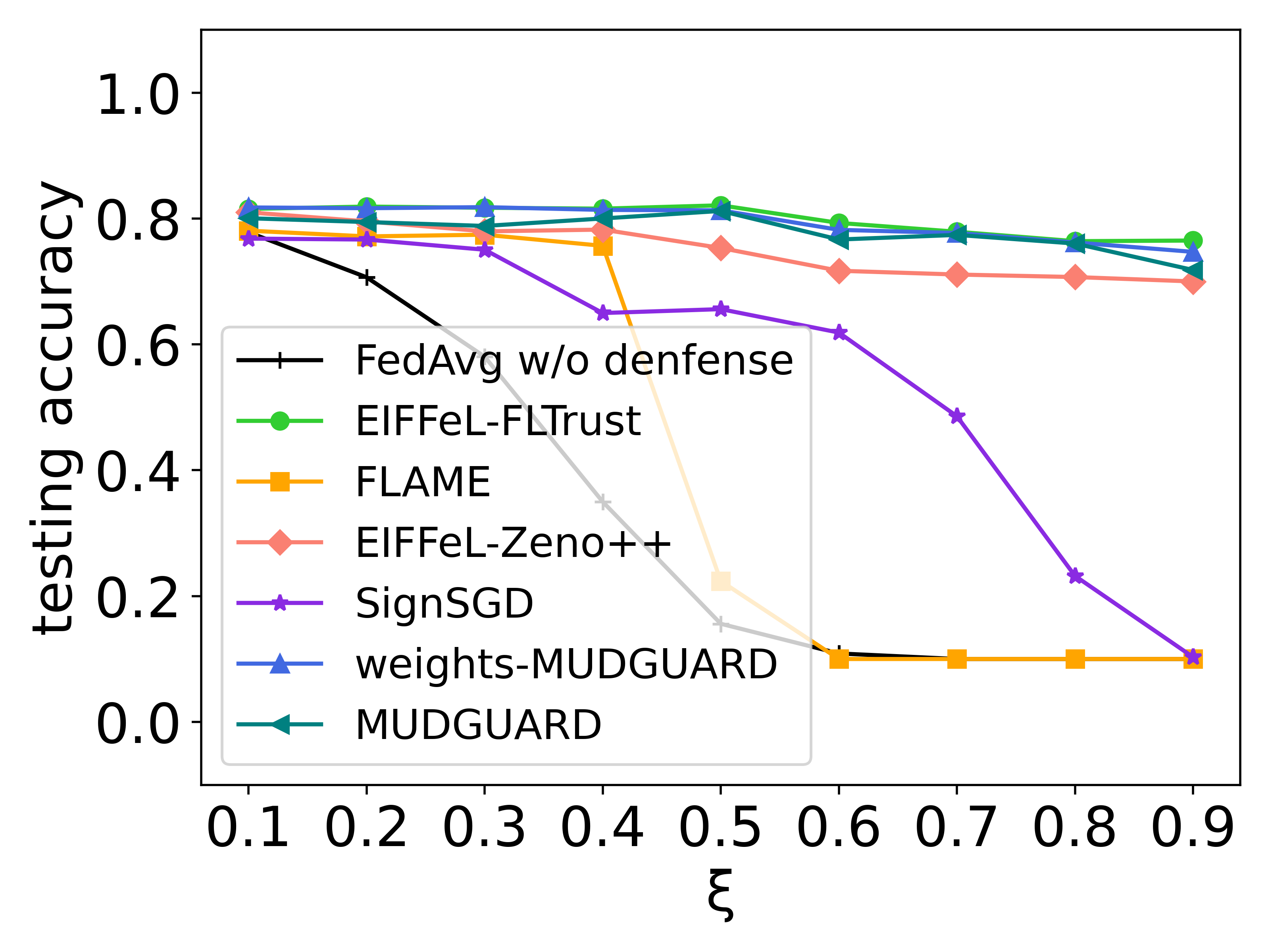}
        \caption{Gaussian Attack}
    \end{subfigure}
    \begin{subfigure}[b]{0.235\textwidth}
        \centering
        \includegraphics[width=1.04\textwidth]{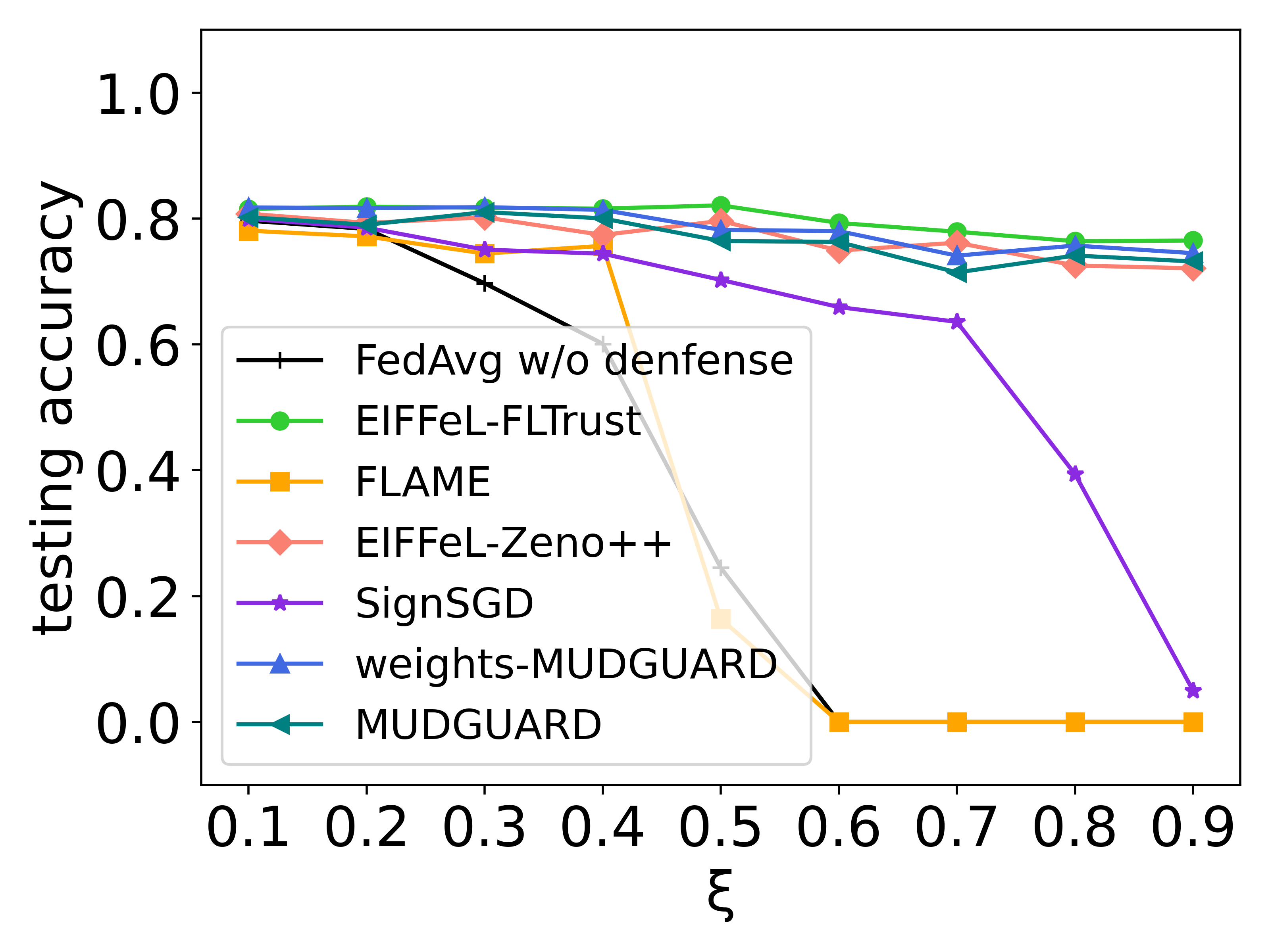}
        \caption{Label Flipping Attack}
    \end{subfigure}    
    \begin{subfigure}[b]{0.235\textwidth}
        \centering
        \includegraphics[width=1.04\textwidth]{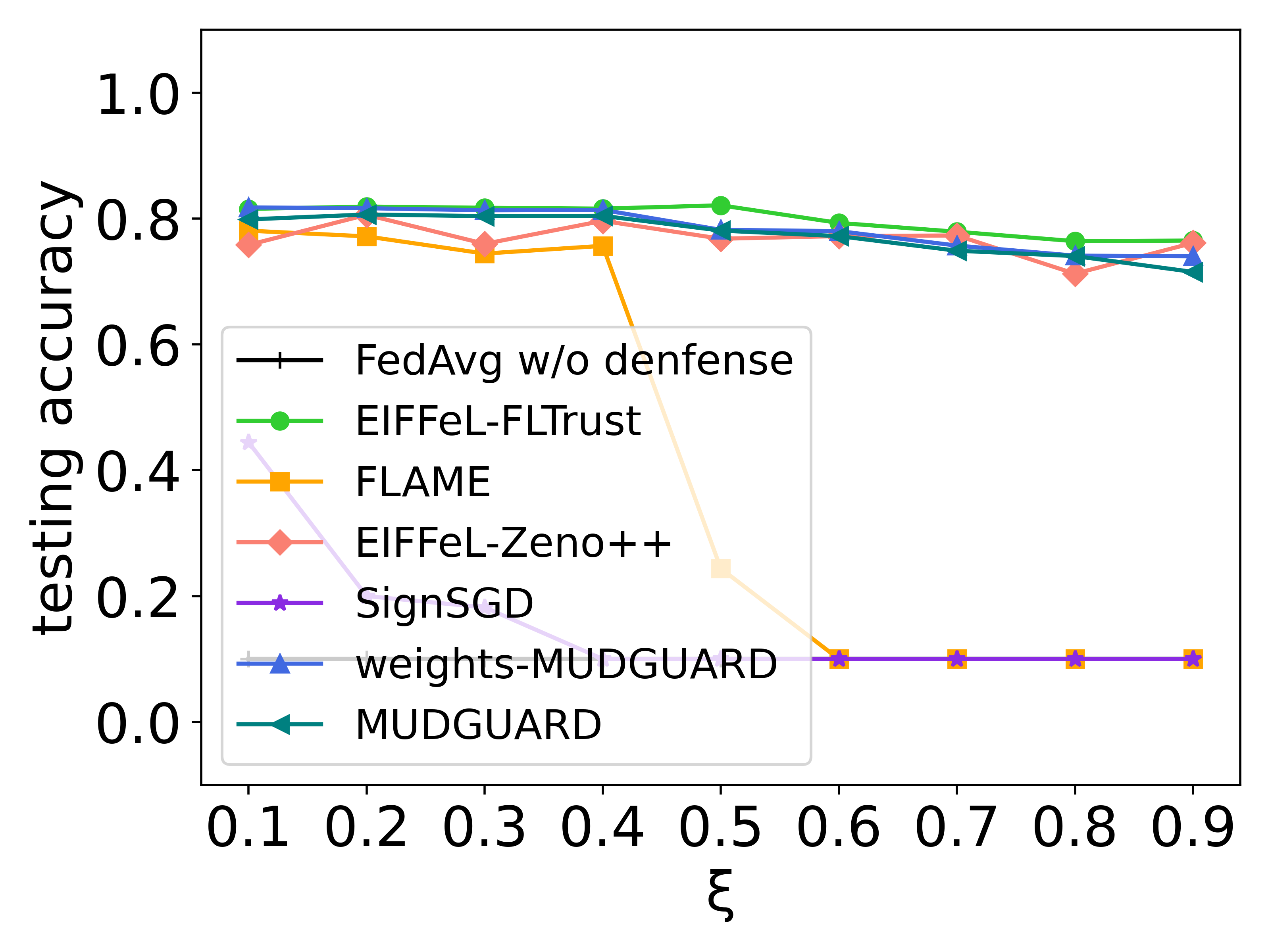}
        \caption{Krum Attack}
    \end{subfigure}    
    \begin{subfigure}[b]{0.235\textwidth}
        \centering
        \includegraphics[width=1.04\textwidth]{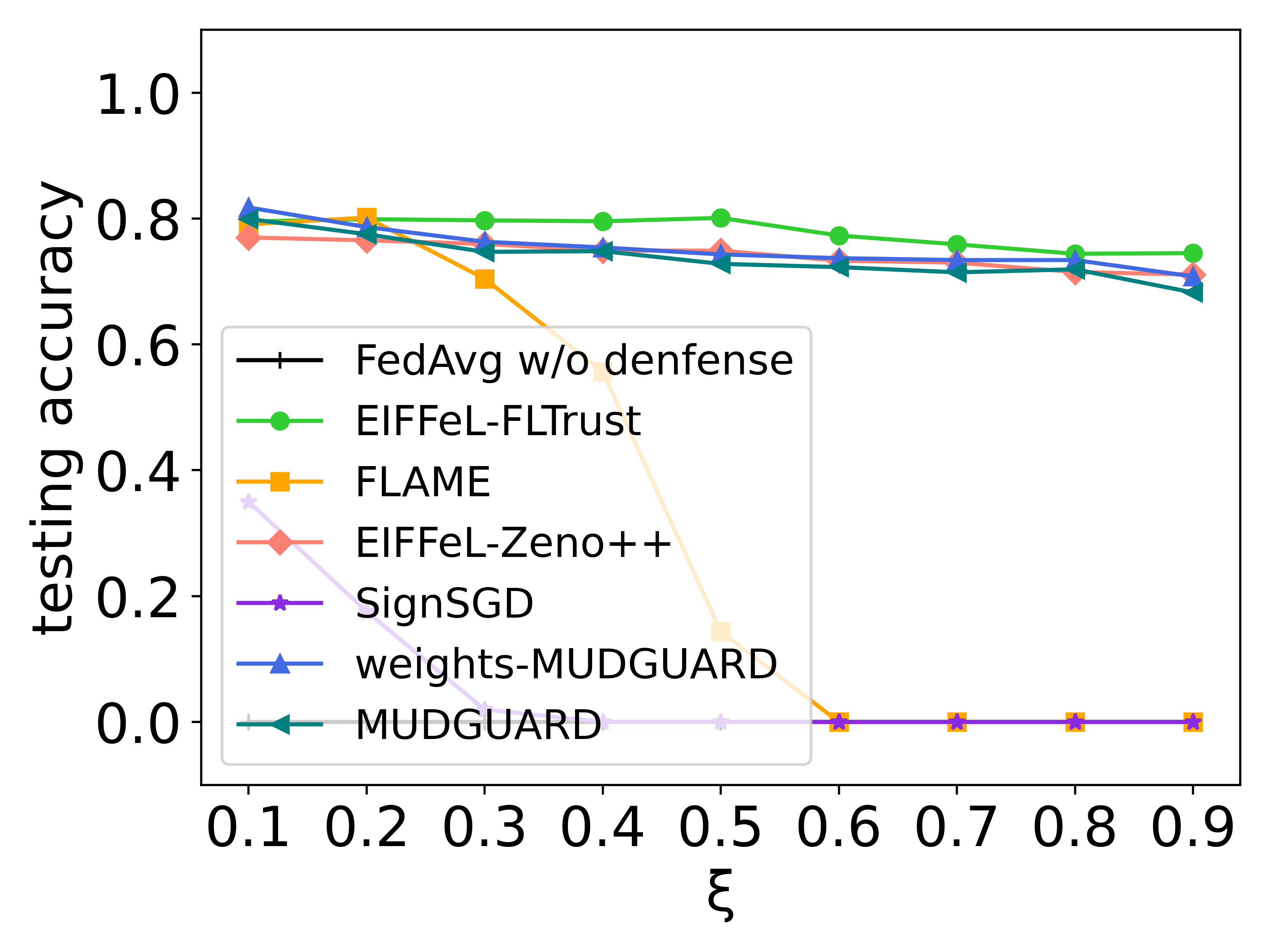}
        \caption{Trim Attack}
    \end{subfigure}  
    \begin{subfigure}[b]{0.235\textwidth}
        \centering
        \includegraphics[width=1.04\textwidth]{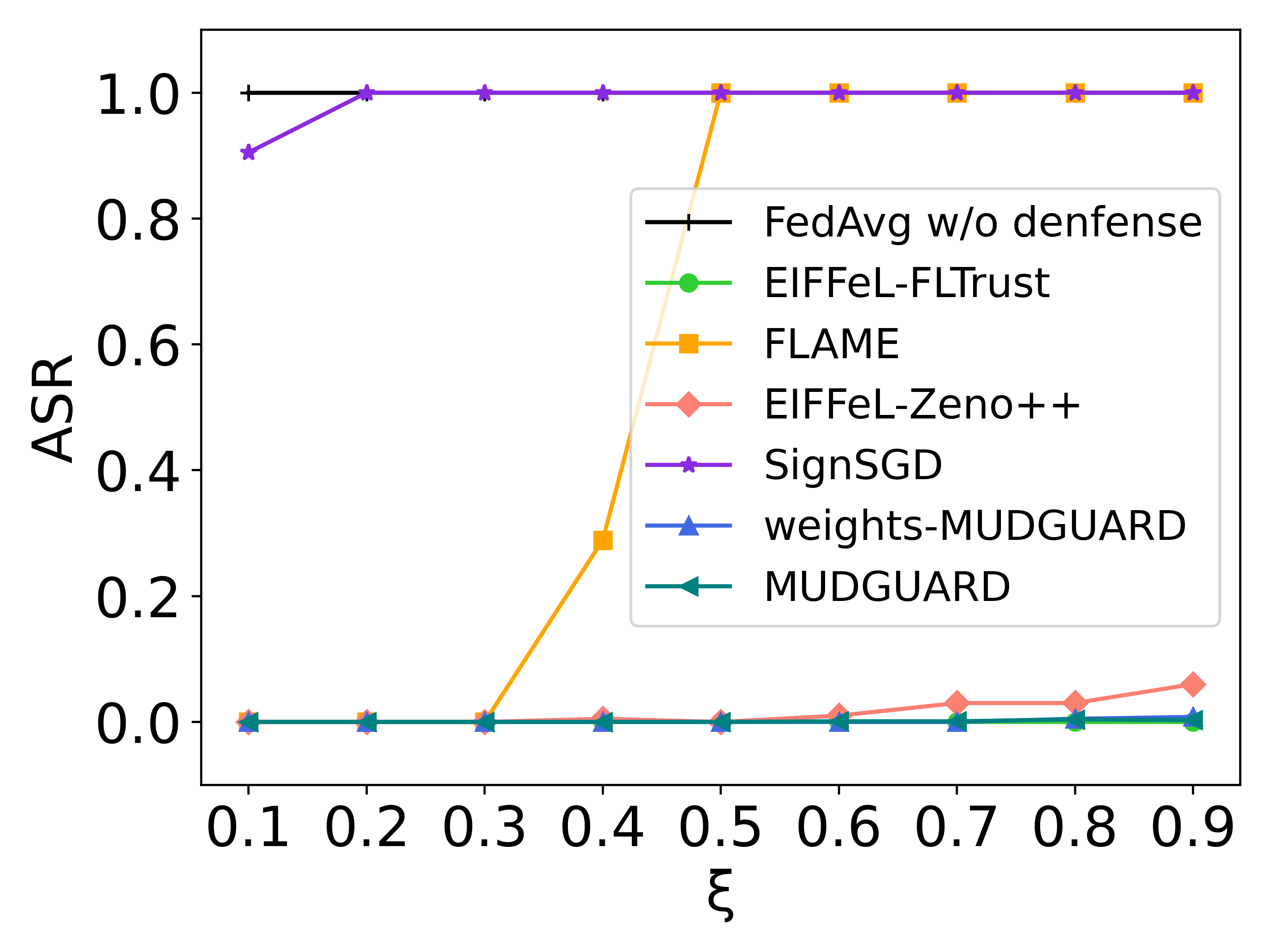}
        \caption{Adaptive Attack}
    \end{subfigure}  
    \begin{subfigure}[b]{0.235\textwidth}
        \centering
        \includegraphics[width=1.04\textwidth]{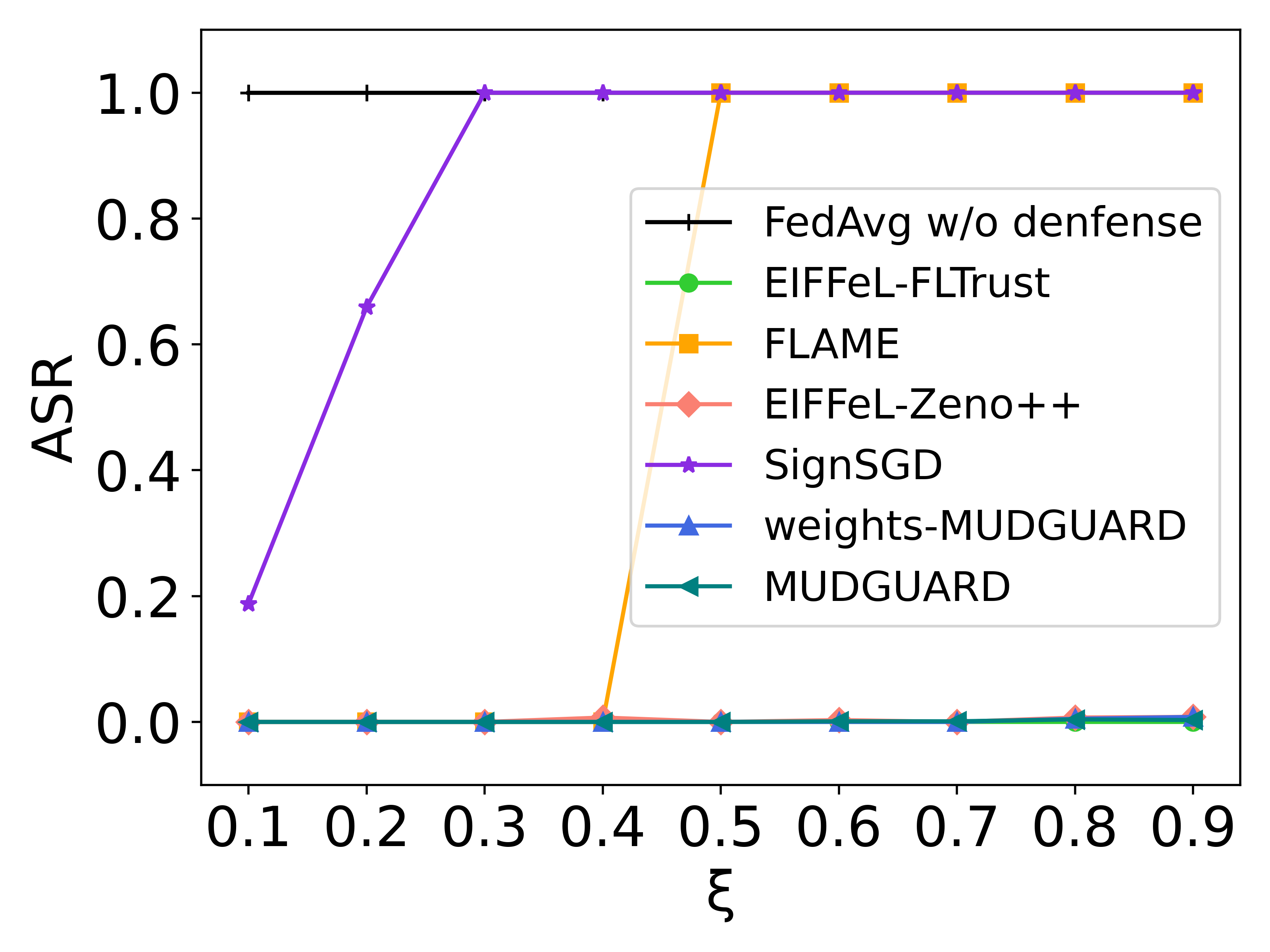}
        \caption{Backdoor Attack}
    \end{subfigure}
  
    \caption{Comparison with Byzantine-robust methods in FMNIST by $\xi=0.1 - 0.9$.}
    \label{fig:compfmnist}
\end{figure}

\begin{figure}[!ht]
    \centering
    \begin{subfigure}[b]{0.235\textwidth}
        \centering
        \includegraphics[width=1.04\textwidth]{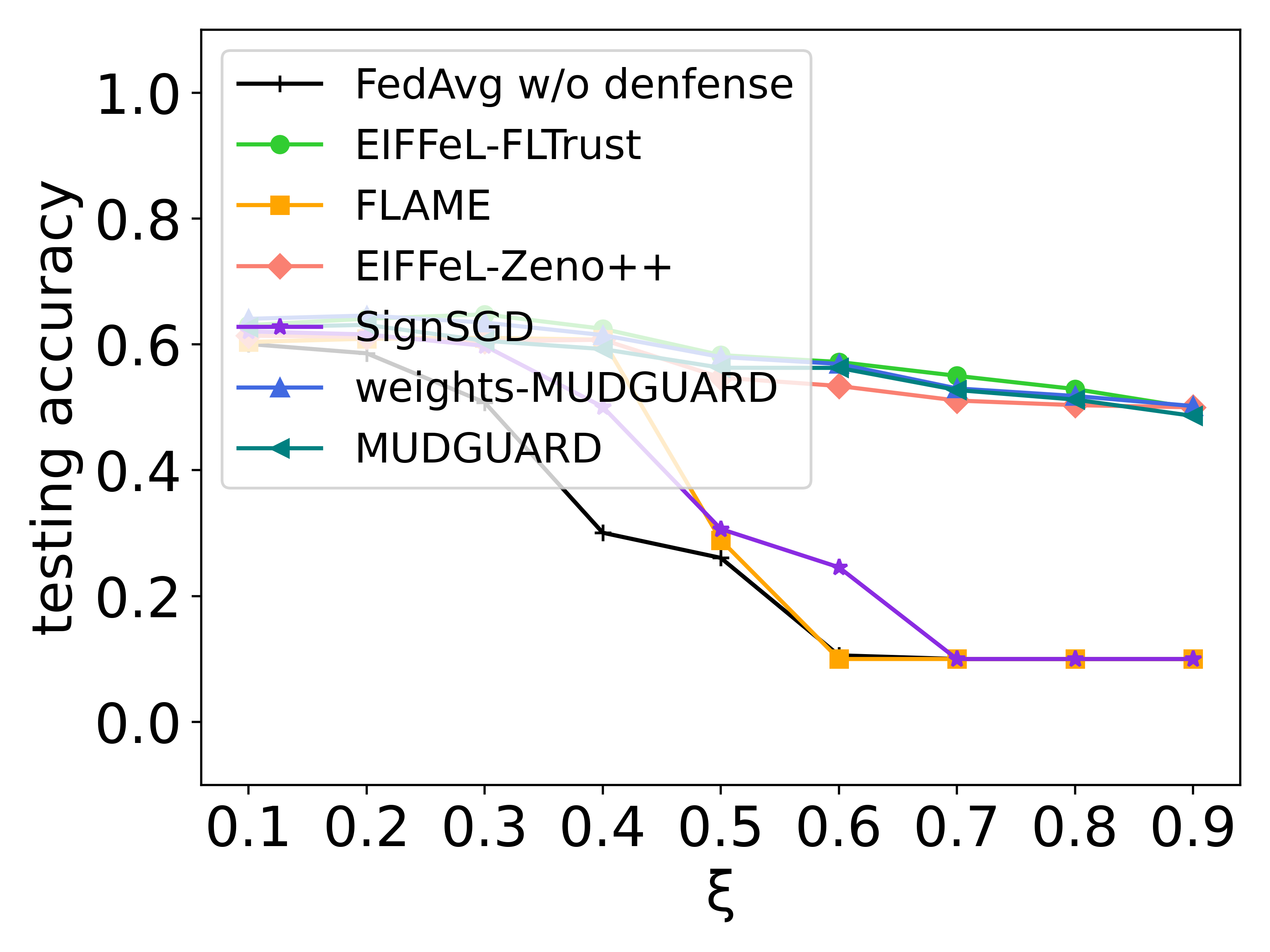}
        \caption{Gaussian Attack}
    \end{subfigure}
    \begin{subfigure}[b]{0.235\textwidth}
        \centering
        \includegraphics[width=1.04\textwidth]{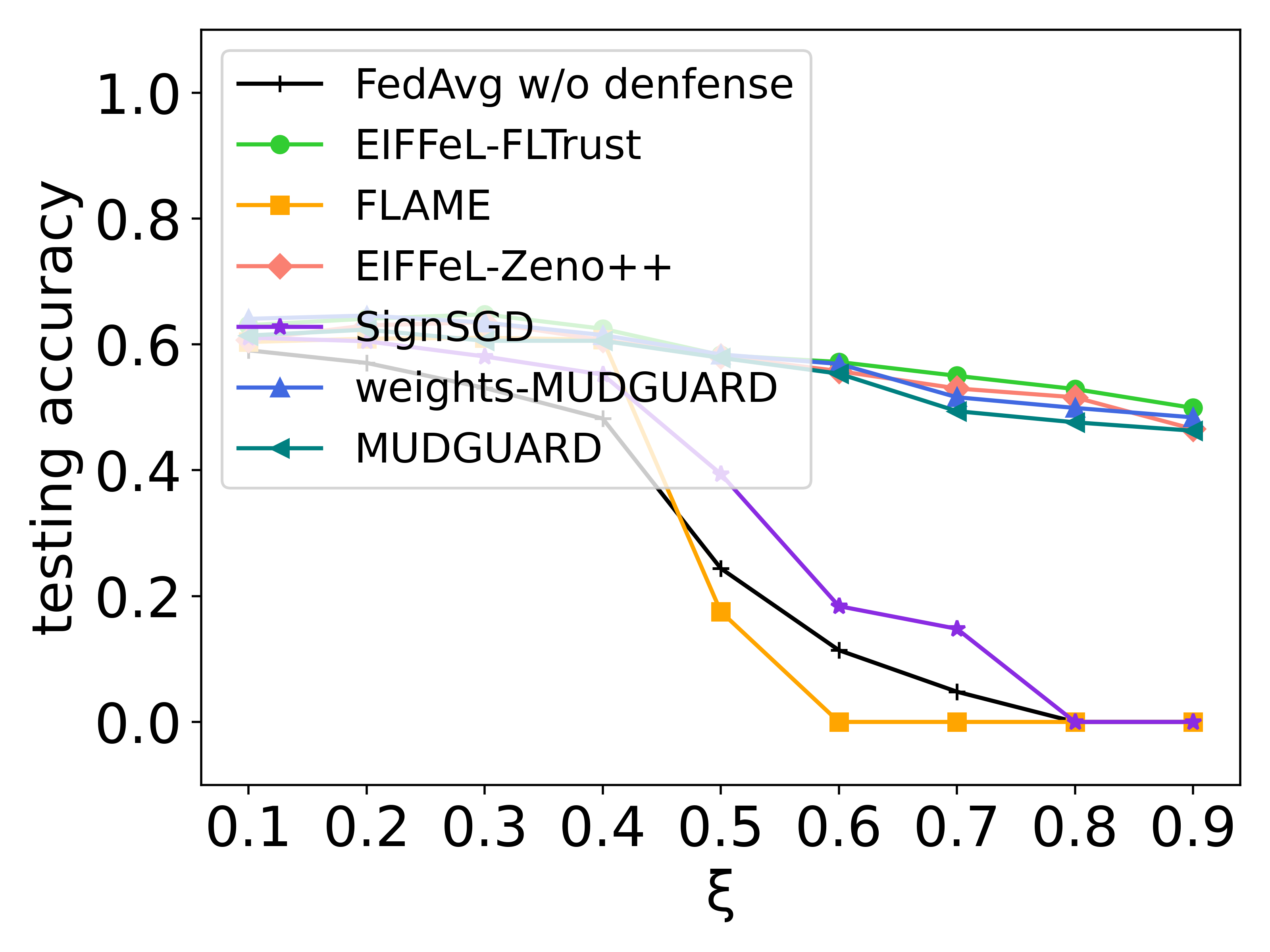}
        \caption{Label Flipping Attack}
    \end{subfigure}    
    \begin{subfigure}[b]{0.235\textwidth}
        \centering
        \includegraphics[width=1.04\textwidth]{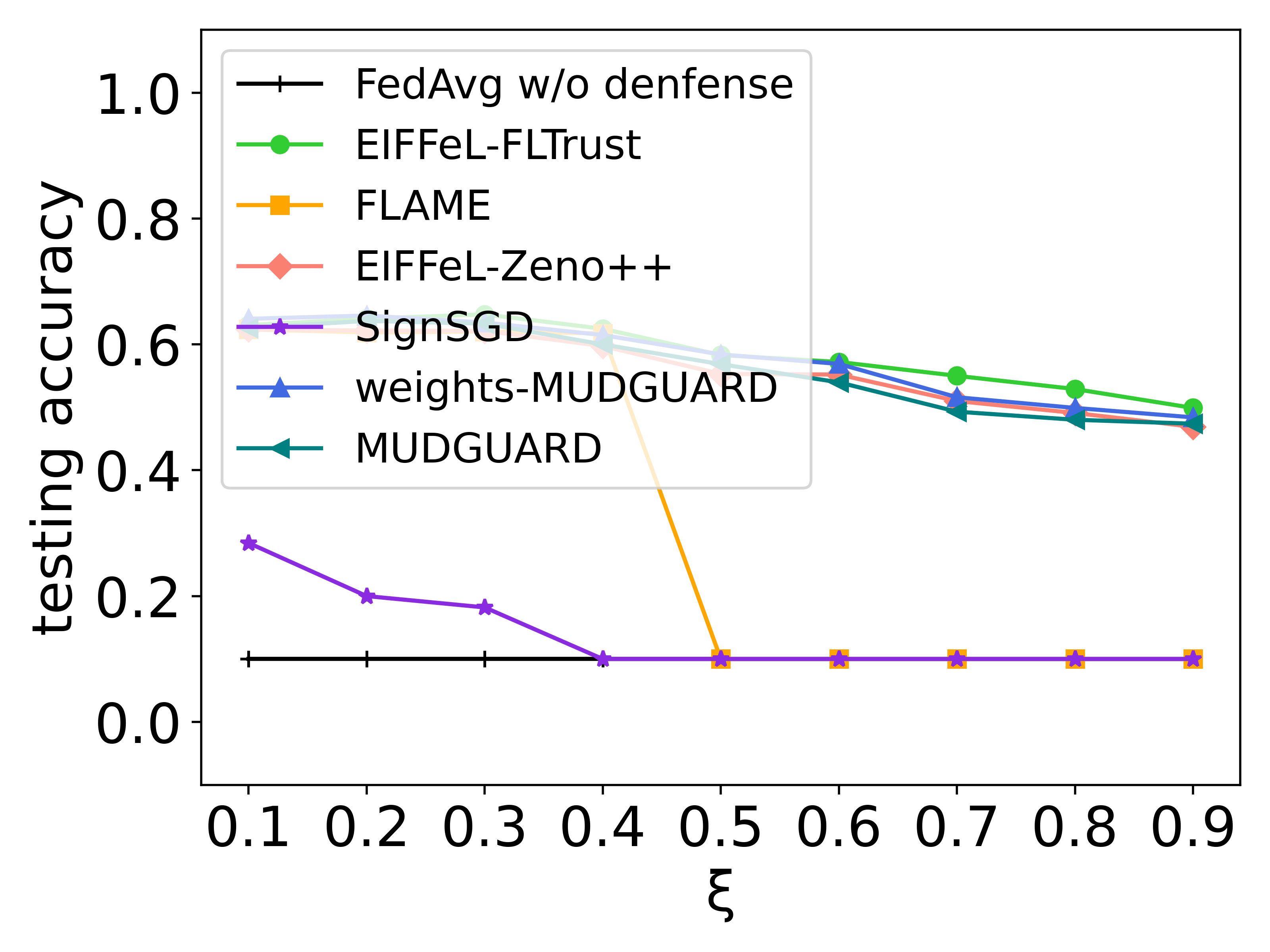}
        \caption{Krum Attack}
    \end{subfigure}    
    \begin{subfigure}[b]{0.235\textwidth}
        \centering
        \includegraphics[width=1.04\textwidth]{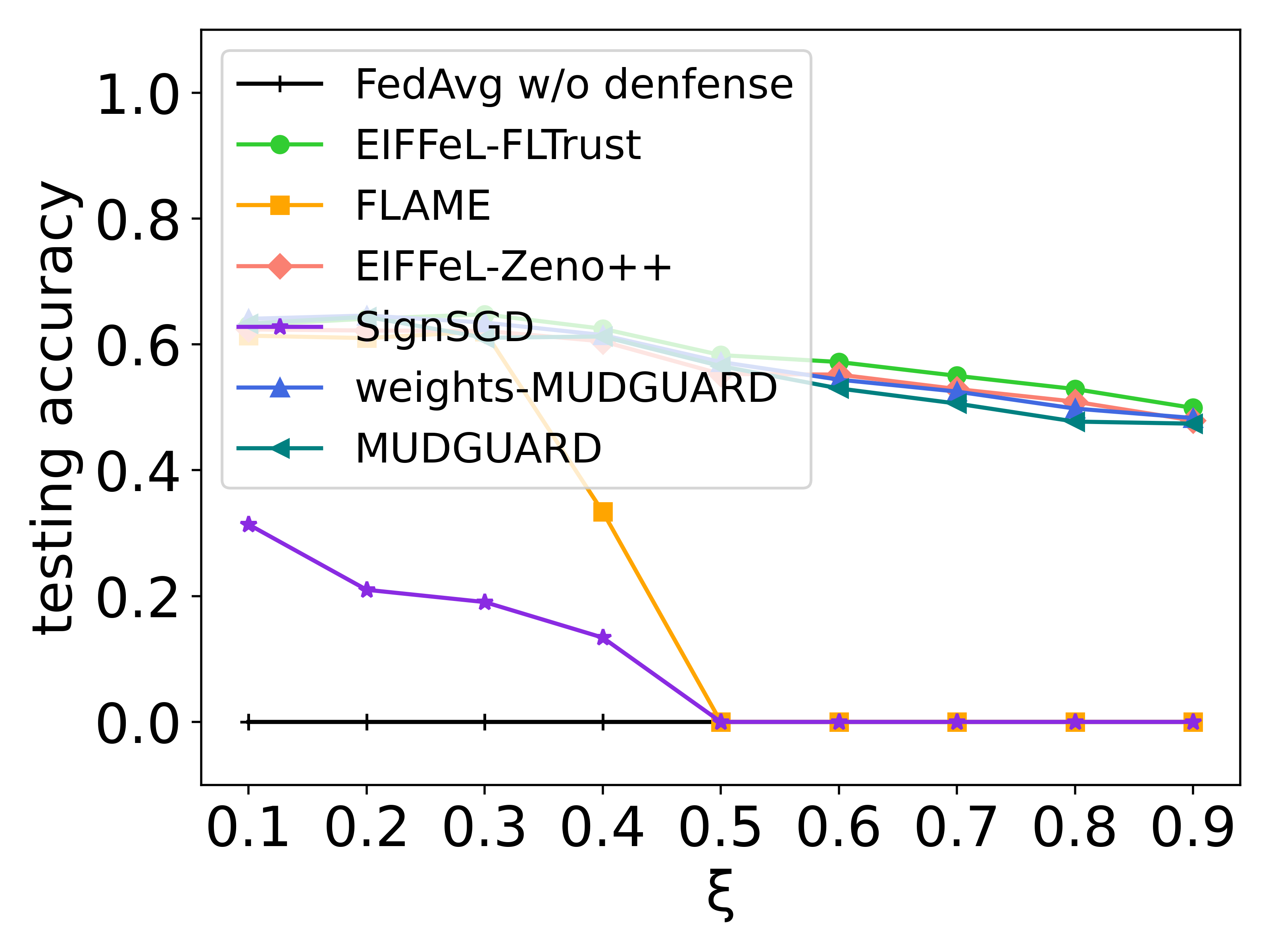}
        \caption{Trim Attack}
    \end{subfigure}  
    \begin{subfigure}[b]{0.235\textwidth}
        \centering
        \includegraphics[width=1.04\textwidth]{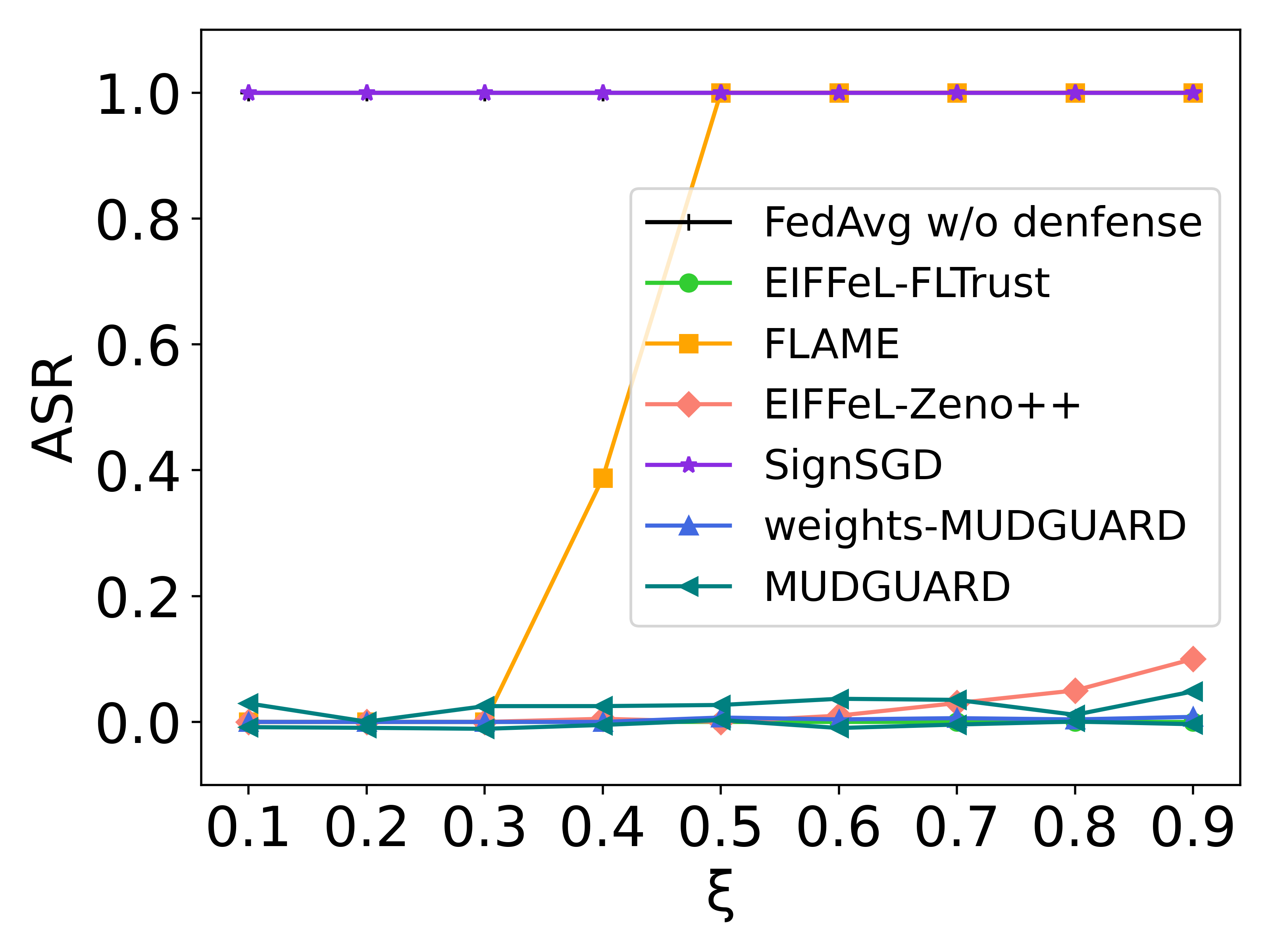}
        \caption{Adaptive Attack}
    \end{subfigure}  
    \begin{subfigure}[b]{0.235\textwidth}
        \centering
        \includegraphics[width=1.04\textwidth]{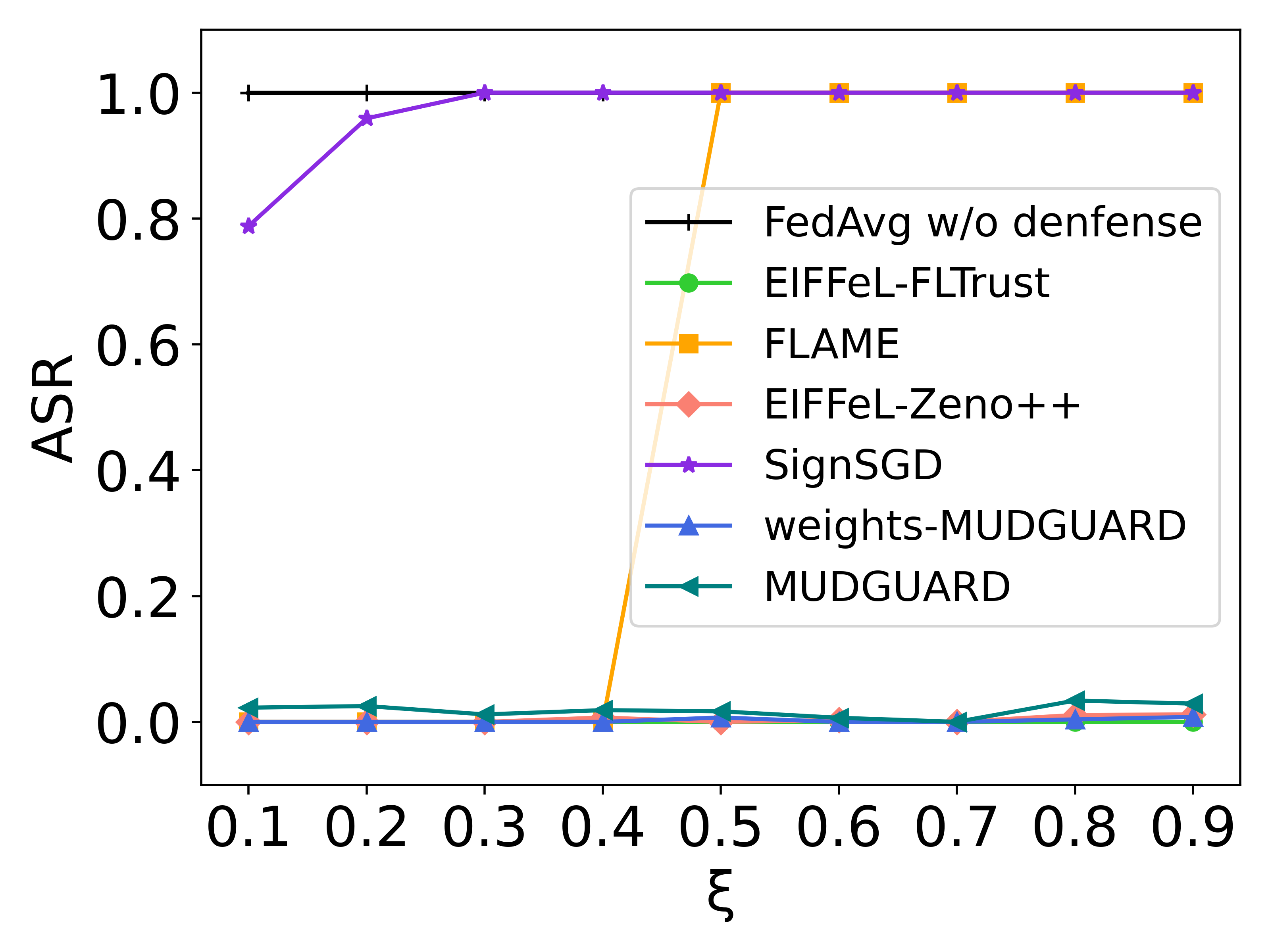}
        \caption{Backdoor Attack}
    \end{subfigure}
    \begin{subfigure}[b]{0.235\textwidth}
        \centering
        \includegraphics[width=1.04\textwidth]{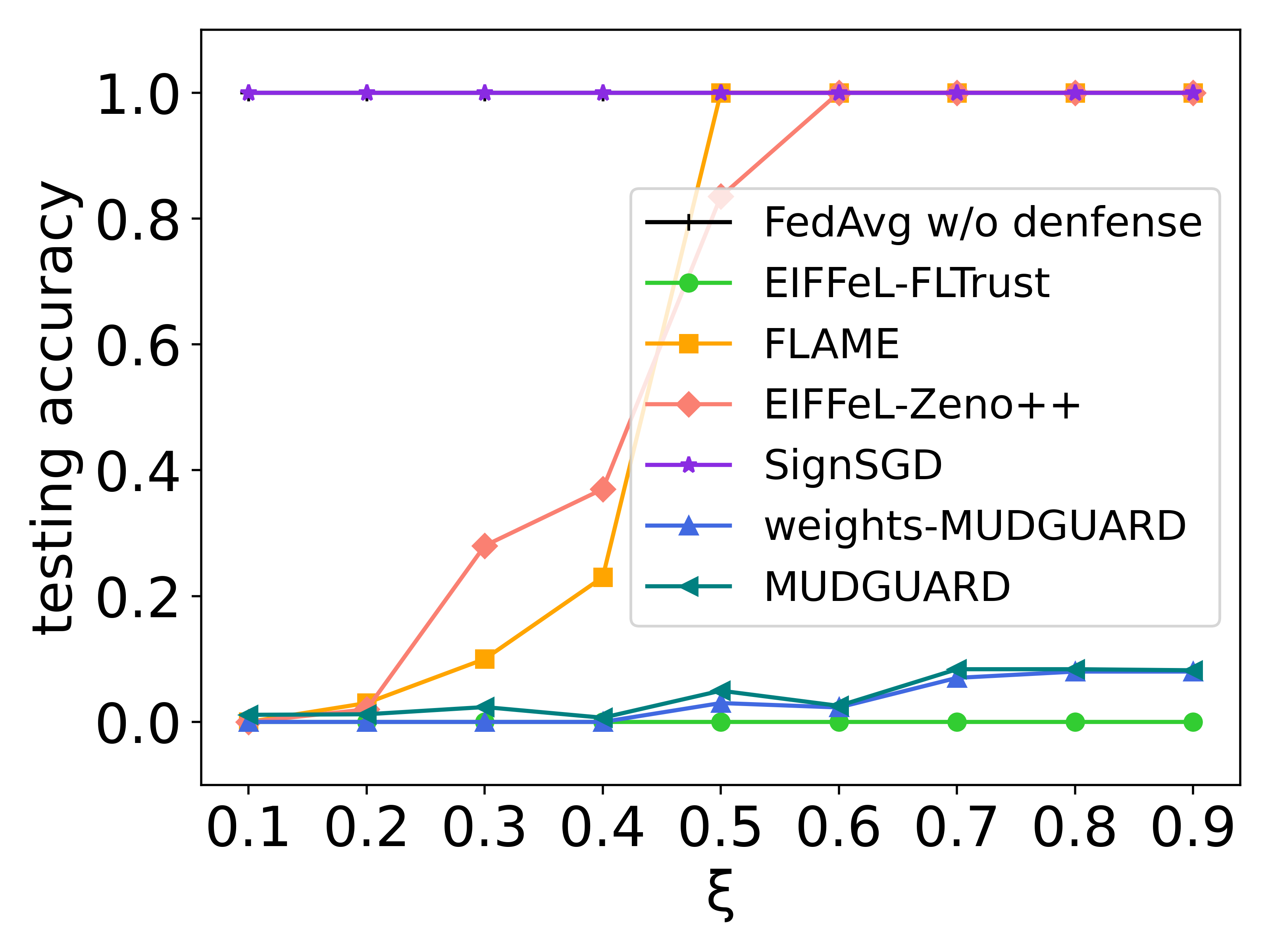}
        \caption{Edge-case Attack}
    \end{subfigure}    
  
    \caption{Comparison with Byzantine-robust methods in CIFAR-10 by $\xi=0.1 - 0.9$.}
    \label{fig:compcifar10}
\end{figure}

\noindent\textbf{Impact of $\alpha$ on robustness of \texttt{MUDGUARD}} 
Turning back to Section~\ref{sec:byagg} and Theorem~\ref{the:1}, the Byzantine-robustness of \texttt{MUDGUARD} relies on whether we can choose a desirable density $\alpha$. From the experimental results in Section~\ref{sec:expacc}, we know that EA has a stronger stealthiness than other attacks. Therefore, we use EA as an example in this section to analyze the influence of $\alpha$ on \texttt{MUDGUARD}. 
Figure~\ref{fig:eaalpha} shows how the ASR and clustering accuracy varies in semi-honest and malicious groups under EA when we change $\alpha$. Consistent with Theorem~\ref{the:1}, once $\alpha$ is great than $\sqrt{2}$, the malicious and semi-honest clients will cluster together, resulting in \texttt{MUDGUARD} loss of the effectiveness of Byzantine-robustness. A similar situation can be found when $\alpha$ is set as too small (in this case, all clients will be identified as noise.). From the experiment, we found 1 is a best practice for $\alpha$, so we set it as the default parameter. Nevertheless, \texttt{MUDGUARD} cannot 100\% guarantee to exclude EA because EA has a strong stealthiness. We will leave this as future work.

\begin{figure}[!ht]
    \centering
    \begin{subfigure}[b]{0.235\textwidth}
        \centering
        \includegraphics[width=1.04\textwidth]{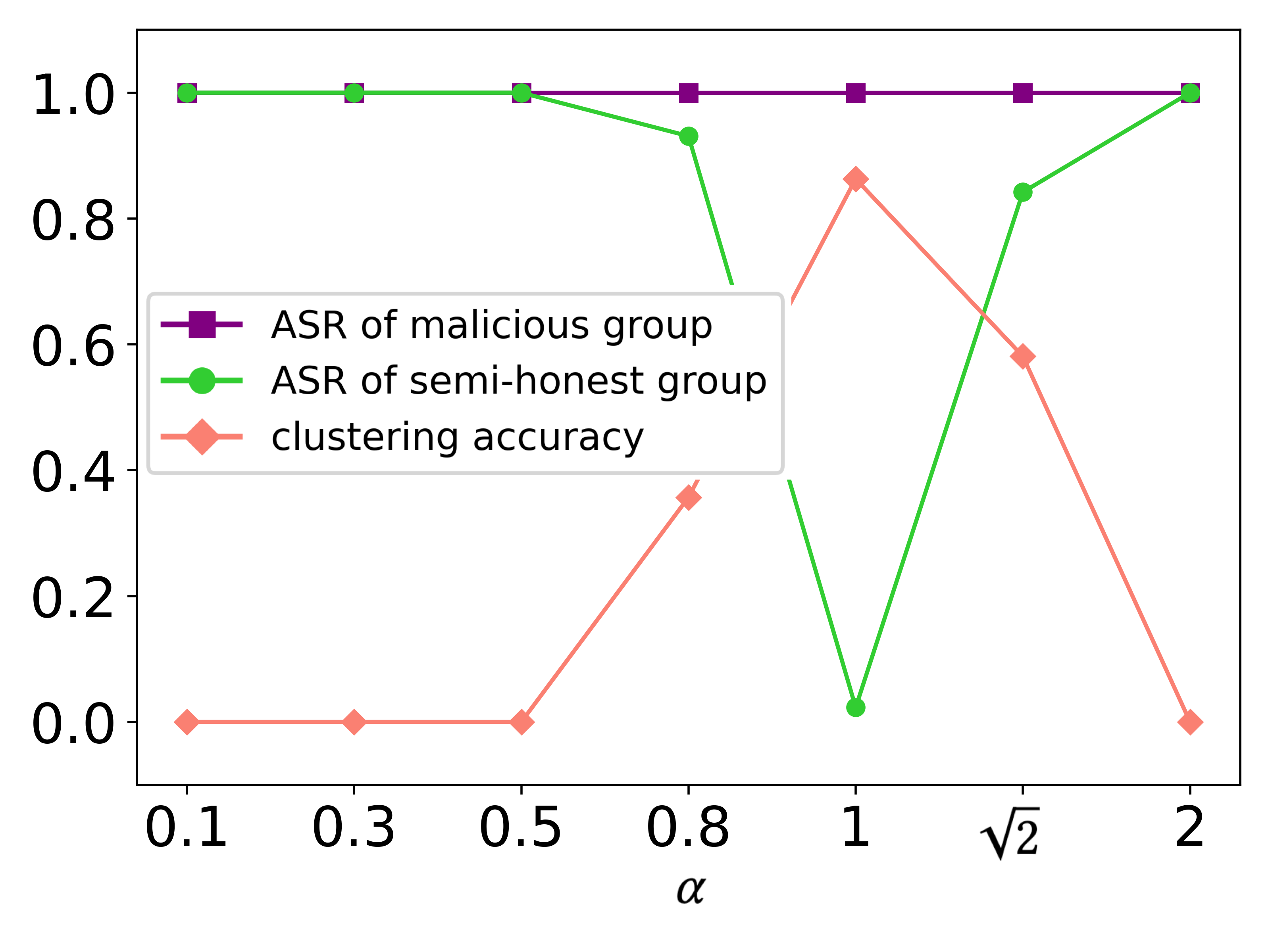}
        \caption{MNIST}
    \end{subfigure}
    \begin{subfigure}[b]{0.235\textwidth}
        \centering
        \includegraphics[width=1.04\textwidth]{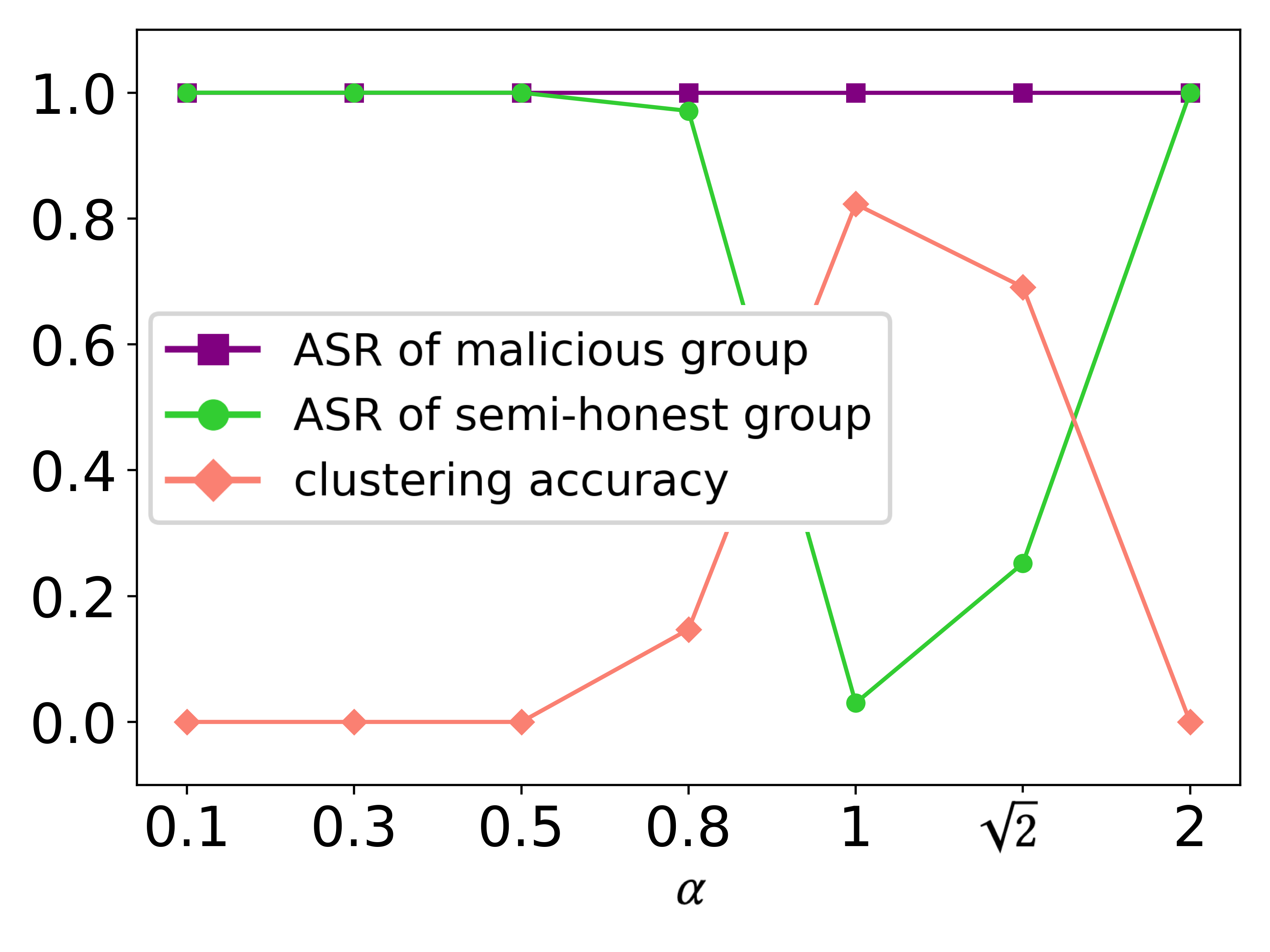}
        \caption{CIFAR-10}
    \end{subfigure}    
  
    \caption{Impact of $\alpha$ on the robustness of \texttt{MUDGUARD}, where EA and default settings are used.}
    \label{fig:eaalpha}
\end{figure}

\noindent\textbf{Impact of $\lambda$ on Adaptive Attack.}
Figure~\ref{fig:aalambdamnist} shows how the ASR varies in semi-honest and malicious groups when we adapt BA to \texttt{MUDGUARD}, where (F)MNIST, CIFAR-10 and default settings in Table~\ref{tab:flsetting} are used. 
In MNIST (Figure~\ref{fig:aalambdamnist}a), the ASR of the semi-honest group remains at a low level (nearly 0\%) while that of the malicious group raises from 0.23 to 1 as $\lambda$ grows from 0.1 to 0.3. 
This is so because when the value of $\lambda$ is low, the malicious clients using AA more focus on evading filtering (i.e., inducing a drop in clustering accuracy). 
Even if malicious clients are grouped with semi-honest clients, they cannot produce practical attack effectiveness. 
When the value of $\lambda$ gradually increases, the malicious clients will focus more on attack performance. 
In this way, \texttt{MUDGUARD} will easily distinguish the malicious from the semi-honest. 
It thus can resist AA.  
Note the experimental results with FMNIST and CIFAR-10 (Figure~\ref{fig:aalambdamnist}b and c) share the same trend with MNIST (Figure~\ref{fig:aalambdamnist}a).
\begin{figure}[t]
    \centering
    \begin{subfigure}[b]{0.235\textwidth}
        \centering
        \includegraphics[width=1.04\textwidth]{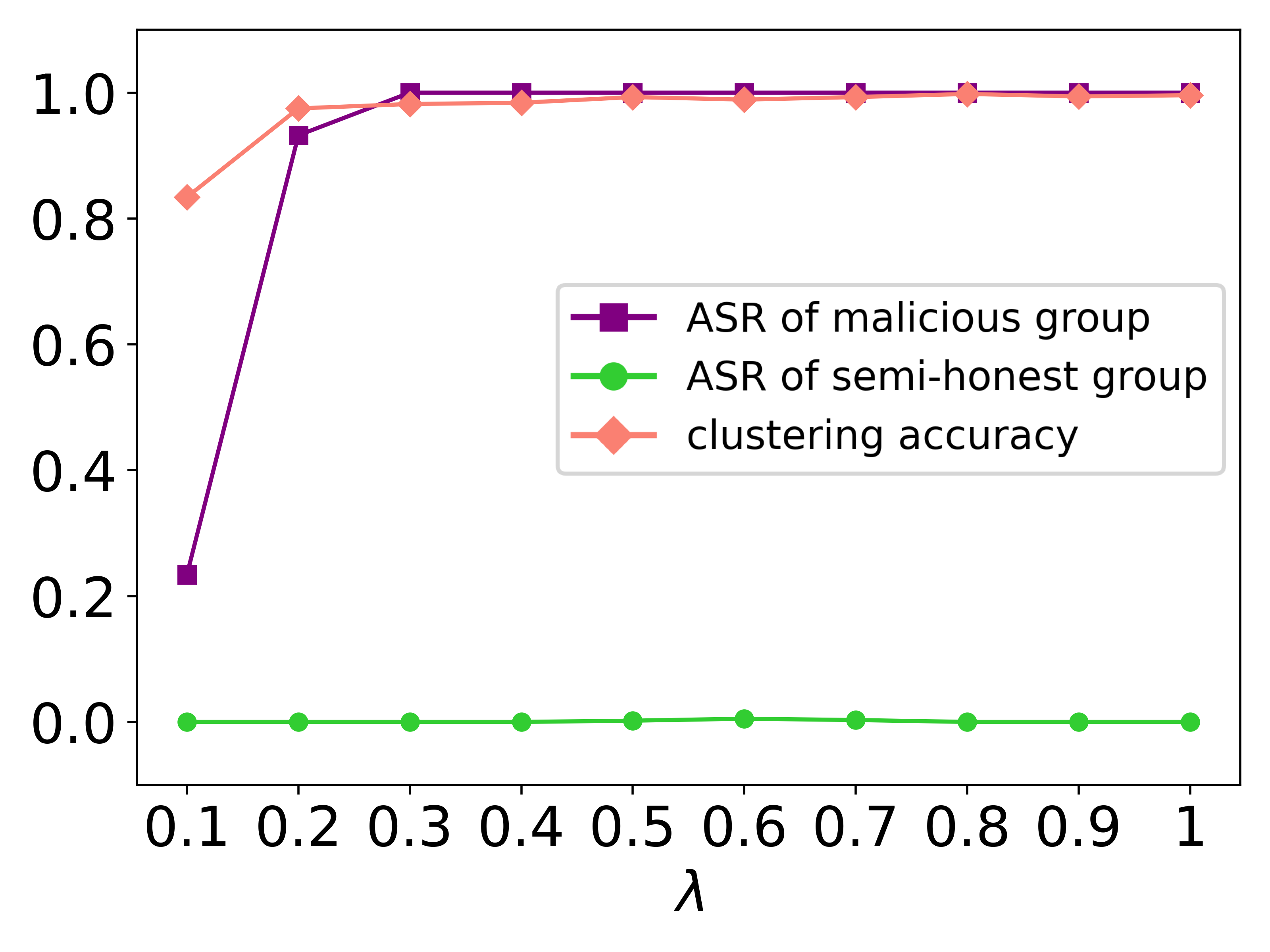}
        \caption{MNIST}
    \end{subfigure}
    \begin{subfigure}[b]{0.235\textwidth}
        \centering
        \includegraphics[width=1.04\textwidth]{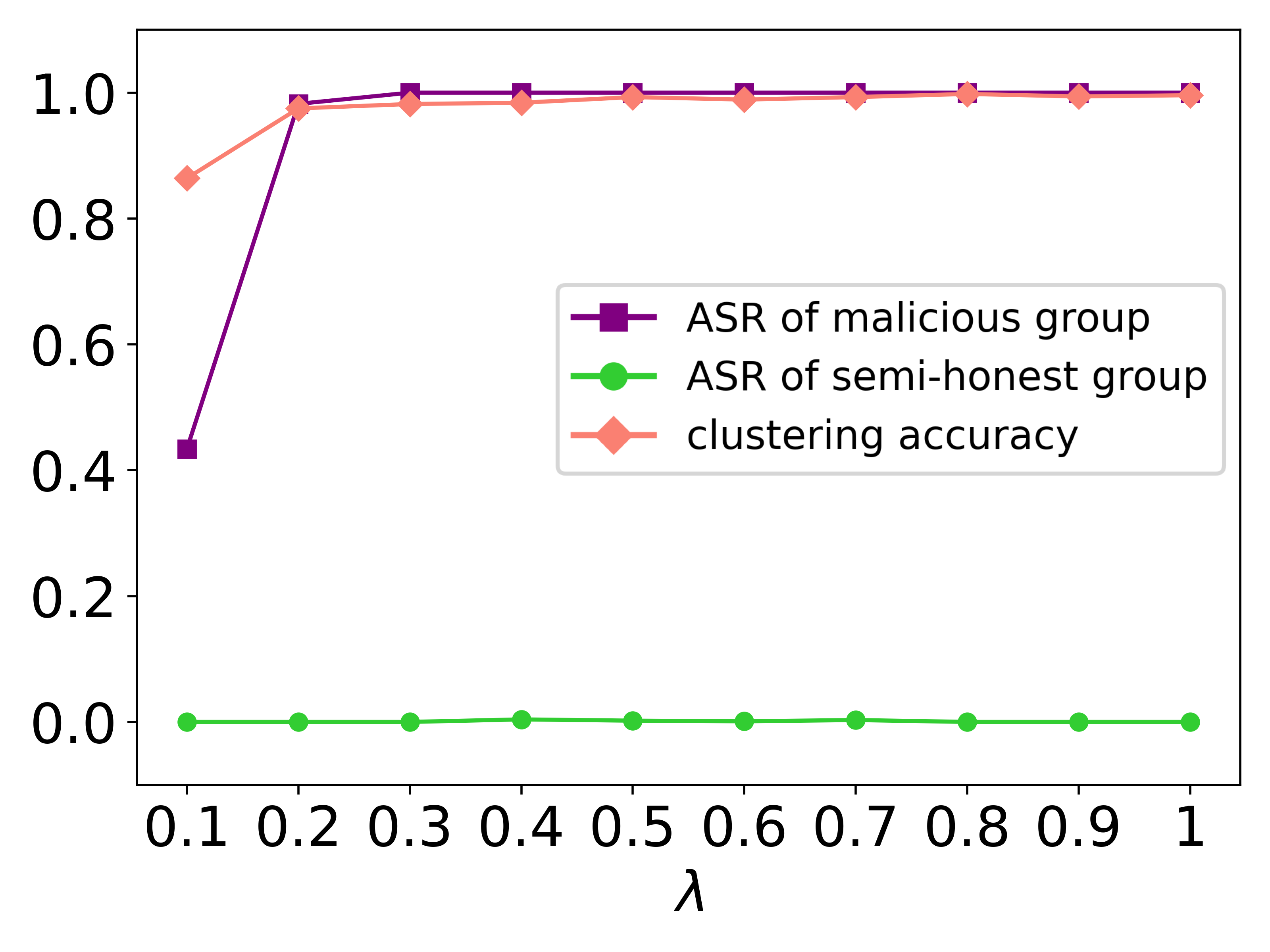}
        \caption{FMNIST}
    \end{subfigure}    
    \begin{subfigure}[b]{0.235\textwidth}
        \centering
        \includegraphics[width=1.04\textwidth]{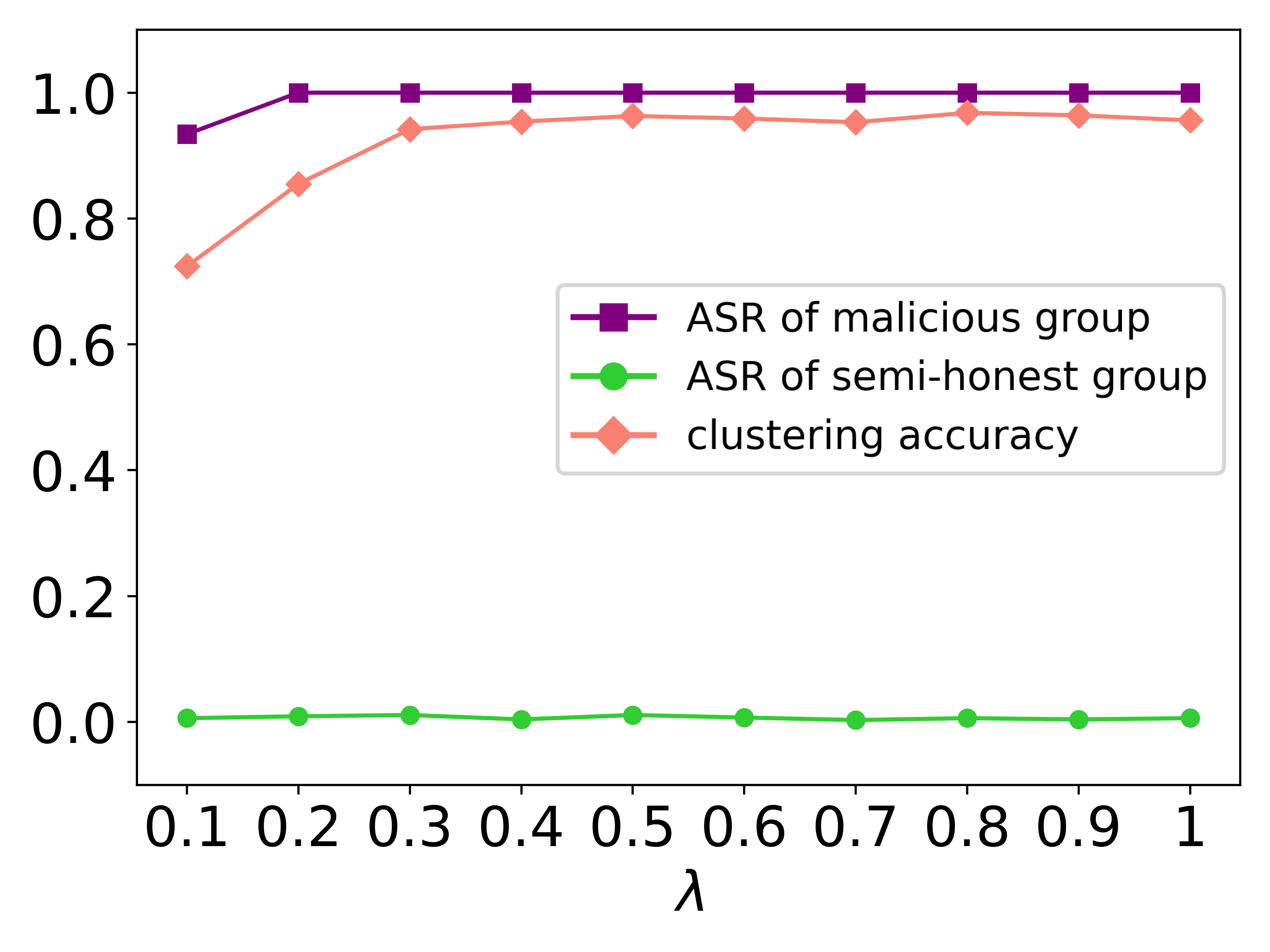}
        \caption{CIFAR-10}
    \end{subfigure}    
  
    \caption{Impact of $\lambda$ on Adaptive Attack, where backdoor attack and default settings are used. }
    \label{fig:aalambdamnist}
\end{figure}

%% file: sections/other_discussions.tex
\subsection{Discussions}
\label{sec:otherdis}

\subsubsection{Boosting Accuracy} Intuitively, CIFAR-10 is more difficult to train than (F)MNIST. SignSGD~\cite{bernstein2018signsgd} achieves 90\% accuracy with centralized learning. 
The reasons for the accuracy drop in FL are (1) the number of clients: when the number increases, the amount of data allocated to a client becomes smaller, and the local model is more prone to overfitting;  
(2) Data heterogeneity: a larger degree of non-iid leads to less information being available for the local model to learn, which harms accuracy;  
(3) Malicious clients: malicious clients taking a part of the training data can be filtered out by \texttt{MUDGUARD}. 
This is equivalent to discarding that part of the training data. 
Additional improvements in accuracy could be obtained by using different models, e.g., ViT~\cite{dosovitskiy2020image}, or by performing hyperparameters grid search for FL, etc. 

\subsubsection{Comparison of Overheads with \texttt{FLAME}} Intuitively, systems that use clustering, like \texttt{FLAME}~\cite{nguyen2022flame}, could experience efficiency problem if left without proper optimization.  
In our design, MPC impacts three main phases: computation of $\cosm$, $L_2$ distance, and element-wise comparison. 
The 1\emph{st} stage is the most computationally expensive. We, therefore, focus on its optimization.
After the optimization, the servers avoid using ``heavy" tools, like HE and Beaver’s multiplication, to calculate cosine so that communication and computational overheads are naturally reduced. Therefore, \texttt{MUDGUARD} has much lighter complexity than \texttt{FLAME}. 

Note that \texttt{FLAME} uses a 2-party semi-honest MPC scheme which is very similar to the semi-honest situation in our design (Table~\ref{tab:overheads}). 
Our optimized version of \texttt{MUDGUARD} is less computationally and communicationally complex than \texttt{FLAME}. 
In \texttt{MUDGUARD}, $\cosm$ is calculated by doing XOR locally on servers, and the matrix size reduces from the number of clients $\times$ the number of model updates ($n\times d$) to $n\times n$ after calculation.  
Then this $n\times n$ matrix is used to calculate the pairwise-$L_2$ distance (where only multiplication is involved). 
In \texttt{FLAME}, directly calculating cosine similarity requires multiplication, division, and the square root of the $n\times d$ matrix, which are relatively intensive, expensive operations in MPC. 
The fact that \texttt{FLAME}'s code is not publicly available prevents a code/implementation-based comparison.

We followed state-of-the-art Byzantine-robust FL~\cite{cao2020fltrust}'s default setting on the client number (number = 100).
Table~\ref{tab:overheads} shows overheads (the runtime and communication costs) on the server side are acceptable; in particular, only 0.4s and 16 MB are required with LeNet in the semi-honest case.

\subsubsection{Defending Against Other Attacks} In the experiments, we consider SOTA (un)targeted attacks.   
We say that interested readers may use other attacks to test \texttt{MUDGUARD}, in which Byzantine-robustness could not be seriously affected. 
We take the Distributed Backdoor Attack (DBA)~\cite{Xie2020DBA} as an example. 
DBA decomposes a global trigger into several pieces distributed to local clients. 
It, however, yields significant changes to some dimensions of updates to maintain the ASR of the backdoor task.
Since the cosine distance between malicious and benign updates are distinguishable, \texttt{MUDGUARD} can still work well under DBA.
Note another attack, Little Is Enough~\cite{baruch2019little}, have not been considered in this work 
because it requires attackers to have knowledge of the gradients of semi-honest clients, which violates privacy preservation. 

\subsubsection{Advantages of Adjusted Cosine Similarity}
\label{sec:advofcosm}
\begin{figure}[!ht]
    \centering
    \includegraphics[width=0.46\textwidth]{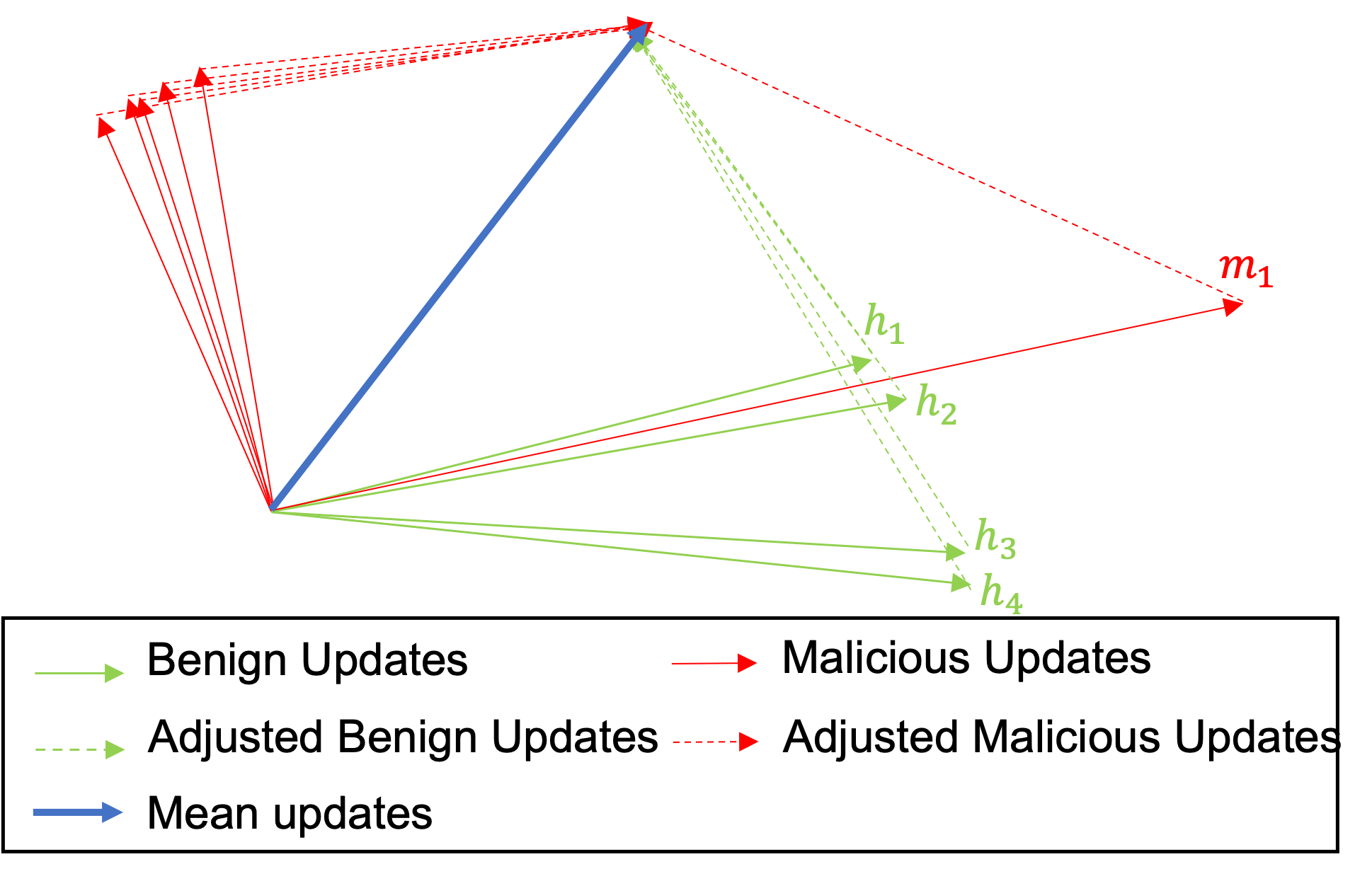}
    \caption{An example of calculation of adjusted cosine similarity}
    \label{fig:adjusted cosine}
\end{figure}

As shown in Figure~\ref{fig:adjusted cosine}, we present an example of the calculation of adjusted cosine similarity. It is clear to see that adjusted cosine similarity is able to capture the magnitudes and directions of updates by transferring updates to adjusted updates. Although $m_1$ does not have too many differences in directions (i.e., $m_1$ will be clustered with $h_1$ and $h_2$), its differences with $h_1$ and $h_2$ in magnitudes can be easily captured by $\cosm$. Furthermore, due to non-iid, the ($h_1$, $h_2$) and ($h_3$, $h_4$) will be clustered into two groups if cosine distance is applied. However, by subtracting mean updates, this influence can be reduced. 

We also show the experimental results in Figure~\ref{fig:pairwise distance}, where FMNIST is used under BA and the number of clients is set to 10 (the first six are benign clients, the rest are malicious clients). The rest of the default settings follow Table~\ref{tab:flsetting}. From Figure~\ref{fig:adjusted cosine}, we can see that if only cosine distance is calculated, honest updates will be classified as noise due to the influence of non-iid. The adjusted cosine similarity weakens this effect (Figure~\ref{fig:adjusted cosine}b). Calculating the $L_2$ distance again will make the distinction between the two groups more obvious (Figure~\ref{fig:adjusted cosine}c).
\begin{figure}[!ht]
    \centering
    \begin{subfigure}[b]{0.5\textwidth}
        \centering
        \includegraphics[width=0.9\textwidth]{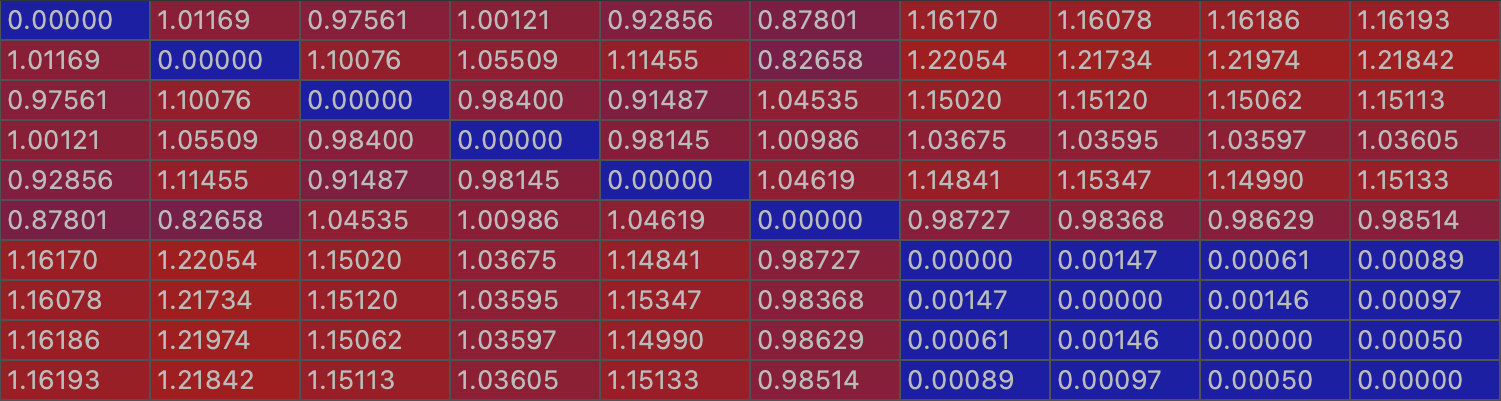}
        \caption{Pairwise cosine distance}
    \end{subfigure}
    \begin{subfigure}[b]{0.5\textwidth}
        \centering
        \includegraphics[width=0.9\textwidth]{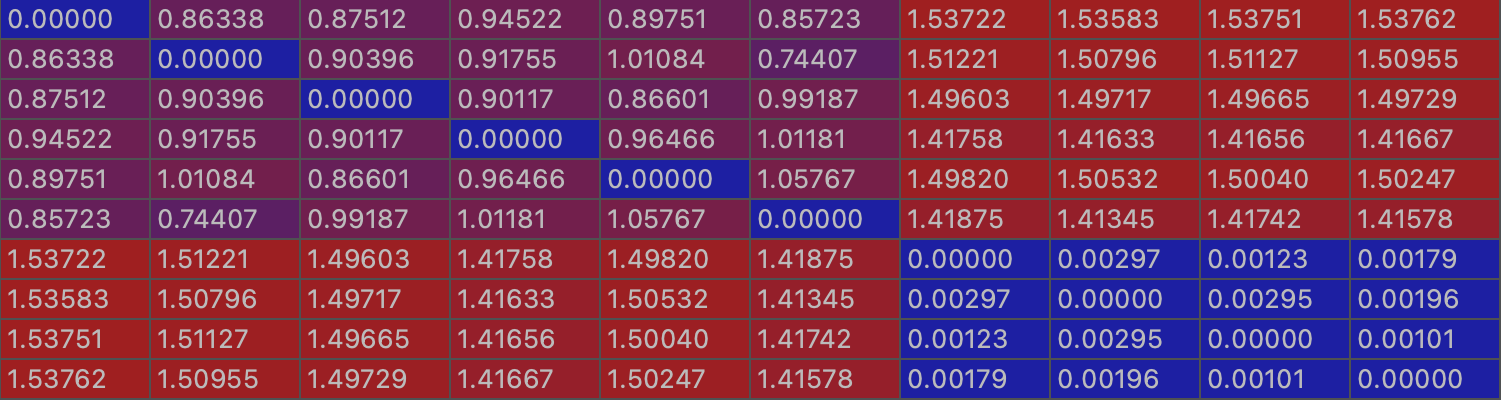}
        \caption{Pairwise adjusted cosine distance ($\cosm$)}
    \end{subfigure}    
    \begin{subfigure}[b]{0.5\textwidth}
        \centering
        \includegraphics[width=0.9\textwidth]{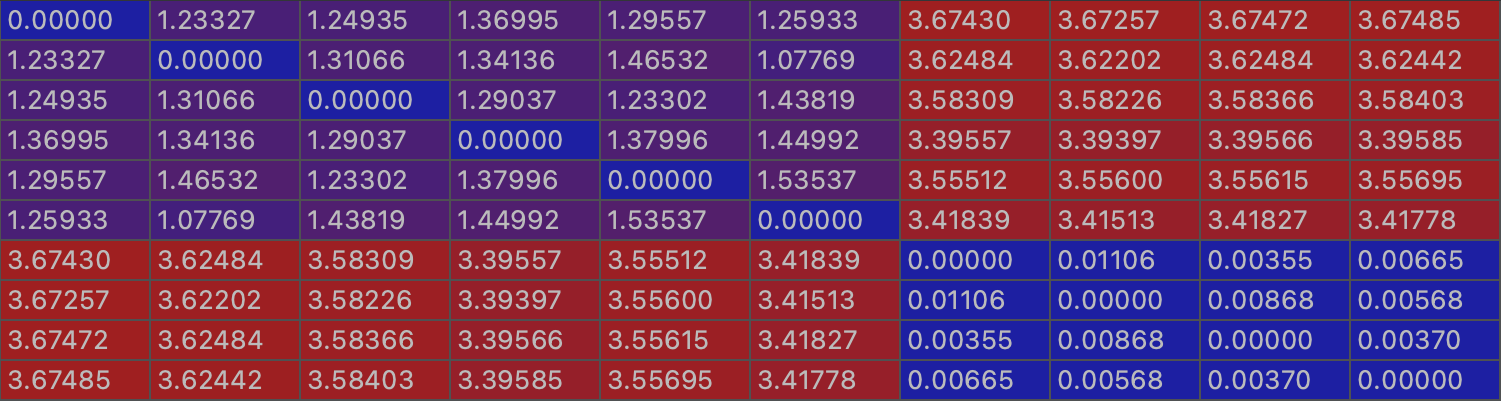}
        \caption{Pairwise $L_2$ distance for $\cosm$}
    \end{subfigure}    
  
    \caption{Calculation results of pairwise distance. }
    \label{fig:pairwise distance}
\end{figure}

Furthermore, we provide a comparison when \texttt{MUDGUARD} uses cosine similarity and adjusted cosine similarity for clustering in Figure~\ref{fig:cosvsadcos}. It is clear to see that the testing accuracy of \texttt{MUDGUARD} with cosine similarity abruptly goes down when the model approaches convergence. On the contrary, this does not happen in \texttt{MUDGUARD} with adjusted cosine similarity. The detailed explanation is given in Section~\ref{sec:byagg}.
\begin{figure}[t]
    \centering
    \includegraphics[width=0.36\textwidth]{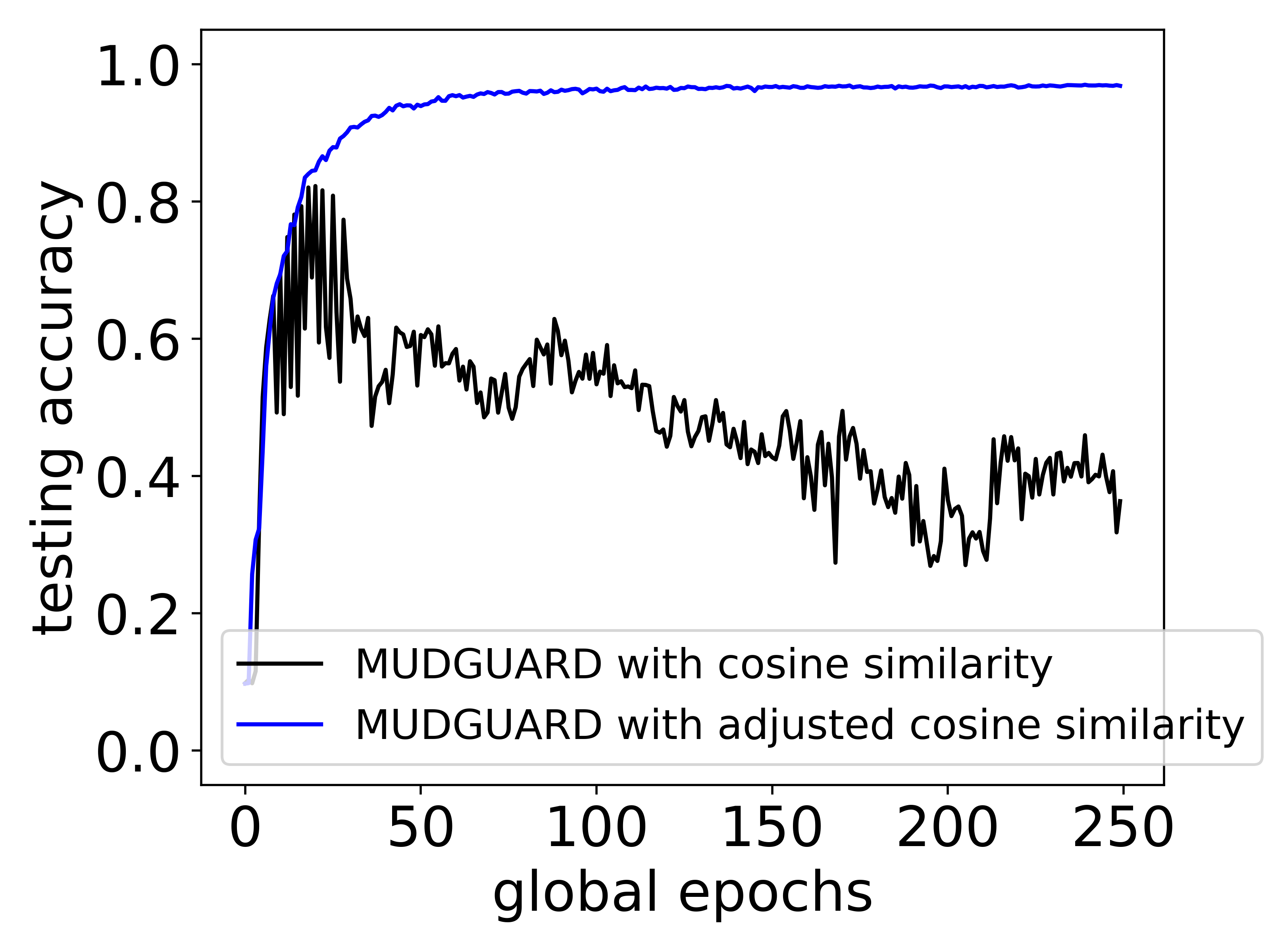}
    \caption{Comparison of MUDGUARD with cosine similarity and adjusted cosine similarity under GA.}
    \label{fig:cosvsadcos}
\end{figure}

\subsubsection{Dynamic Attacks}

Recall that \texttt{MUDGUARD} can perform well in terms of testing accuracy under the assumption that malicious clients consistently perform one type of attack (e.g., GA) throughout the whole training period.
It can also perform well if we allow malicious clients to perform different attacks on the epochs, e.g., GA to the first 10 epochs and then Krum attacks to the remaining epochs.
We notice that all the attacks (we consider in this work) except GA may require several epochs of training (as a buffer) to produce attack effects. 
But these buffer epochs can boost \texttt{MUDGUARD}'s clustering performance. 
This is so because the clustering ability is enhanced with the increase of training rounds. 
On the other hand, the TNR and TPR (of the clustering) could be relatively low in these epochs. 
Some malicious clients can be clustered into a semi-honest group,
but this will not seriously affect the accuracy performance of the model.

One may think malicious clients are allowed to perform all the attacks in a single epoch. 
But so far, it is unknown how to group those attacks together friendly and meanwhile maximize their attack effects. 
In practice, the attacks may deliver an update in different directions, and further, they may even yield influence on each other. 
For example, GA could easily destroy the convergence of BA.
We say that it is an interesting open problem to consider launching GA, LFA, Krum, Trim, and AA attacks together in an epoch to evaluate the accuracy and ARS. 

\subsubsection{Varying Clients Subsampling Rate}
We assert that \texttt{MUDGUARD} can perform well under different clients' subsampling rates. 
This benefits from the proposed \emph{Model Segmentation} that can resist malicious-majority clients. 
Imagine that in the context of an honest majority (e.g., 40 out of 100 are malicious), if the subsampling rate is set relatively low (e.g., 10\%), we eventually have a malicious majority case in one training epoch with a high probability. 
Our experimental results have demonstrated that \texttt{MUDGUARD} can still achieve Byzantine-robustness under a low subsampling rate. 

\subsubsection{Learning Rate and Local Epoch}

They are subtle parameters that decide the performance of FL training. 
A low learning rate can slow down the speed of convergence, and on the other hand, a high rate hinders the model's convergence, harming accuracy. 
As for the local epoch, in the case of iid, if carefully increasing the number of epochs, we can make the model converge fast. 
But under a large degree of non-iid (e.g., q=0.5), the increase of epoch leads to updates in different directions, making the \texttt{FedAvg} algorithm invalid~\cite{mcmahan2017communication}. 
In this work, we set these two parameters according to the recommendations given in \cite{mcmahan2017communication, bernstein2018signsgd}. 
Exploring their impacts on training is orthogonal to the main focus of this work.

\subsubsection{Privacy-preserving DBSCAN}
This is one of the core parts we used to build \texttt{MUDGUARD}.
It can apply to other real-world domains, e.g., anomaly detection and encrypted traffic analytics, where data should be clustered securely.
But we note that the current \texttt{MUDGUARD} with optimization may not scale well in these applications.  %
We did the optimization for the sake of efficiency by using SignSGD and binary secret sharing, which cannot support precise floating-point arithmetic.

\subsubsection{Test Datasets}
We notice that the EA is not applicable for FMINST since the research work~\cite{10.5555/3495724.3497072} did not provide a backdoor dataset.
We leave this as an open problem. 
Based on the performance of \texttt{MUDGUARD} on MNIST and CIFAR-10, we claim that even in FMNIST, the ASR of EA stays low, which is consistent with our main conclusion.

In the experiments, we conduct three image datasets (MNIST, FMNIST, and CIFAR-10).
We claim that \texttt{MUDGUARD} can be further used to capture the Byzantine-robust and privacy-preserving features of the models trained on text and speech datasets, where the text dataset for the next-word prediction task can require Recurrent Neural Networks (RNNs).
To train these datasets, we can also use weights or gradients to update the models (in the context of FL). 
If the models or datasets are poisoned (by malicious clients), malicious updates have differences in updates directions from benign ones.
Then we can use \texttt{MUDGUARD} to protect benign clients. 
We could use other types of datasets during training, but this will not affect our conclusions on Byzantine robustness and privacy preservation. 

\subsubsection{ Limitations} 
\emph{Using weights as updates.} To provide cost-effective secure computations, the proposed \texttt{MUDGUARD} only implements the update method by SignSGD. 
If we use weights as updates, the secret shares sent to the servers will be in floating-point format.
In this case, we will have to downgrade the design to the unoptimized \texttt{MUDGUARD} in Table~\ref{tab:overheads}, which could cause a considerable amount of both communication costs and runtime. 
An interesting future work could be to propose a more lightweight (than the current design) and secure MPC framework for \texttt{MUDGUARD}.

\emph{Performance under EA.} Although \texttt{MUDGUARD} does achieve good performance in terms of accuracy and (to some extent) efficiency, under the EA, \texttt{MUDGUARD}'s ASR cannot be eventually reduced to nearly 0\%.
In future work, we will improve the TNR and TPR of the clustering algorithm so as to recognize subtle differences between malicious and semi-honest clients.

%% file: bare_conf_NDSS2024.bbl
\begin{thebibliography}{10}
\providecommand{\url}[1]{#1}
\csname url@samestyle\endcsname
\providecommand{\newblock}{\relax}
\providecommand{\bibinfo}[2]{#2}
\providecommand{\BIBentrySTDinterwordspacing}{\spaceskip=0pt\relax}
\providecommand{\BIBentryALTinterwordstretchfactor}{4}
\providecommand{\BIBentryALTinterwordspacing}{\spaceskip=\fontdimen2\font plus
\BIBentryALTinterwordstretchfactor\fontdimen3\font minus
  \fontdimen4\font\relax}
\providecommand{\BIBforeignlanguage}[2]{{%
\expandafter\ifx\csname l@#1\endcsname\relax
\typeout{** WARNING: IEEEtranS.bst: No hyphenation pattern has been}%
\typeout{** loaded for the language `#1'. Using the pattern for}%
\typeout{** the default language instead.}%
\else
\language=\csname l@#1\endcsname
\fi
#2}}
\providecommand{\BIBdecl}{\relax}
\BIBdecl

\bibitem{abadi2016deep}
M.~Abadi, A.~Chu, I.~Goodfellow, H.~B. McMahan, I.~Mironov, K.~Talwar, and
  L.~Zhang, ``Deep learning with differential privacy,'' in \emph{CCS}, 2016,
  pp. 308--318.

\bibitem{10.1145/3243734.3243854}
T.~Araki, A.~Barak, J.~Furukawa, M.~Keller, Y.~Lindell, K.~Ohara, and
  H.~Tsuchida, ``Generalizing the spdz compiler for other protocols,'' in
  \emph{CCS}, 2018.

\bibitem{bagdasaryan2020backdoor}
E.~Bagdasaryan, A.~Veit, Y.~Hua, D.~Estrin, and V.~Shmatikov, ``How to backdoor
  federated learning,'' in \emph{AISTATS}, 2020, pp. 2938--2948.

\bibitem{baruch2019little}
G.~Baruch, M.~Baruch, and Y.~Goldberg, ``A little is enough: Circumventing
  defenses for distributed learning,'' in \emph{NIPS}, 2019.

\bibitem{bellare2012foundations}
M.~Bellare, V.~T. Hoang, and P.~Rogaway, ``Foundations of garbled circuits,''
  in \emph{CCS}, 2012, pp. 784--796.

\bibitem{bernstein2018signsgd}
J.~Bernstein, Y.-X. Wang, K.~Azizzadenesheli, and A.~Anandkumar, ``signsgd:
  Compressed optimisation for non-convex problems,'' in \emph{ICML}, 2018, pp.
  560--569.

\bibitem{berrut2004barycentric}
J.-P. Berrut and L.~N. Trefethen, ``Barycentric lagrange interpolation,''
  \emph{SIAM review}, pp. 501--517, 2004.

\bibitem{bhagoji2019analyzing}
A.~N. Bhagoji, S.~Chakraborty, P.~Mittal, and S.~Calo, ``Analyzing federated
  learning through an adversarial lens,'' in \emph{ICML}, 2019, pp. 634--643.

\bibitem{biggio2012poisoning}
B.~Biggio, B.~Nelson, and P.~Laskov, ``Poisoning attacks against support vector
  machines,'' in \emph{ICML}, 2012, pp. 1467--1474.

\bibitem{blanchard2017machine}
P.~Blanchard, E.~M. El~Mhamdi, R.~Guerraoui, and J.~Stainer, ``Machine learning
  with adversaries: Byzantine tolerant gradient descent,'' in \emph{NIPS},
  2017, pp. 118--128.

\bibitem{bonawitz2017practical}
K.~Bonawitz, V.~Ivanov, B.~Kreuter, A.~Marcedone, H.~B. McMahan, S.~Patel,
  D.~Ramage, A.~Segal, and K.~Seth, ``Practical secure aggregation for
  privacy-preserving machine learning,'' in \emph{CCS}, 2017, pp. 1175--1191.

\bibitem{bottou2010large}
L.~Bottou, ``Large-scale machine learning with stochastic gradient descent,''
  in \emph{COMPSTAT}, 2010, pp. 177--186.

\bibitem{campello2013density}
R.~J. Campello, D.~Moulavi, and J.~Sander, ``Density-based clustering based on
  hierarchical density estimates,'' in \emph{PAKDD}, 2013, pp. 160--172.

\bibitem{canetti2001universally}
R.~Canetti, ``Universally composable security: A new paradigm for cryptographic
  protocols,'' in \emph{FOCS}, 2001, pp. 136--145.

\bibitem{cao2020fltrust}
X.~Cao, M.~Fang, J.~Liu, and N.~Z. Gong, ``Fltrust: Byzantine-robust federated
  learning via trust bootstrapping,'' in \emph{NDSS}, 2021.

\bibitem{cryptoeprint:2020:1330}
A.~Dalskov, D.~Escudero, and M.~Keller, ``Fantastic four:honest-majority
  four-party secure computation with malicious security,'' in \emph{USENIX
  Security}, 2021, pp. 2183--2200.

\bibitem{damgaard2019new}
I.~Damg{\aa}rd, D.~Escudero, T.~Frederiksen, M.~Keller, P.~Scholl, and
  N.~Volgushev, ``New primitives for actively-secure mpc over rings with
  applications to private machine learning,'' in \emph{IEEE S\&P}, 2019, pp.
  1102--1120.

\bibitem{damgaard2013practical}
I.~Damg{\aa}rd, M.~Keller, E.~Larraia, V.~Pastro, P.~Scholl, and N.~P. Smart,
  ``Practical covertly secure mpc for dishonest majority--or: breaking the spdz
  limits,'' in \emph{ESORICS}, 2013, pp. 1--18.

\bibitem{10.1007/978-3-642-32009-5_38}
I.~Damg{\aa}rd, V.~Pastro, N.~Smart, and S.~Zakarias, ``Multiparty computation
  from somewhat homomorphic encryption,'' in \emph{CRYPTO}, 2012.

\bibitem{dayan2021federated}
I.~Dayan, H.~R. Roth, A.~Zhong, A.~Harouni, A.~Gentili, A.~Z. Abidin, A.~Liu,
  A.~B. Costa, B.~J. Wood, C.-S. Tsai \emph{et~al.}, ``Federated learning for
  predicting clinical outcomes in patients with covid-19,'' \emph{Nature
  medicine}, pp. 1735--1743, 2021.

\bibitem{dosovitskiy2020image}
A.~Dosovitskiy, L.~Beyer, A.~Kolesnikov, D.~Weissenborn, X.~Zhai,
  T.~Unterthiner, M.~Dehghani, M.~Minderer, G.~Heigold, S.~Gelly \emph{et~al.},
  ``An image is worth 16x16 words: Transformers for image recognition at
  scale,'' in \emph{ICLR}, 2021.

\bibitem{dwork2008differential}
C.~Dwork, ``Differential privacy: A survey of results,'' in \emph{TAMC}, 2008,
  pp. 1--19.

\bibitem{elgamal1985public}
T.~ElGamal, ``A public key cryptosystem and a signature scheme based on
  discrete logarithms,'' \emph{IEEE transactions on information theory}, pp.
  469--472, 1985.

\bibitem{ester1996density}
M.~Ester, H.-P. Kriegel, J.~Sander, X.~Xu \emph{et~al.}, ``A density-based
  algorithm for discovering clusters in large spatial databases with noise.''
  in \emph{kdd}, 1996, pp. 226--231.

\bibitem{fang2020local}
M.~Fang, X.~Cao, J.~Jia, and N.~Gong, ``Local model poisoning attacks to
  $\{$Byzantine-Robust$\}$ federated learning,'' in \emph{USENIX Security},
  2020, pp. 1605--1622.

\bibitem{kim2012device}
D.~Fiore, R.~Gennaro, and V.~Pastro, ``Efficiently verifiable computation on
  encrypted data,'' in \emph{CCS}, 2014, pp. 844--855.

\bibitem{furukawa2017high}
J.~Furukawa, Y.~Lindell, A.~Nof, and O.~Weinstein, ``High-throughput secure
  three-party computation for malicious adversaries and an honest majority,''
  in \emph{EUROCRYPT}, 2017, pp. 225--255.

\bibitem{NEURIPS2020_c4ede56b}
J.~Geiping, H.~Bauermeister, H.~Dr\"{o}ge, and M.~Moeller, ``Inverting
  gradients - how easy is it to break privacy in federated learning?'' in
  \emph{NIPS}, 2020, pp. 16\,937--16\,947.

\bibitem{gentry2009fully}
C.~Gentry, \emph{A fully homomorphic encryption scheme}, 2009.

\bibitem{NEURIPS2020_e32cc80b}
A.~Ghosh, J.~Chung, D.~Yin, and K.~Ramchandran, ``An efficient framework for
  clustered federated learning,'' in \emph{NIPS}, 2020, pp. 19\,586--19\,597.

\bibitem{pmlr-v119-hamer20a}
J.~Hamer, M.~Mohri, and A.~T. Suresh, ``{F}ed{B}oost: A communication-efficient
  algorithm for federated learning,'' in \emph{ICML}, 2020, pp. 3973--3983.

\bibitem{He_2016_CVPR}
K.~He, X.~Zhang, S.~Ren, and J.~Sun, ``Deep residual learning for image
  recognition,'' in \emph{CVPR}, 2016.

\bibitem{keller2020mp}
M.~Keller, ``Mp-spdz: A versatile framework for multi-party computation,'' in
  \emph{CCS}, 2020, pp. 1575--1590.

\bibitem{kingma2014adam}
D.~P. Kingma and J.~Ba, ``Adam: A method for stochastic optimization,'' in
  \emph{ICLR}, 2015.

\bibitem{koti2021swift}
N.~Koti, M.~Pancholi, A.~Patra, and A.~Suresh, ``$\{$SWIFT$\}$: Super-fast and
  robust $\{$Privacy-Preserving$\}$ machine learning,'' in \emph{USENIX
  Security}, 2021, pp. 2651--2668.

\bibitem{krizhevsky2009learning}
A.~Krizhevsky, G.~Hinton \emph{et~al.}, ``Learning multiple layers of features
  from tiny images,'' 2009.

\bibitem{baybfed}
K.~Kumari, P.~Rieger, H.~Fereidooni, M.~Jadliwala, and A.~Sadeghi, ``Baybfed:
  Bayesian backdoor defense for federated learning,'' in \emph{IEEE Symposium
  on Security and Privacy (SP)}, 2023.

\bibitem{lecun1989backpropagation}
Y.~LeCun, B.~Boser, J.~S. Denker, D.~Henderson, R.~E. Howard, W.~Hubbard, and
  L.~D. Jackel, ``Backpropagation applied to handwritten zip code
  recognition,'' \emph{Neural computation}, pp. 541--551, 1989.

\bibitem{lecun-mnisthandwrittendigit-2010}
\BIBentryALTinterwordspacing
Y.~LeCun and C.~Cortes, ``{MNIST} handwritten digit database,'' 2010. [Online].
  Available: \url{http://yann.lecun.com/exdb/mnist/}
\BIBentrySTDinterwordspacing

\bibitem{lehmkuhl2021muse}
R.~Lehmkuhl, P.~Mishra, A.~Srinivasan, and R.~A. Popa, ``Muse: Secure inference
  resilient to malicious clients,'' in \emph{USENIX Security}, 2021, pp.
  2201--2218.

\bibitem{lindell2017framework}
Y.~Lindell and A.~Nof, ``A framework for constructing fast mpc over arithmetic
  circuits with malicious adversaries and an honest-majority,'' in \emph{CCS},
  2017, pp. 259--276.

\bibitem{luo2021feature}
X.~Luo, Y.~Wu, X.~Xiao, and B.~C. Ooi, ``Feature inference attack on model
  predictions in vertical federated learning,'' in \emph{ICDE}, 2021, pp.
  181--192.

\bibitem{mcmahan2017communication}
B.~McMahan, E.~Moore, D.~Ramage, S.~Hampson, and B.~A. y~Arcas,
  ``Communication-efficient learning of deep networks from decentralized
  data,'' in \emph{AISTATS}, 2017, pp. 1273--1282.

\bibitem{guerraoui2018hidden}
E.~M.~E. Mhamdi, R.~Guerraoui, and S.~Rouault, ``The hidden vulnerability of
  distributed learning in byzantium,'' in \emph{ICML}, 2018, pp. 3521--3530.

\bibitem{mohassel2018aby3}
P.~Mohassel and P.~Rindal, ``Aby3: A mixed protocol framework for machine
  learning,'' in \emph{CCS}, 2018, pp. 35--52.

\bibitem{7958569}
P.~Mohassel and Y.~Zhang, ``Secureml: A system for scalable privacy-preserving
  machine learning,'' in \emph{IEEE S\&P}, 2017, pp. 19--38.

\bibitem{nasr2019comprehensive}
M.~Nasr, R.~Shokri, and A.~Houmansadr, ``Comprehensive privacy analysis of deep
  learning: Passive and active white-box inference attacks against centralized
  and federated learning,'' in \emph{IEEE S\&P}, 2019, pp. 739--753.

\bibitem{nasr2020improving}
M.~Nasr, R.~Shokri, and A.~houmansadr, ``Improving deep learning with
  differential privacy using gradient encoding and denoising,'' 2020.

\bibitem{nguyen2022flame}
T.~D. Nguyen, P.~Rieger, H.~Chen, H.~Yalame, H.~Möllering, H.~Fereidooni,
  S.~Marchal, M.~Miettinen, A.~Mirhoseini, S.~Zeitouni, F.~Koushanfar, A.-R.
  Sadeghi, and T.~Schneider, ``Flame: Taming backdoors in federated learning,''
  in \emph{USENIX Security}, 2022.

\bibitem{paillier1999public}
P.~Paillier, ``Public-key cryptosystems based on composite degree residuosity
  classes,'' in \emph{EUROCRYPT}, 1999, pp. 223--238.

\bibitem{paszke2019pytorch}
A.~Paszke, S.~Gross, F.~Massa, A.~Lerer, J.~Bradbury, G.~Chanan, T.~Killeen,
  Z.~Lin, N.~Gimelshein, L.~Antiga \emph{et~al.}, ``Pytorch: An imperative
  style, high-performance deep learning library,'' \emph{NIPS}, pp. 8026--8037,
  2019.

\bibitem{rabin2005exchange}
M.~O. Rabin, ``How to exchange secrets with oblivious transfer.'' \emph{IACR
  Cryptol. ePrint Arch.}, 2005.

\bibitem{riazi2019xonn}
M.~S. Riazi, M.~Samragh, H.~Chen, K.~Laine, K.~Lauter, and F.~Koushanfar,
  ``$\{$XONN$\}$: Xnor-based oblivious deep neural network inference,'' in
  \emph{USENIX Security}, 2019, pp. 1501--1518.

\bibitem{cryptoeprint:2019:207}
D.~Rotaru and T.~Wood, ``Marbled circuits: Mixing arithmetic and boolean
  circuits with active security,'' 2019.

\bibitem{roy2022EIFFeL}
A.~Roy~Chowdhury, C.~Guo, S.~Jha, and L.~van~der Maaten, ``Eiffel: Ensuring
  integrity for federated learning,'' in \emph{CCS}, 2022, pp. 2535--2549.

\bibitem{shamir1979share}
A.~Shamir, ``How to share a secret,'' \emph{Communications of the ACM}, pp.
  612--613, 1979.

\bibitem{truex2019hybrid}
S.~Truex, N.~Baracaldo, A.~Anwar, T.~Steinke, H.~Ludwig, R.~Zhang, and Y.~Zhou,
  ``A hybrid approach to privacy-preserving federated learning,'' in
  \emph{AISec}, 2019, pp. 1--11.

\bibitem{10.5555/3495724.3497072}
H.~Wang, K.~Sreenivasan, S.~Rajput, H.~Vishwakarma, S.~Agarwal, J.-y. Sohn,
  K.~Lee, and D.~Papailiopoulos, ``Attack of the tails: Yes, you really can
  backdoor federated learning,'' in \emph{NIPS}, 2020.

\bibitem{DBLP:conf/nips/WangSRVASLP20}
H.~Wang, K.~Sreenivasan, S.~Rajput, H.~Vishwakarma, S.~Agarwal, J.~Sohn,
  K.~Lee, and D.~S. Papailiopoulos, ``Attack of the tails: Yes, you really can
  backdoor federated learning,'' in \emph{NIPS}, 2020.

\bibitem{wu2020value}
N.~Wu, F.~Farokhi, D.~Smith, and M.~A. Kaafar, ``The value of collaboration in
  convex machine learning with differential privacy,'' in \emph{IEEE S\&P},
  2020, pp. 304--317.

\bibitem{xiao2017fashion}
H.~Xiao, K.~Rasul, and R.~Vollgraf, ``Fashion-mnist: a novel image dataset for
  benchmarking machine learning algorithms,'' \emph{arXiv preprint
  arXiv:1708.07747}, 2017.

\bibitem{Xie2020DBA}
C.~Xie, K.~Huang, P.-Y. Chen, and B.~Li, ``Dba: Distributed backdoor attacks
  against federated learning,'' in \emph{ICLR}, 2020.

\bibitem{xie2020zeno++}
C.~Xie, S.~Koyejo, and I.~Gupta, ``Zeno++: Robust fully asynchronous sgd,'' in
  \emph{ICML}, 2020, pp. 10\,495--10\,503.

\bibitem{pmlr-v80-yin18a}
D.~Yin, Y.~Chen, R.~Kannan, and P.~Bartlett, ``{B}yzantine-robust distributed
  learning: Towards optimal statistical rates,'' in \emph{ICML}, 2018, pp.
  5650--5659.

\bibitem{aono2017privacy}
L.~Zhu, Z.~Liu, and S.~Han, ``Deep leakage from gradients,'' \emph{NIPS}, 2019.

\end{thebibliography}
